\documentclass[preprint,11pt]{elsarticle}
\journal{}
\usepackage{graphicx}
\usepackage{verbatim}
\usepackage{pifont,latexsym,ifthen,amsthm,rotating,calc,textcase,booktabs}
\usepackage{amsfonts,amssymb,amsbsy,amsmath,extarrows,bm}
\allowdisplaybreaks
\newtheorem{theorem}{Theorem}[section]
\newtheorem{lemma}[theorem]{Lemma}
\newtheorem{corollary}[theorem]{Corollary}
\newtheorem{remark}[theorem]{Remark}
\newtheorem{proposition}[theorem]{Proposition}

\newtheorem{Dbar}{$\bar{\partial}$-Problem}[section]
\numberwithin{equation}{section}
\usepackage{natbib}
\newtheorem{RHP}{RH problem}[section]
\usepackage{todonotes}
\usepackage{float}
\usepackage{subfigure}
\usepackage{verbatim}
\usepackage{hyperref}
\hypersetup{
    colorlinks=true, 
    linktoc=all,     
    linkcolor=blue,  
    citecolor=blue
}

\DeclareMathOperator*{\res}{Res}

\DeclareMathOperator*{\re}{Re}
\begin{document}



\begin{frontmatter}
\title{Soliton resolution and asymptotic stability of $N$-solitons to the Degasperis-Procesi equation on the line }


\author[inst2]{Xuan Zhou}

\author[inst3]{Zhaoyu Wang}

\author[inst3]{Engui Fan$^{*,}$  }

\address[inst2]{ College of Mathematics and Systems Science, Shandong University of Science and Technology, Qingdao, 266590,  P. R. China}
\address[inst3]{ School of Mathematical Sciences and Key Laboratory of Mathematics for Nonlinear Science, Fudan University, Shanghai,200433, P. R. China\\
* Corresponding author and e-mail address: faneg@fudan.edu.cn  }





\begin{abstract}

  The Degasperis-Procesi (DP) equation
  \begin{align}
    &u_t-u_{txx}+3\kappa  u_x+4uu_x=3u_x u_{xx}+uu_{xxx},   \nonumber
  \end{align}
  serving as a model delineating the propagation of shallow water waves, stands as a completely integrable system and admits a $3\times3$  matrix
  Lax pair. In this manuscript, we study the soliton resolution and large time behavior of solutions to the Cauchy problem of the DP equation with generic initial data in Schwarz space. Employing the $\bar \partial$-generalization of the Deift-Zhou nonlinear steepest descent method, we deduce different long time asymptotic expansions of
  the solution $u(x,t)$ in two distinct types of space-time  regions. This  result verifies the soliton resolution conjecture and asymptotic stability of $N$-soliton solutions for the DP equation.
\end{abstract}

\begin{keyword}
Degasperis-Procesi equation \sep Riemann-Hilbert problem \sep $\bar{\partial}$-steepest descent method \sep
 Soliton resolution, Asymptotic stability.\vspace{2mm}

  \end{keyword}
\end{frontmatter}
\tableofcontents

\section{Introduction}
The present paper is concerned with   soliton resolution and asymptotic stability of $N$-soliton solutions to  the  Degasperis-Procesi (DP) equation  on the line
\begin{align}
    &u_t-u_{txx}+3\kappa u_x+4uu_x=3u_x u_{xx}+uu_{xxx}, \label{DP}\\
    &u(x,0)=u_{0}(x),\label{intva}
\end{align}
where the initial data   $ u_0( x)$ is the function in  the  Schwarz space $\mathcal{S}(\mathbb{R})$, and  $\kappa$ is a positive constant  characterizing  the effect of the linear dispersion.
The DP equation was first discovered in a study for asymptotically integrable partial differential equations  \cite{AM}.
Afterward it was found that the DP equation  arises for modeling the propagation of shallow water waves over a flat bed in so-called  moderate amplitude regime  \cite{RS,RI1, AD}.  This regime can be characterized as capturing stronger nonlinear effects than dispersive, which, particularly, accommodate wave breaking phenomena.
This is in contrast with the so-called  shallow water regime, where various
integrable systems like  KdV  equation arise by balancing nonlinearity and dispersion \cite{BD}.

Among the models of  moderate amplitude regime,  the Camassa-Holm (CH) equation and the DP equation are two are integrable equation
 admitting bi-Hamiltonian structure and a Lax pair representation   \cite{AD,RI2}.   The CH and DP equations  correspond to $b=2$ and $b=3$ respectively to the   $b$-family   equation
\begin{equation*}
u_t-u_{txx}+b\kappa u_x+(b+1)uu_x=bu_x u_{xx}+uu_{xxx}.
\end{equation*}
Despite many similarities between DP and CH equation, we would like to point out that these two equations are truly different. First, the  DP equation   not only  admits  peakon solitons, but also shock peakons
 \cite{GM,HH,H};  Second, the  DP equation has entirely different form of  conservation laws with the CH equation \cite{RD,ADD}; Third, the CH equation is a re-expression of geodesic flow on the diffeomorphism group or on the Bott-Virasoro group, while no geometric derivation of the DP equation \cite{AB,G}.
  Last, compared with the  CH  equation associated with a  $2\times 2$ matrix spectral \cite{RD},
        the spectral analysis of the corresponding Lax pairs is quite different and becomes very difficult
        due to the fact that the spectral problem of the DP equation is a $3\times 3$ matrix-valued form  \cite{ADD}.
These  differences above result in some essential additional technique difficulties   in  the   analysis of the  inverse scattering transform,  the  RH    problem or long time asymptotics for the DP equation.

In recent years, there has been  attracted a lot of attention on  the DP equation \eqref{DP}
 due to its integrable structure and pretty mathematical
 properties \cite{J1,Y1}.   Liu and Yin proved the global existence and blow-up phenomena for the DP equation  \cite{YZ}.
   Constantin,  Ivanov and Lenells developed  the inverse scattering transform   for the DP equation and the implementation of the dressing method \cite{ARIJ}.
   Lenells proved  that the solution of  the initial-boundary value problem for the DP equation on the half-line
    can be expressed in term of  the solution of a RH problem  \cite{J2}.
Later Boutet de Monvel,  Lenells and Shepelsky  successfully extended the Deift-Zhou  nonlinear steepest descent method \cite{RN6,RN9}
to the  the   DP equation  on the half-line and  obtained  the long time asymptotics  in the similarity region
\cite{AJD}.   Hou, Zhao, Fan and Qiao  constructed the algbro-geometric solutions  for the DP hierarchy
\cite{HZF}.    Feola,   Giuliani and    Procesi developed the KAM theory close to an elliptic fixed point for quasi-linear Hamiltonian perturbations of the
DP  equation on the circle    \cite{FGP}.  Boutet de Monvel  and Shepelsky   developed  the RH  method to   the Cauchy problem of the   DP equation (\ref{DP}), and especially    in   the region $\xi\in(0,3)$ without consideration of solitons they further  obtained  the following   long time asymptotics  \cite{AD2}.
  Despite this progress,   the  long time  asymptotics in other regions  or  the long time  asymptotics in  the case when solitons appear  still have   been   unknown and   remain  to be studied in detail.

In this work,  for  the  Schwarz initial data
 $u_0(x)$  that support    solitons,  we   give  a full and rigorous  description of  the long-time asymptotics for the DP equation in the whole   upper half-plane
   $(x,t)\in \mathbb{R}\times \mathbb{R}^+$.  Due to   the presence of solitons, the long-time asymptotics of solutions for  the DP equation
   are necessary more detailed than in the case without solitons \cite{AD2}.   A  key tool to prove our results
is  the $\bar\partial$-generalization of the nonlinear steepest descent method proposed by McLaughlin and Miller \cite{KTRPD1,KTRPD2}.
This method, which is more convenient  to deal with  the   presence of    discrete spectrum associated with  the initial data,
  has been successfully used to obtain  the long-time asymptotics and
 the soliton resolution conjecture for some integrable systems in recent years   \cite{DM,CJ,BJ,Liu3,CL,YF,CF,ZF}.
More importantly, our results are   different  from \cite{AD2} in two  aspects.  Firstly,
 we divide the upper half $ (x,t)$-plane into two distinct  regions denoted by
\begin{itemize}
\item   \uppercase\expandafter{\romannumeral1}. Solitonic region: $\xi:=\frac{x}{t} >3$ and $\xi<-\frac{3}{8}$, in which we prove the soliton resolution and  asymptotic stability of
    $N$  soliton solutions  for the DP  equation;

\item   \uppercase\expandafter{\romannumeral2}.  Solitonless region.

   \uppercase\expandafter{\romannumeral2}-1:  $-\frac{3}{8}<\xi< 0$.  Zakharov-Manakov  region  and  24 phase points.

     \uppercase\expandafter{\romannumeral2}-2:  $0\leq \xi<3$. Zakharov-Manakov  region  and   12 phase points.




\end{itemize}
as depicted in Figure \ref{xtregions}. Secondly,   the leading order approximation  to  the solution $u(x,t)$ of  the DP equation in the  region    \uppercase\expandafter{\romannumeral1} can be  expressed as  a sum of   single solitons with different velocities.
This  result  is  the  verification of  soliton resolution conjecture  and  the asymptotic stability of $N$-soliton solutions  for the DP equation (\ref{DP}).

\begin{figure}
	\begin{center}
		\begin{tikzpicture}[scale=0.85]
			\draw[blue!20, fill=blue!20] (0,0)--(4.5,0)--(4.5,1.5)--(0, 0);
			\draw[green!20, fill=green!20] (0,0)--(0,3.5)--(4.5,3.5)--(4.5,1.5);
             \draw[yellow!20, fill=yellow!20] (0,0)--(-1.313,3.5)--(0,3.5)--(0,0);
			\draw[blue!20, fill=blue!20] (0,0)--(-1.313,3.5)--(-4.5,3.5)--(-4.5,0);
			\draw [-latex] (-5.5,0)--(5.5,0);
			\draw [ -latex](0,0)--(0,4.5);
			\draw [red,thick](0,0)--(4.5,1.5);
			\draw [red,thick](0,0)--(-1.313,3.501);
			\node    at (0,-0.3)  {$0$};
			\node [right] at (5.6,0)  {$x$};
			\node [right]   at (0,4.6)  {$t$};
			\node  [above]  at (1.6,1.8) {\footnotesize II-2};
            \node  [above]  at (-0.4,2.3) {\footnotesize II-1};
			\node  [below]  at (2.9,0.7) {\footnotesize I};
            \node  [left]  at (-2.5,0.8) {\footnotesize I};
			\node  [right]  at (4.5,1.5) {\scriptsize $ \xi=3 $};
			\node  [above]  at (-1.313,3.5) {\scriptsize $  \xi=-\frac{3}{8} $};
		\end{tikzpicture}
	\end{center}
	\caption{\footnotesize  The asymptotic regions of the solution for  the DP equation. The  upper half $ (x,t)$-plane is divided into
two kinds of   asymptotic regions: I. Solitonic region;  II. Solitonless region which includes two Zakharov-Manakov regions II-1 and II-2.   }
\label{xtregions}
\end{figure}
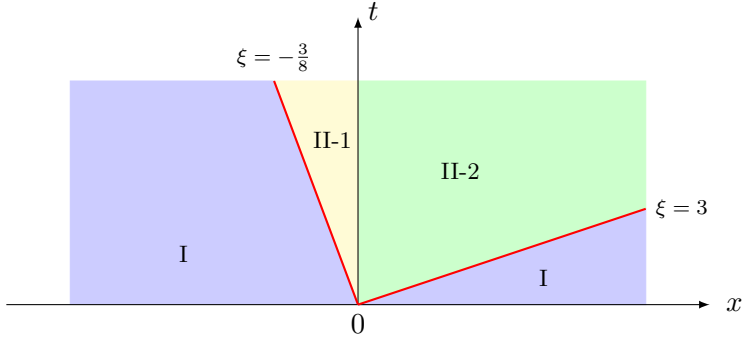

\subsection{Main results}

The main  results of the present paper is stated in the following theorem.
\begin{theorem}\label{th1}
Let $u(x,t)$ be the solution for the Cauchy problem \eqref{DP}--\eqref{intva}, where the initial data $u_{0}\in\mathcal{S}(\mathbb{R})$ is given along with scattering data  $\left\{r(k), \{\zeta_n, c_n\}_{n=1}^{N}\right\}$.

\begin{itemize}
\item  {\rm\bf For the solitonic  region I:} $\xi\in(-\infty,-\frac{3}{8})\cup(3,+\infty)$. Order the discrete spectrums $\zeta_n$ such that
\begin{equation}
{\rm Re} \ \zeta_1>{\rm Re} \ \zeta_2>\cdots>{\rm Re} \ \zeta_{N},
\end{equation}
then the solution $u(x,t)$ satisfies
\begin{equation}\label{uusol}
\Big\lvert u(x,t)-u^{sol,N}(x,t)\Big\rvert\leq Ct^{-1+\rho}, \quad t\rightarrow\infty,
\end{equation}
where $u^{sol,N}(x,t)$ is specified in \eqref{ressol}, representing the $N$-soliton solution characterized by the corresponding scattering data $ \{\tilde{r}(k)\equiv0,\{\zeta_n, \tilde{C}_n\}_{n=1}^{N}\}$ and
\begin{equation}
\tilde{C}_n=c_n\exp\left\{-\mathrm{i}\int_{I(\xi)}\nu(s)\left(\frac{1}{s-\omega^2\zeta_n}+\frac{1}{s-\omega\zeta_n}-\frac{2}{s-\zeta_n}\right)ds\right\}.
\end{equation}
Moreover, the solution $u(x,t)$  can be expressed as a sum of $N$ single  soliton solutions with different velocities
\begin{equation}\label{resol}
u(x,t)=\sum_{n=1}^{N} \mathcal{U}^{sol}(\zeta_n;x,t)+\mathcal{O}(t^{-1+\rho}), \quad t\rightarrow\infty,
\end{equation}
where $\mathcal{U}^{sol}(\zeta_n;x,t)$ denotes a single soliton solution as detailed in \eqref{1U}. The constant $\rho$, which appears in \eqref{uusol} and \eqref{resol}, is a free parameter that ranges within the interval $\left(0, \frac{1}{4}\right)$.
\item  {\rm\bf   For Zakharov-Manakov    regions  II-1: $\xi\in(-\frac{3}{8},0)$ and II-2:} $\xi\in\left[0,3\right)$,
\begin{align}\label{u3/8}
u(x,t)=t^{-\frac{1}{2}}f_1(e^{\frac{\pi}{6}\mathrm{i}};x,t)+\mathcal{O}(t^{-3/4}), \quad t\rightarrow\infty,
\end{align}
where
\begin{align*}
f_1(e^{\frac{\pi}{6}\mathrm{i}};x,t)=\sum_{j=1}^{3}\frac{\partial}{\partial t}\left(\left(H^{(0)}(e^{\frac{\pi}{6}\mathrm{i}})\right)_{j2}-\left(H^{(0)}(e^{\frac{\pi}{6}\mathrm{i}})\right)_{j1}\right),
\end{align*}
$H^{(0)}(e^{\frac{\pi}{6}\mathrm{i}})$ is defined by \eqref{H0}, $\left(H^{(0)}(e^{\frac{\pi}{6}\mathrm{i}})\right)_{jk}$ represents the element in the $j$-th row and $k$-th column.
\end{itemize}
\end{theorem}

The results above also confirm the soliton resolution conjecture and asymptotic stability of $N$  soliton solutions for
 the DP  equation respectively.

\subsection{Organization of the paper }\label{plan}

In Section \ref{2}, we get down to the inverse scattering transform. Based on the spectral analysis of Lax pair, a
basic vector RH problem for $m(k)$ associated with  the Cauchy problem  \eqref{DP}-(\ref{intva}) is established.
Moving on to Section \ref{3}, the RH problem for $m(k)$ is normalized  into a RH problem for $m^{(1)}(k)$,
which converts the residue conditions of poles away from the critical lines into the exponentially decay jumps on the contours around the poles.
We have implemented the vector RH technique, which effectively eradicates singularities present in the initial matrix-valued RH problem and significantly aids in our quest to recover the potential.
Section \ref{4} marks a pivotal stride in opening $\bar{\partial}$-lenses to construct a mixed $\bar{\partial}$-RH problem for $m^{(2)}(k)$.
This problem is skillfully decomposed into a pure RH problem for $M^{R}(k)$ and a pure $\bar{\partial}$-problem for $m^{(3)}(k)$ solved in Section \ref{7}. Our exploration reveals that the large-time contribution of $M^{R}(k)$ can be attributed to two key factors. There is a distinct contribution stemming from the discrete spectrum, where we solve a reflectionless RH problem $M^{r}(k)$
for the soliton components in Section \ref{5}. In Section \ref{6}, we encounter another contribution originating from jump contours, which is approximated by a local solvable model $M^{lo}(k)$ near phase points, while the residual error function is from a small norm RH problem outside phase points. Finally, we bring forth the comprehensive proof of Theorem \ref{th1} in Section \ref{8}.

\section{Inverse  scattering transform  and  RH Problem}\label{RHconstruct}\label{2}

 In this section, we  state some  basic
results on  the   inverse scattering transform and the RH problem  associated with the Cauchy problem
  \eqref{DP}-\eqref{intva}  as a preparatory work.   The details   can be found in
\cite{AD2}.

\subsection{The Lax pair and spectral analysis}

It is worth noting that without loss of generality, one can select $\kappa=1$  in the DP equation \eqref{DP}.
Then the DP equation \eqref{DP} admits the Lax pair
\begin{equation}\label{lax pair}
     \Phi_x=U\Phi, \quad \Phi_t=V\Phi,
\end{equation}
where
\begin{equation}
\begin{split}
& U=\left(\begin{array}{ccc} 0 & 1 & 0 \\ 0 & 0 & 1 \\ z^3q^3 & 1 & 0\end{array}\right), \quad
 V=\left(\begin{array}{ccc}  u_x-\frac{2}{3 } z^{-3} & -u & z^{-3}  \\ u+1 &  z^{-3} & -u \\ u_x-z^3uq^3 & 1 & -u_x+ z^{-3}\end{array}\right),\\
\end{split}
\end{equation}
and $q=(1+u-u_{xx})^{1/3}$. With the innovative scale defined as $y(x,t)=x-\int_{x}^{\infty}(q(\varsigma,t)-1)\mathrm{d}\varsigma$, we derive the Volterra  integral equation
\begin{equation}\label{INTM}
M(z)=I+\int_{ \pm \infty}^{x} \mathrm{e}^{Q(x, z)-Q(\varsigma, z)}\left(\hat{U} M(\varsigma, z)\right) \mathrm{e}^{-Q(x, z)+Q(\varsigma, z)} \mathrm{d}\varsigma,
\end{equation}
where
\begin{align}
&Q=y(x,t)\Lambda(z)+t H(z), \\
&\Lambda(z)=\operatorname{diag}\left\{\lambda_1(z), \lambda_2(z), \lambda_3(z)\right\}, \quad H(z)=\frac{1}{3z^3}I+\Lambda^{-1}(z),\\
&\hat{U}=P^{-1}(z)\left(\begin{array}{ccc} \frac{q_x}{q}&0&0 \\ 0&0&0 \\ 0&\frac{1}{q}-q&-\frac{q_x}{q} \end{array}\right)P(z), \quad P(z)= \left(\begin{array}{ccc} 1 & 1 & 1 \\\lambda_{1}(z) & \lambda_{2}(z) & \lambda_{3}(z) \\
\lambda_{1}^{2}(z) & \lambda_{2}^{2}(z) & \lambda_{3}^{2}(z) \end{array}\right).
\end{align}
Here $\lambda_j(z),\ j=1,2,3$ satisfy the algebraic equation
\begin{align}
\lambda^3(z)-\lambda(z)=z^3(k),
\end{align}
where
\begin{align}
z(k)=\frac{1}{\sqrt{3}} k\left(1+\frac{1}{k^{6}}\right)^{1/3}, \quad \lambda_{j}(k)=\frac{1}{\sqrt{3}}\left(\omega^{j} k+\frac{1}{\omega^{j} k}\right), \quad \omega=\mathrm{e}^{\frac{2\pi\mathrm{i}}{3}}.
\end{align}

We define
\begin{equation}
\Sigma=\mathop{\cup}\limits_{j=1}^6l_{j}, \quad l_{j}=\mathrm{e}^{\frac{(j-1)\pi\mathrm{i}}{3}}\mathbb{R}_{+},
\end{equation}
which divides  the $k$-plane into six sectors
\begin{equation}
D_{j}=\left\{k\mid\frac{\pi}{3}(j-1)<\arg k<\frac{\pi}{3}j\right\}, \quad j=1, \ldots, 6.
\end{equation}

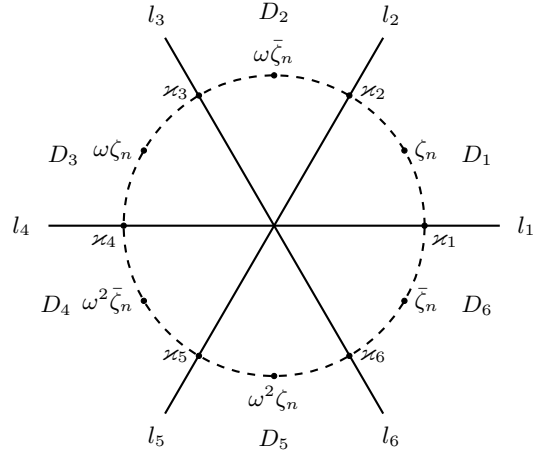
\begin{figure}[H]
\begin{center}
\begin{tikzpicture}[scale=0.5]
\draw [thick](-6,0)--(6,0);
\draw [thick](-2.9,-5)--(2.9,5);
\draw [thick](-2.9,5)--(2.9,-5);
\draw [thick, dashed] (0,0)circle(4cm);
\node  [right]  at (6.2,0) {\footnotesize$l_1$};
\node  [left]  at (-6.2,0) {\footnotesize$\l_4$};
\node  [right]  at (2.6,5.6) {\footnotesize$l_2$};
\node  [left]  at (-2.6,-5.6) {\footnotesize$l_5$};
\node  [left]  at (-2.6,5.6) {\footnotesize$l_3$};
\node  [right]  at (2.6,-5.6) {\footnotesize$\l_6$};
\filldraw[black] (4,0) circle [radius=0.07];
\filldraw[black] (2,3.464) circle [radius=0.07];
\filldraw[black] (-2,3.464) circle [radius=0.07];
\filldraw[black] (-4,0) circle [radius=0.07];
\filldraw[black] (-2,-3.464) circle [radius=0.07];
\filldraw[black] (2,-3.464) circle [radius=0.07];

\node  [right]  at (3.9,-0.4) {\footnotesize$\varkappa_1$};
\node  [right] at (2,3.55) {\footnotesize$\varkappa_2$};
\node  [left]  at (-2,3.55) {\footnotesize$\varkappa_3$};
\node  [left]  at (-3.9,-0.4) {\footnotesize$\varkappa_4$};
\node  [left]  at (-2,-3.464) {\footnotesize$\varkappa_5$};
\node  [right]  at (2,-3.464) {\footnotesize$\varkappa_6$};

\node  [below]  at (5.4,2.4) {\footnotesize$D_1$};
\node  [below]  at (0,6.2) {\footnotesize$D_2$};
\node  [below]  at (-5.6,2.4) {\footnotesize$D_3$};
\node  [below]  at (-5.8,-1.6) {\footnotesize$D_4$};
\node  [above]  at (0,-6.2) {\footnotesize$D_5$};
\node  [below]  at (5.4,-1.6) {\footnotesize$D_6$};

\filldraw[black] (3.464,2) circle [radius=0.07];
\filldraw[black] (0,4) circle [radius=0.07];
\filldraw[black] (-3.464,2) circle [radius=0.07];
\filldraw[black] (-3.464,-2) circle [radius=0.07];
\filldraw[black] (0,-4) circle [radius=0.07];
\filldraw[black] (3.464,-2) circle [radius=0.07];

\node  [right]  at (3.464,2) {\footnotesize$\zeta_n$};
\node  [above]  at (0,4) {\footnotesize$\omega\bar{\zeta}_n$};
\node  [left]  at (-3.464,2) {\footnotesize$\omega\zeta_n$};
\node  [left]  at (-3.464,-2) {\footnotesize$\omega^2\bar{\zeta}_n$};
\node  [below]  at (0,-4)       {\footnotesize$\omega^2\zeta_n$};
\node  [right]  at (3.464,-2) {\footnotesize$\bar{\zeta}_n$};

\end {tikzpicture}
\end{center}
\caption{The critical rays $l_j$,  analytical domains $D_j$, singularity  points $\varkappa_j, j=1,\dots,6$ and discrete spectrums $\zeta_n$, $ \omega^{\ell}\zeta_n, \omega^{\ell}\bar{\zeta}_n, n=1,\dots,N,\ \ell=0,1,2$ in the $k$-plane.}
\label{analregion&spectrumsdis}
\end{figure}


\begin{proposition}
 $M(k):=\Big( M_{ij}(k;x,t)\Big)_{3\times 3}$ satisfies symmetry relations:
\begin{align}\label{S1}
M(k)&=\Gamma_1\overline{M(\bar{k})}\Gamma_1=\Gamma_2\overline{M(\omega^2\bar{k})}\Gamma_2=\Gamma_3\overline{M(\omega\bar{k})}
\Gamma_3\nonumber\\
&=\Gamma_4M(\omega k)\Gamma_4^{-1}=\overline{M(\bar{k}^{-1})},
\end{align}
where
\begin{equation}\label{Gamma}
    \Gamma_1=\left(\begin{array}{ccc} 0 & 1 & 0 \\ 1 & 0 & 0 \\ 0 & 0 & 1\end{array}\right),\
    \Gamma_2=\left(\begin{array}{ccc} 0 & 0 & 1 \\ 0 & 1 & 0 \\ 1 & 0 & 0\end{array}\right), \
    \Gamma_3=\left(\begin{array}{ccc} 1 & 0 & 0 \\ 0 & 0 & 1 \\ 0 & 1 & 0 \end{array}\right), \
    \Gamma_4=\left(\begin{array}{ccc} 0 & 0 & 1 \\ 1 & 0 & 0 \\ 0 & 1 & 0 \\ \end{array}\right).
\end{equation}
\end{proposition}
The limiting values of $M(k)$ satisfies the jump relation
\begin{equation}\label{jump0}
M_+(k)=M_-(k)e^{Q}V_0(k)e^{-Q}, \quad k\in\Sigma,
\end{equation}
where $V_0(k)$ is determined by the initial value $u_0$. Take $k\in\mathbb{R}^{\pm}$ as an example, $V_0(k)$ has  a special matrix structure
\begin{equation*}
V_0(k)=\left(\begin{array}{ccc} 1 & \bar{r}_{\pm}(k) & 0 \\ -r_{\pm}(k) & 1-|r_{\pm}(k)|^2 & 0 \\ 0 & 0 & 1\end{array}\right),
\end{equation*}
$r_{\pm}(k)\in L^{\infty}(\mathbb{R})$ are  scalar functions with $r(k)=\mathcal{O}(k^{-1})$ as $k\to\infty$, which together with the symmetry $r_{\pm}(k)=\overline{r_{\pm}(\bar{k}^{-1})}$ leads to $\lim\limits_{k\to 0}r_{\pm}(k)=0$, and then $r_{\pm}(k)\in L^{2}(\mathbb{R}^{\pm})$. Naturally, we define the reflection coefficient as
\begin{align*}
	r(k)=\left\{ \begin{array}{ll}
		r_\pm(k),   & k\in \mathbb{R}^\pm,\\
		0  , & k=0.
	\end{array}\right.
\end{align*}
It follows that $r(k)\in L^{\infty}(\mathbb{R})\cap L^{2}(\mathbb{R})$. Moreover, $r(\pm1)=0$.

According to \cite{AD2}, the eigenvalues of spectral functions are limited on the unit circle $\{k\in \mathbb{C}: |k|=1\}$, which
are simple and finite. Let us denote two distinct categories of eigenvalues as $k_j \ (j=1,\cdots,N_1)$, $  k_{\ell}^A \ (\ell=1,\cdots,N_1^A)$ within the domain $D_1$, where $N=N_1+N_1^A$. To further elucidate this concept, we introduce the subsequent indicator sets
\begin{equation}\label{index set}
\widetilde{\mathcal{N}}=\{1,\cdot\cdot\cdot,N_1\}, \quad \widetilde{\mathcal{N}}^A=\{N_1+1,\cdot\cdot\cdot,N\}, \quad \mathcal{N}=\widetilde{\mathcal{N}}\cup\widetilde{\mathcal{N}}^A.
\end{equation}
For convenience, we rewrite the $\zeta_n$ instead of $k_n$ to represent the discrete spectrum with
\begin{align}
   \zeta_n=\left\{
   \begin{aligned}
   & k_n,   \quad n\in\widetilde{\mathcal{N}},\\
   &k_{n-N_1}^A, \quad n\in\widetilde{\mathcal{N}}^A.
   \end{aligned}
        \right.
\end{align}
Assisted by the norming constant $c_n$ for $\zeta_n, \ n\in\mathcal{N}$ , the residue conditions can be depicted in two distinctive matrix forms
\begin{align}\label{221}
  \underset{k=\zeta_n}{\rm Res}M(k)=\left\{
   \begin{aligned}
   & \lim_{k\rightarrow \zeta_n}M(k)e^Q\left(\begin{array}{ccc} 0 & -c_n& 0\\ 0 & 0 & 0 \\ 0 & 0 & 0 \end{array}\right)e^{-Q}, \quad n\in\widetilde{\mathcal{N}},\\
   & \lim_{k\rightarrow \zeta_n}M(k)e^Q\left(\begin{array}{ccc} 0 & 0 & 0\\ 0 & 0 & -c_n \\ 0 & 0 & 0 \end{array}\right)e^{-Q}, \quad n\in\widetilde{\mathcal{N}}^A.
   \end{aligned}
        \right.
\end{align}

Consider the discrete spectrums in other domains $D_j, \ j=2,\dots,6$. Based on the symmetries \eqref{S1}, it is apparent that for every $\zeta_n$, $n\in\mathcal{N}$, the eigenvalues of spectral functions also include $\omega\bar{\zeta}_n, \ \omega\zeta_n, \ \omega^2\bar{\zeta}_n, \ \omega^2\zeta_n, \ \bar{\zeta}_n$. Likewise, we can represent
\begin{align*}
&\zeta_{n+N}=\omega\bar{\zeta}_n, \quad \zeta_{n+2N}=\omega\zeta_n, \\
&\zeta_{n+3N}=\omega^2\bar{\zeta}_n, \quad \zeta_{n+4N}=\omega^2\zeta_n, \quad \zeta_{n+5N}=\bar{\zeta}_n,
\end{align*}
and
\begin{equation}
\begin{aligned}
&C_n=c_n, \quad C_{n+N}=\omega\bar{c}_n, \quad C_{n+2N}=\omega c_n, \\
& C_{n+3N}=\omega^2\bar{c}_n, \quad C_{n+4N}=\omega^2c_n, \quad C_{n+5N}=\bar{c}_n.
\end{aligned}
\end{equation}
In summary, the discrete spectrum can be characterized as
\begin{equation}
\mathcal{K}=\{\zeta_n\}_{n=1}^{6N},
\end{equation}
whose distribution on the $k$-plane is shown in Figure \ref{analregion&spectrumsdis}.

\subsection{Set up of RH problem}

To establish the basic RH problem, we consider the phase function
\begin{align}\label{oritheta12}
&\theta_{12}(k)=\left(k-\frac{1}{k}\right)\left(\hat{\xi}-\frac{3}{k^{2}+k^{-2}-1}\right), \quad \hat{\xi}:=\frac{y}{t},\\
&\theta_{13}(k)=-\theta_{12}(\omega^2k), \quad \theta_{23}(k)=\theta_{12}(\omega k).
\end{align}
Consequently, the ensuing RH problem emerges.
\begin{RHP}\label{rhpM}
Find a $3\times 3$ matrix-valued function $M(k):= M(k;y,t)$ such that
\begin{itemize}
\item Analyticity: $M(k)$ is meromorphic in $\mathbb{C}\backslash\Sigma$.
\item Jump relation: $M_{+}(k)=M_{-}(k)V(k)$, $k\in\mathbb{R}\cup\omega\mathbb{R}\cup\omega^2\mathbb{R}$, where
\begin{align}
   V(k)=\left\{
   \begin{aligned}
    &\left(\begin{array}{ccc} 1 & \bar{r}(k) e^{\mathrm{i}t\theta_{12}(k)} & 0 \\ -r(k)e^{-\mathrm{i}t\theta_{12}(k)} & 1-|r(k)|^2 & 0 \\ 0 & 0 & 1 \end{array}\right):= V^0(k), \quad k\in\mathbb{R},\\
    &\Gamma_4^2V^0(\omega^2k)\Gamma_4^{-2}, \quad k\in\omega\mathbb{R},\\
    &\Gamma_4V^0(\omega k)\Gamma_4^{-1}, \quad k\in\omega^2 \mathbb{R}.\\
   \end{aligned}
        \right.
\end{align}
\item Asymptotic behaviors: $M(k)=I+\mathcal{O}(k^{-1}), \quad  k\rightarrow\infty$.
\item Singularities: $M(k)$ has singularity at $\varkappa_{\nu}, \ \nu=1, \ldots, 6$ with
\begin{align}
   M(k)=\left\{
   \begin{aligned}
    &M_{\pm 1}(k)+\mathcal{O}(1), \quad k\rightarrow\pm1,\\
    &\Gamma_4^2M_{\pm 1}(\omega^2k)\Gamma_4^{-2}+\mathcal{O}(1), \quad k\rightarrow\pm\omega,\\
    &\Gamma_4M_{\pm 1}(\omega k)\Gamma_4^{-1}+\mathcal{O}(1), \quad k\rightarrow\pm\omega^2,\\
   \end{aligned}
        \right.
\end{align}
where
\begin{align*}
&M_{\pm 1}(k)=\frac{1}{k\mp 1}\left(\begin{array}{ccc} \alpha_{\pm} & \alpha_{\pm}  & \beta_{\pm} \\ -\alpha_{\pm} & -\alpha_{\pm}  & -\beta_{\pm} \\ 0 & 0 & 0 \end{array}\right),\\
&\alpha_{\pm}=-\bar{\alpha}_{\pm}, \ \beta_{\pm}=-\bar{\beta}_{\pm}.
\end{align*}
\item Residue conditions: For $\zeta_n\in D_1\cap\mathcal{K}$,
\begin{equation}\label{resM11}
\begin{aligned}
&\res\limits_{k=\zeta_n} M(k)=\lim_{k\rightarrow \zeta_n}M(k)B_n,\\
&\res\limits_{k=\omega \bar{\zeta}_n}M(k)=\omega\lim_{k\rightarrow \omega\bar{\zeta}_n}M(k)\Gamma_3 \bar{B}_n\Gamma_3:=\lim_{k\rightarrow \omega\bar{\zeta}_n}M(k)B_{n+N},\\
&\res\limits_{k=\omega \zeta_n}M(k)=\omega\lim_{k\rightarrow \omega \zeta_n}M(k)\Gamma_4^2 B_n\Gamma_4^{-2}:=\lim_{k\rightarrow \omega \zeta_n}M(k)B_{n+2N},\\
&\res\limits_{k=\omega^2\bar{\zeta}_n}M(k)=\omega^2 \lim_{k\rightarrow\omega^2\bar{\zeta}_n}M(k)\Gamma_2\bar{B}_n\Gamma_2:=\lim_{k\rightarrow \omega^2\bar{\zeta}_n}M(k)B_{n+3N},\\
&\res\limits_{k=\omega^2\zeta_n}M(k)=\omega^2\lim_{k\rightarrow\omega^2\zeta_n}M(k)\Gamma_4 B_n\Gamma_4^{-1}:=\lim_{k\rightarrow \omega^2\zeta_n}M(k)B_{n+4N},\\
&\res\limits_{k=\bar{\zeta}_n} M(k)=\lim_{k\rightarrow \bar{\zeta}_n}M(k)\Gamma_1\bar{B}_n\Gamma_1:=\lim_{k\rightarrow \bar{\zeta}_n}M(k)B_{n+5N},
\end{aligned}
\end{equation}
where
\begin{equation}\label{res22}
   B_n=\left\{
   \begin{aligned}
    &\left(\begin{array}{ccc} 0 & -c_ne^{\mathrm{i}t\theta_{12}(\zeta_n)} & 0 \\ 0 & 0 & 0 \\ 0 & 0 & 0 \end{array}\right), \quad n\in\widetilde{\mathcal{N}},\\
    &\left(\begin{array}{ccc} 0 & 0 & 0 \\ 0 & 0 & -c_ne^{\mathrm{i}t\theta_{23}(\zeta_n)} \\ 0 & 0 & 0 \end{array}\right), \quad n\in\widetilde{\mathcal{N}}^A,\\
   \end{aligned}
        \right.
\end{equation}
\end{itemize}
\end{RHP}

\subsection{RH characterization of the solution for the DP equation}

Introducing a vector-valued function
\begin{equation}\label{mtov}
m(k):=(m_1(k), \ m_2(k),\  m_3(k))=(1, \ 1, \ 1)M(k),
\end{equation}
which can be viewed as an elegant shift from the $3\times 3$ matrix RH problem to the $1\times 3$ vector RH problem.
This transition serves to remove the singularities at $\varkappa_j, j=1,\dots,6$, thereby ushering in the subsequent vector RH problem.
\begin{RHP}\label{rhpvm}
Find a row vector-valued function $m(k):= m(k;y,t)$ such that
\begin{itemize}
\item  $m(k)$ is meromorphic in $\mathbb{C}\backslash\Sigma$.
\item Jump relation: $m_{+}(k)=m_{-}(k)V(k)$, $k\in\mathbb{R}\cup\omega\mathbb{R}\cup\omega^2\mathbb{R}$.
\item Asymptotic behaviors: $m(k)=\begin{pmatrix}1,& 1,& 1\end{pmatrix}+\mathcal{O}(k^{-1}), \quad  k\rightarrow\infty.$
\item $m(k)$ has the same form of residue conditions as $M(k)$ in RH problem \ref{rhpM}.
\end{itemize}
\end{RHP}
The solution  for the Cauchy problem \eqref{DP}--\eqref{intva} can be expressed in the following parametric form
\begin{equation}\label{rescon}
\begin{aligned}
&u(y,t)=\frac{\partial}{\partial t}\log{\frac{m_2}{m_1}}(e^{\frac{\pi}{6}\mathrm{i}};y,t),\\
&x(y,t)=y+\log{\frac{m_2}{m_1}}(e^{\frac{\pi}{6}\mathrm{i}};y,t).
\end{aligned}
\end{equation}

\section{Normalization of the RH problem}\label{3}

\subsection{Phase points and signature table}

To deal with the oscillatory components $e^{\pm\mathrm{i}t\theta_{ij}(k)}$ in RH problem \ref{rhpM}, we consider the imaginary part of $\theta_{12}(k)$
\begin{equation}\label{Imthe12}
\begin{aligned}
&{\rm Im}\theta_{12}(k)=\hat{\xi}(1+|k|^{-2})\operatorname{Im}k+3\operatorname{Im}k\\
&\times \frac{|k|^6+2|k|^4+2|k|^2-4(1+|k|^2)\operatorname{Re}^2k+1}{|k|^8-2(1+|k|^4)(\operatorname{Re}^2k-\operatorname{Im}^2k)
+3(\operatorname{Re}^2k-\operatorname{Im}^2k)^2-4\operatorname{Re}^2k\operatorname{Im}^2k+1}.
\end{aligned}
\end{equation}
The signature table of ${\rm Im}\theta_{12}(k)$ is presented in Figure \ref{figtheta}. Further, we conclude the distribution of stationary phase points depends on $\hat{\xi}$ as follows
    \begin{itemize}
       \item  For $\hat{\xi}\in(-\infty,-\frac{3}{8})\cup(3,+\infty)$, there is no stationary phase point on $\mathbb{R}$;
       \item  For $\hat{\xi}\in(-\frac{3}{8},0)$, there are $8$ stationary phase points on $\mathbb{R}$, see Figure \ref{case1};
       \item For $\hat{\xi}\in\left[0,3\right)$, there are $4$ stationary phase points on $\mathbb{R}$, see Figure \ref{case2}.
    \end{itemize}
The number of phase points can be represented as
\begin{align}\label{nop}
{p(\hat{\xi})=\left\{
        \begin{aligned}
    &0, \quad \hat{\xi}\in(-\infty,-\frac{3}{8})\cup(3,+\infty),\\
    &8, \quad \hat{\xi}\in(-\frac{3}{8},0),\\
    &4, \quad \hat{\xi}\in\left[0,3\right).
    \end{aligned}
        \right.}
\end{align}

\begin{figure}[htbp]
     \centering
     \subfigure[$\hat{\xi}<-\frac{3}{8}$]{\label{figure0}
      \begin{minipage}[t]{0.32\linewidth}
       \centering
       \includegraphics[width=1.52in]{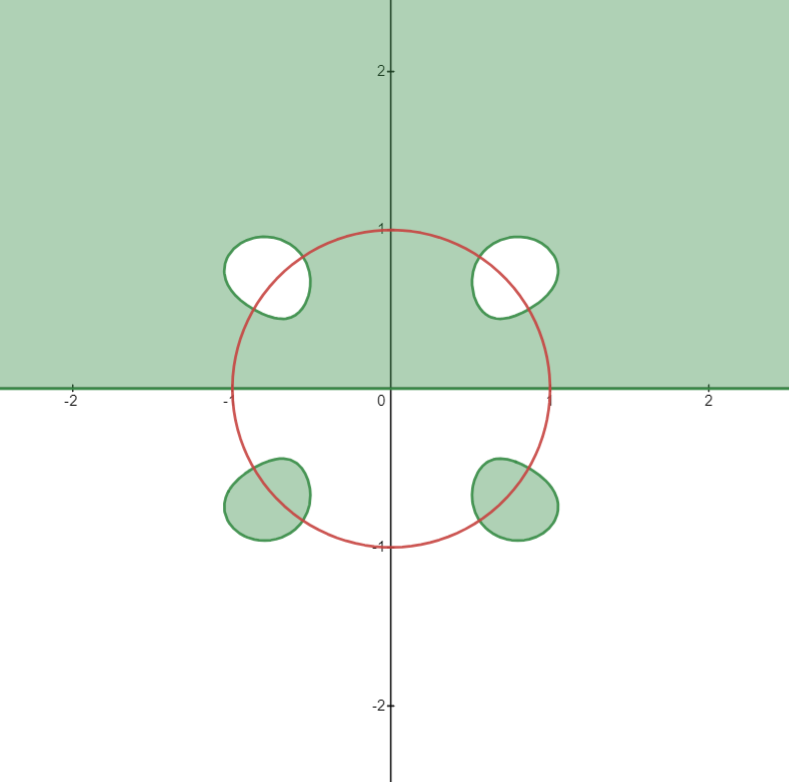}
      \end{minipage}
     }%
    \subfigure[$-\frac{3}{8}<\hat{\xi}<0$]{\label{figurec}
     \begin{minipage}[t]{0.32\linewidth}
      \centering
      \includegraphics[width=1.52in]{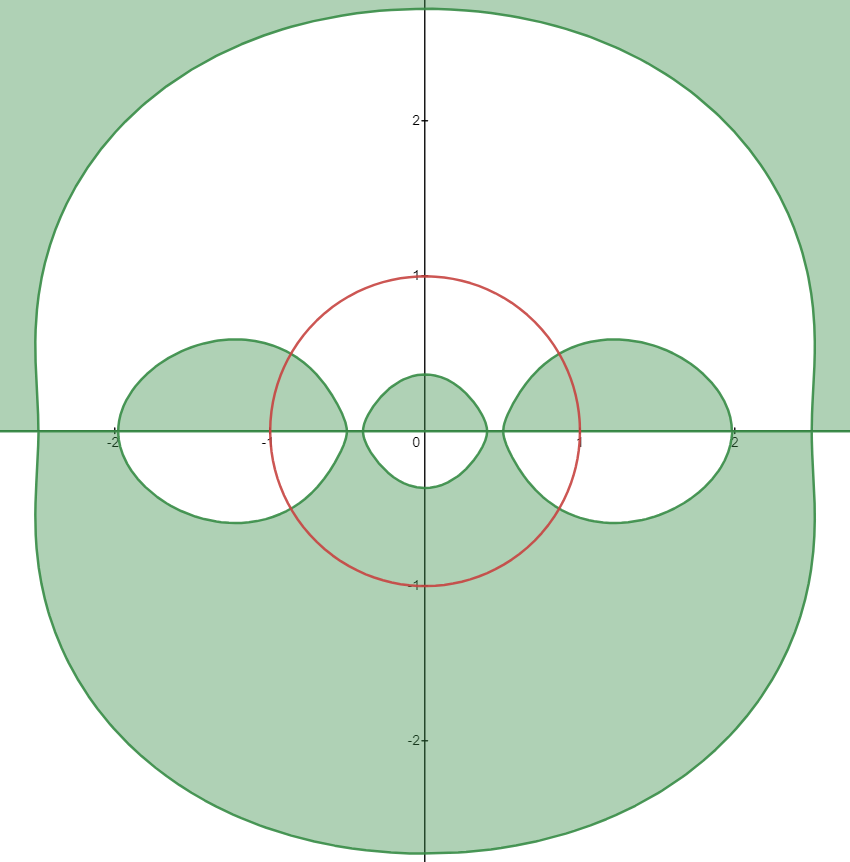}
     \end{minipage}
    }%

    \subfigure[$0\leq\hat{\xi}<3$]{\label{figured}
 \begin{minipage}[t]{0.32\linewidth}
  \centering
  \includegraphics[width=1.52in]{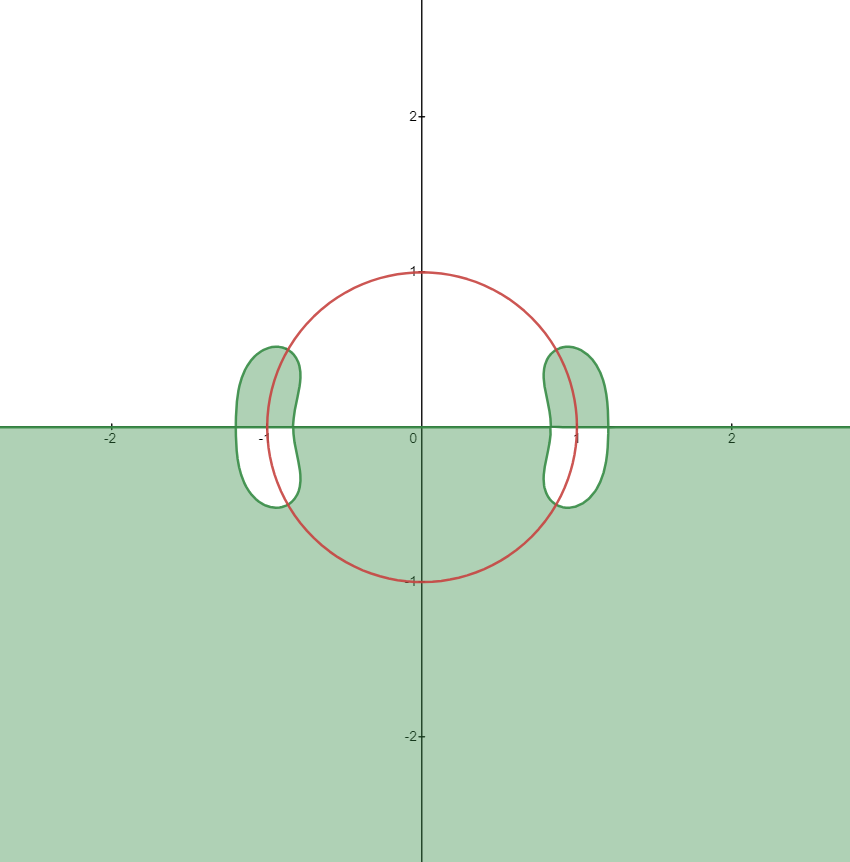}
 \end{minipage}
 }%
  \subfigure[$\hat{\xi}>3$]{\label{figuref}
      \begin{minipage}[t]{0.32\linewidth}
       \centering
       \includegraphics[width=1.52in]{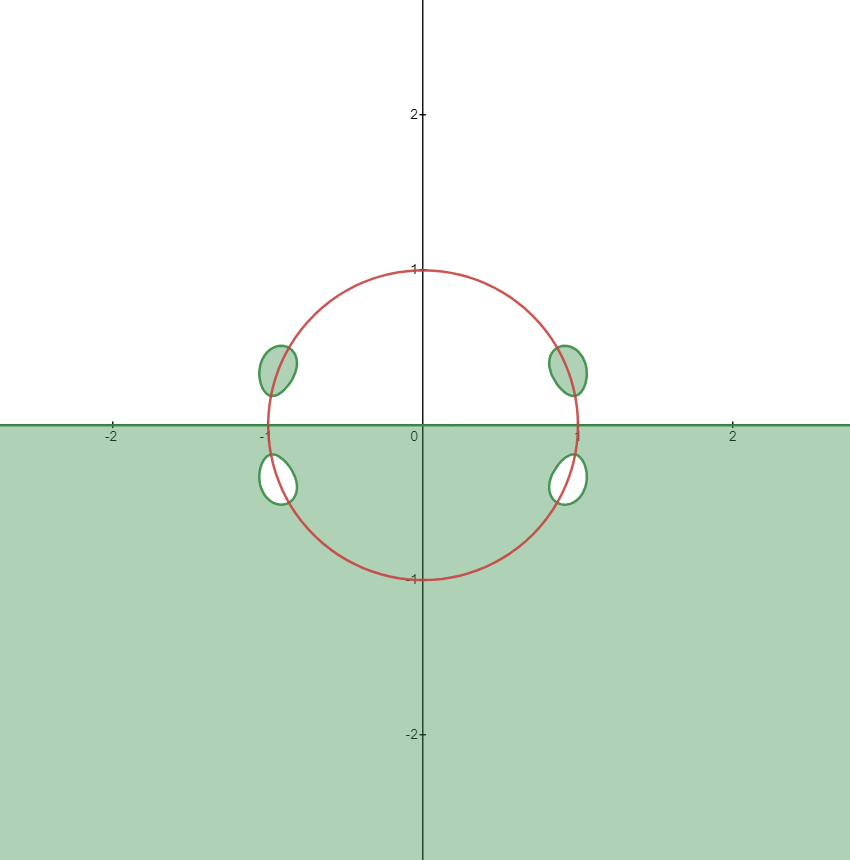}
      \end{minipage}
     }%
     \caption{\footnotesize Signature table of ${\rm Im}\theta_{12}(k)$ with different $\hat{\xi}=y/t$:
     $\textbf{(a)}$ $\hat{\xi}<-\frac{3}{8}$,
    $\textbf{(b)}$ $-\frac{3}{8}<\hat{\xi}<0$,
    $\textbf{(c)}$ $0\leq \hat{\xi}<3$,
    $\textbf{(d)}$ $\hat{\xi}>3$.
  ${\rm Im}\theta_{12}(k)<0$ in the green region and ${\rm Im}\theta_{12}(k)>0$ in the white region. Moreover, the red line is the unit circle.}
     \label{figtheta}
\end{figure}

\subsection{Conjugation}\label{conju}

The $\bar{\partial}$-nonlinear steepest descent method for oscillatory RH problems consists in a series of
transformations, in order to reach a model RH problem which can be explicitly solved. We
aim to construct a regular (holomorphic) RH problem as our basis for the
subsequent analysis, which requires the execution of two fundamental operations:
    \begin{itemize}
       \item[(i)] To investigate the asymptotics of different space-time regions, different factorizations for the
                  jump matrix $V(k)$ should be taken.
        \item[(ii)] The residue conditions for $m(k)$ should be converted into the jump conditions on the
                   auxiliary contours such that the jump matrices vanish rapidly to the identity matrix as
                   the parameter $t\rightarrow\infty$.
    \end{itemize}

The operation (i) is aided by two well known factorizations of the jump matrix
\begin{equation}\label{facjum1}
 V(k)=b(k)^{-\dag}b(k)=B(k)T_0(k)B(k)^{-\dag}, \quad k\in\mathbb{R},
\end{equation}
where
\begin{align*}
&b(k)^{-\dag}=\left(\begin{array}{ccc}  1 & 0 & 0 \\  -r(k)e^{-\mathrm{i}t\theta_{12}(k)} & 1 & 0 \\ 0 & 0 & 1 \end{array}\right), \quad b(k)=\left(\begin{array}{ccc}  1 & \bar{r}(k)e^{\mathrm{i}t\theta_{12}(k)} & 0 \\  0 & 1 & 0 \\ 0 & 0 &1 \end{array}\right),\\
&B(k)^{-\dag}=\left(\begin{array}{ccc}  1 & 0 & 0 \\ -\frac{r(k)}{1-|r(k)|^2}e^{-\mathrm{i}t\theta_{12}(k)} & 1 & 0 \\ 0 & 0 & 1 \end{array}\right), \quad B(k)=\left(\begin{array}{ccc}  1 & \frac{\bar{r}(k)}{1-|r(k)|^2}e^{\mathrm{i}t\theta_{12}(k)} & 0 \\  0 & 1 & 0 \\ 0 & 0 & 1 \end{array}\right),\\
&T_0(k)=\left(\begin{array}{ccc}  \frac{1}{1-|r(k)|^2} & 0 & 0 \\ 0 & 1-|r(k)|^2 & 0 \\ 0 & 0 & 1 \end{array}\right).
\end{align*}

To achieve operations (i) and (ii), we shall define the critical line $L(\hat{\xi}):=\frac{\sqrt{3}}{2}\sqrt{1+1/\hat{\xi}}.$ This will allow us to categorize the discrete spectrums $\zeta_n, \ n\in\mathcal{N}$ into six sets
\begin{align*}
&\Delta_{1}=\{j\in\widetilde{\mathcal{N}}: {\rm Re} \ {\zeta_j}<L(\hat{\xi})\}, \\
&\nabla_{1}=\{j\in\widetilde{\mathcal{N}}: {\rm Re} \ {\zeta_j}>L(\hat{\xi})\},\\
&\Delta_{2}=\{\ell\in\widetilde{\mathcal{N}}^A: {\rm Re} \ {\zeta_{\ell}}<L(\hat{\xi})\}, \\
&\nabla_{2}=\{\ell\in\widetilde{\mathcal{N}}^A: {\rm Re} \ {\zeta_{\ell}}>L(\hat{\xi})\},\\
&\Lambda_{1}=\{j\in\widetilde{\mathcal{N}}: \lvert{\rm Re} \ {\zeta_j}-L(\hat{\xi})\rvert<\delta_{0}\}, \\
& \Lambda_{2}=\{\ell\in\widetilde{\mathcal{N}}^A: \lvert{\rm Re} \ {\zeta_{\ell}}-L(\hat{\xi})\rvert<\delta_{0}\},
\end{align*}
where $\delta_{0}$ is a fixed small enough constant such that the sets $\{\lvert k-\zeta_n\rvert<\delta_0, \ n\in\mathcal{N}\}$ are pairwise disjoint. Furthermore, denote
\begin{equation}
\Delta=\Delta_1\cup\Delta_2, \quad \nabla=\nabla_1\cup\nabla_2, \quad \Lambda=\Lambda_1\cup\Lambda_2. \label{lambda}
\end{equation}

\begin{remark}
Noticing the fact that the discrete spectrum set $\{\zeta_n, \ n\in\mathcal{N}\}$ for $-\frac{3}{8}<\hat{\xi}<3$ is away from the critical line $L(\hat{\xi})$, it follows that $\Lambda=\varnothing$ in the case.
\end{remark}
Introducing the matrix function $T(k)=\text{diag}\{T_1(k),T_2(k),T_3(k)\}$, where
\begin{equation}\label{T}
\begin{aligned}
& T_i(k)=\frac{H(\omega^{i+1}k)}{H(\omega^{i+2}k)}, \quad i=1,2,3,\\
&H(k)=\underset{j\in\Delta_1}\prod\frac{k-\zeta_j}{k-\bar{\zeta}_j}\underset{\ell\in\Delta_2}\prod\frac{k-\omega\zeta_{\ell}}{k-\omega^2\bar{\zeta}_{\ell}}
\exp\left\{\frac{1}{2\pi\mathrm{i}}\int_{I(\hat{\xi})}\frac{\log{(1-|r(s)|^2)}}{s-k}ds\right\},
\end{aligned}
\end{equation}
$I(\hat{\xi})$ is an indicator function given by
\begin{align}
&I(\hat{\xi})=\left\{
        \begin{aligned}
    &\mathbb{R}, \quad \hat{\xi}\in(-\infty,-\frac{3}{8}),\\
    &(-\infty,  k_8)\cup\left(\mathop{\cup}\limits_{i=1}^3(k_{2i+1},k_{2i})\right)\cup(k_1,+\infty), \quad \hat{\xi}\in(-\frac{3}{8},0),\\
    &(k_4,k_3)\cup(k_2,k_1), \quad \hat{\xi}\in\left[0,3\right),\\
    &\emptyset, \quad \hat{\xi}\in(3,+\infty).
    \end{aligned}
        \right.
\end{align}

To implement the operation (ii), we shall introduce the following parameter
\begin{align}
	\varrho=\frac{1}{4}\min&\left\lbrace \min_{n\in \mathcal{N}} |\operatorname{Im}\zeta_n| ,\min_{n\in \mathcal{N},\arg k=\frac{\pi }{3}\mathrm{i}}|\zeta_n-k|,\min_{n\in
     \mathcal{N}\setminus\Lambda,\operatorname{Im}\theta_{12}(k)=0}|\zeta_n-k|,\right.\nonumber\\
	&\left. \min_{n\in\mathcal{N}}|\zeta_n-e^{\frac{\pi}{6}\mathrm{i}}|,\min_{n\neq m\in \mathcal{N}}|\zeta_n-\zeta_m|\right\rbrace .
\end{align}
Subsequently, the small disks $\scriptsize{\mathbb{D}_n:=\mathbb{D}(\zeta_n,\varrho)}$ are pairwise disjoint, also disjoint with critical lines and the contours $\Sigma$. It is crucial to emphasize that $\scriptsize{e^{\frac{\pi}{6}\mathrm{i}}\notin \mathbb{D}_n}$.

For $n=1,\cdots,6N$, let us define
\begin{align}\label{G}
G(k)=\left\{
        \begin{aligned}
    &I-\frac{B_n}{k-\zeta_n}, \quad k\in\mathbb{D}_n, \ n-k_0N\in\nabla, \ k_0\in\{0,\cdots,5\},\\
    &\begin{pmatrix} 1& 0& 0 \\-\frac{k-\zeta_n}{C_{n}e^{\mathrm{i}t\theta_{12}(\zeta_{n})}}& 1 & 0 \\ 0&	0 & 1 \end{pmatrix}, \quad k\in\mathbb{D}_n,n\in\Delta_1 \text{ or }n-2N\in\Delta_2,\\
    &\begin{pmatrix} 1&	0 & 0 \\ 0& 1 & 0\\ 0& -\frac{k-\zeta_n}{C_{n}e^{\mathrm{i}t\theta_{13}(\zeta_{n})}} & 1 \end{pmatrix}, \quad k\in\mathbb{D}_n,n-N\in\Delta_1 \text{ or }n-5N\in\Delta_2,\\
    &\begin{pmatrix} 1&	0 & -\frac{k-\zeta_n}{C_{n}e^{-\mathrm{i}t\theta_{13}(\zeta_{n})}} \\ 0& 1 & 0\\ 0&	0 & 1\end{pmatrix}, \quad k\in\mathbb{D}_n,n-2N\in\Delta_1 \text{ or }n-4N\in\Delta_2,\\
    &\begin{pmatrix} 1&	0 &  0\\ 0& 1 &  -\frac{k-\zeta_n}{C_{n}e^{-\mathrm{i}t\theta_{23}(\zeta_{n})}}\\ 0&	0 & 1\end{pmatrix}, \quad k\in\mathbb{D}_n,n-3N\in\Delta_1 \text{ or }n-N\in\Delta_2,\\
    &\begin{pmatrix} 1&	0 & 0 \\ 0& 1 & 0 \\ 0 &-\frac{k-\zeta_n}{C_{n}e^{\mathrm{i}t\theta_{23}(\zeta_{n})}} & 1 \end{pmatrix}, \quad k\in\mathbb{D}_n,n-4N\in\Delta_1 \text{ or }n\in\Delta_2,\\
    &\begin{pmatrix} 1&	-\frac{k-\zeta_n}{C_{n}e^{-\mathrm{i}t\theta_{12}(\zeta_{n})}} & 0 \\ 0& 1 & 0\\ 0&	0 & 1 \end{pmatrix}, \quad k\in\mathbb{D}_n,n-5N\in\Delta_1\text{ or } n-3N\in\Delta_2,\\
    &I, \quad elsewhere.
    \end{aligned}
        \right.
\end{align}

Consider the following contour
\begin{equation}
\begin{aligned}
&\Sigma^{(1)}=\Sigma\cup\Sigma^{(ci)}, \quad \Sigma^{(ci)}=\mathop{\cup}\limits_{n\in\breve{\Lambda}}\partial\mathbb{D}_n,\\
& \breve{\Lambda}=\big\{n: n-k_0N\in\mathcal{N}/\Lambda, \ k_0\in\{0,1,\cdots,5\}\big\}.
\end{aligned}
\end{equation}
We proceed with the transformation
\begin{equation}\label{defM1}
m^{(1)}(k)=m(k)G(k)T(k),
\end{equation}
which satisfies the following RH problem.

\begin{RHP}\label{rhp1}
Find a $1\times 3$  vector-valued function $m^{(1)}(k):=m^{(1)}(k;y,t)$ such that
\begin{itemize}
    \item $m^{(1)}(k)$ is meromorphic in $\mathbb{C}\backslash\Sigma^{(1)}$.
    \item $m^{(1)}(k)$ satisfies the jump relation
    \begin{equation}
    m_+^{(1)}(k)=m_-^{(1)}(k)V^{(1)}(k),
    \end{equation}
    where
   \begin{align}\label{rhp1jump}
    V^{(1)}(k)=
   & \left\{
        \begin{aligned}
    &\begin{pmatrix} 1& 0& 0 \\-\rho(k)(T_{12})_{+}e^{-\mathrm{i}t\theta_{12}}& 1 & 0 \\ 0&	0 & 1 \end{pmatrix}
    \begin{pmatrix} 1& \bar{\rho}(k)(T_{21})_{-}e^{\mathrm{i}t\theta_{12}}& 0 \\0 & 1 & 0 \\ 0&	0 & 1 \end{pmatrix}, \quad k\in\mathbb{R},\\
    &\begin{pmatrix} -\rho(\omega^2k)(T_{31})_{+}e^{\mathrm{i}t\theta_{13}}& 1& 0 \\0 & 0 & 1 \\ 1&	0 & 0 \end{pmatrix}
    \begin{pmatrix} \bar{\rho}(\omega^2k)(T_{13})_{-}e^{-\mathrm{i}t\theta_{13}}& 0& 1 \\1 & 0 & 0 \\ 0&	1 & 0 \end{pmatrix}, \quad k\in\omega\mathbb{R},\\
    &\begin{pmatrix} 0& 0& 1 \\1 & 0 & 0 \\ -\rho(\omega k)(T_{32})_{+}e^{-\mathrm{i}t\theta_{23}}&	1 & 0 \end{pmatrix}
    \begin{pmatrix} 0& 1& \bar{\rho}(\omega k)(T_{23})_{-}e^{\mathrm{i}t\theta_{23}} \\0 & 0 & 1 \\ 1&	0 & 0 \end{pmatrix}, \quad k\in\omega^2\mathbb{R},\\
    &T^{-1}(k)G(k)T(k), \quad k\in\partial\mathbb{D}_n\cap D_{2\nu-1}, \nu=1,2,3,\\
	&T^{-1}(k)G^{-1}(k)T(k), \quad k\in\partial\mathbb{D}_n\cap D_{2\nu}, \nu=1,2,3,
    \end{aligned}
        \right.
    \end{align}
    with
    \begin{align}
    \rho(k)=
   & \left\{
        \begin{aligned}
    &r(k), \quad k\in\mathbb{R}\backslash I(\hat{\xi}),\\
	&-\frac{r(k)}{1-|r(k)|^2}, \quad k\in I(\hat{\xi}).
    \end{aligned}
        \right.
    \end{align}
    \item $m^{(1)}(k)$ admits the asymptotic behavior
    \begin{equation}\label{rhp1aym}
    m^{(1)}(k)=(1, \ 1, \ 1)+\mathcal{O}(k^{-1}), \quad  k\rightarrow\infty.
    \end{equation}
    \item $m^{(1)}(k)$ has simple poles at $\zeta_n$  $n-k_0N\in\Lambda$ with the following residue condition
    \begin{align}
	\underset{k=\zeta_n}{\rm Res}m^{(1)}(k)=\lim_{k\to\zeta_n}m^{(1)}(k)\left(T^{-1}(k)B_nT(k)\right) .
	\end{align}
\end{itemize}
\end{RHP}

\section{Hybrid $\bar{\partial}$-RH problem }\label{4}

In this section, our objective is to remove the jump from the original jump contour $\mathbb{R}\cup\omega\mathbb{R}\cup\omega^2\mathbb{R}$, in such a way that the new problem takes advantage of the decay of $e^{\pm 2\mathrm{i}t\theta_{ij}(k)}$.

\subsection{opening $\bar{\partial}$-lenses in region $\xi\in(-\infty,-\frac{3}{8})\cup(3,+\infty)$}

Noticing that $\xi\sim \hat{\xi}, \ t\to \infty $ (refer to  Remark \ref{r8}),  we can equivalently  consider  the space-time regions $\hat{\xi}\in(-\infty,-\frac{3}{8})\cup(3,+\infty)$.

Fix an angle $\phi_0 (0<\phi_0<\frac{\pi}{6})$ sufficiently small such that the set $\left\{z\in\mathbb{C}:|\frac{\re z}{z}| > \cos \varphi \right\}$ does not intersect any of the disks $\mathbb{D}_n$.
Define jump contours and regions in Figure \ref{Sju}, let
\begin{align*}
\Sigma^{(ju)}=\mathop{\cup}\limits_{j=1}^6\Sigma_j, \quad \Omega=\mathop{\cup}\limits_{j=1}^{6}\Omega_{j}.
\end{align*}
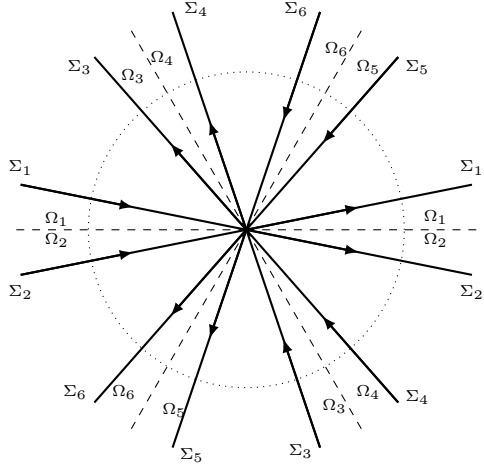
\begin{figure}[h]
\begin{center}
	\begin{tikzpicture}[scale=0.6]								
		\draw[dashed](0,0)--(5.1,0);
		\draw[dashed,rotate=60](0,0)--(5.1,0);
        \draw[dashed,rotate=120](0,0)--(5.1,0);
        \draw[dashed,rotate=180](0,0)--(5.1,0);
        \draw[dashed,rotate=240](0,0)--(5.1,0);
        \draw[dashed,rotate=300](0,0)--(5.1,0);

        \draw[thick](0,0)--(5,1);
        \node[above]at (5,1){\tiny$\Sigma_1$};
        \node[above]at (4.2,-0.1){\tiny$\Omega_1$};
		\draw[thick](0,0)--(-5,1);
        \node[above]at (-5,1) {\tiny$\Sigma_1$};
         \node[above]at (-4.2,-0.15){\tiny$\Omega_1$};
		\draw[thick](0,0)--(-5,-1);
        \node[below]at (-5,-1) {\tiny$\Sigma_2$};
         \node[below]at (-4.2,0.15){\tiny$\Omega_2$};
		\draw[thick](0,0)--(5,-1);
        \node[below]at (5,-1) {\tiny$\Sigma_2$};
        \node[below]at (4.2,0.15){\tiny$\Omega_2$};

        \draw[-latex][thick](-5,1)--(-2.5,0.5);
		\draw[-latex][thick](-5,-1)--(-2.5,-0.5);
		\draw[-latex][thick](0,0)--(2.5,0.5);
		\draw[-latex][thick](0,0)--(2.5,-0.5);

        \draw[thick,rotate=60](0,0)--(5,1);
        \node[left]at (1.6,4.9){\tiny$\Sigma_6$};
         \node[left]at (2.5,4) {\tiny$\Omega_6$};
		\draw[thick,rotate=60](0,0)--(-5,1);
        \node[right]at (-1.7,-5) {\tiny$\Sigma_5$};
        \node[left]at (-1.1,-4) {\tiny$\Omega_5$};
		\draw[thick,rotate=60](0,0)--(-5,-1);
        \node[left]at (-3.3,-3.6) {\tiny$\Sigma_6$};
        \node[right]at (-3.2,-3.6) {\tiny$\Omega_6$};
		\draw[thick,rotate=60](0,0)--(5,-1);
        \node[right]at (3.3,3.6) {\tiny$\Sigma_5$};
        \node[left]at (3.2,3.6) {\tiny$\Omega_5$};	

        \draw[-latex][thick,rotate=240](-5,1)--(-2.5,0.5);
		\draw[-latex][thick,rotate=240](-5,-1)--(-2.5,-0.5);
		\draw[-latex][thick,rotate=240](0,0)--(2.5,0.5);
		\draw[-latex][thick,rotate=240](0,0)--(2.5,-0.5);

        \draw[thick,rotate=120](0,0)--(5,1);
        \node[right]at (-1.6,4.9){\tiny$\Sigma_4$};
        \node[right]at (-2.35,3.8) {\tiny$\Omega_4$};
		\draw[thick,rotate=120](0,0)--(5,-1);
        \node[left]at (-3.2,3.6) {\tiny$\Sigma_3$};
        \node[right]at (-3,3.4) {\tiny$\Omega_3$};
		\draw[thick,rotate=120](0,0)--(-5,1);
        \node[left]at (1.7,-4.9) {\tiny$\Sigma_3$};
        \node[left]at (2.45,-3.9) {\tiny$\Omega_3$};
		\draw[thick,rotate=120](0,0)--(-5,-1);
        \node[right]at (3.3,-3.7) {\tiny$\Sigma_4$};	
        \node[left]at (3.2,-3.6) {\tiny$\Omega_4$};

        \draw[-latex][thick,rotate=120](-5,1)--(-2.5,0.5);
		\draw[-latex][thick,rotate=120](-5,-1)--(-2.5,-0.5);
		\draw[-latex][thick,rotate=120](0,0)--(2.5,0.5);
		\draw[-latex][thick,rotate=120](0,0)--(2.5,-0.5);

	   \draw [dotted] (0,0)circle(3.5cm);
	\end{tikzpicture}
\label{Sju}
\caption{The contours $\Sigma^{(ju)}$ for $\hat{\xi}\in(3,+\infty)$, while the contours $\Sigma^{(ju)}$ for $\hat{\xi}\in(-\infty,-\frac{3}{8})$ have the opposite direction.}
\end{center}
\end{figure}

\begin{lemma}\label{lR1}
Define continuous functions $R_j(k): \ \bar{\Omega}_j\rightarrow\mathbb{C}, \ j=1,\dots,6$, which have continuous first partial derivatives
on $\Omega_j$ and boundary values
  \begin{align}
  R_1(k)&=\left\{
        \begin{aligned}
    &\rho(k)(T_{12})_+(k), \quad k\in\mathbb{R},\\
    &0, \quad k\in\Sigma_1,
    \end{aligned}
        \right.\\
  R_2(k)&=\left\{
        \begin{aligned}
    &\bar{\rho}(k)(T_{21})_-(k),\quad k\in\mathbb{R},\\
    &0, \quad k\in\Sigma_2,
    \end{aligned}
        \right.\\
  R_3(k)&=\left\{
        \begin{aligned}
    &\rho(\omega^2k)(T_{31})_+(k), \quad k\in\omega\mathbb{R},\\
    &0, \quad k\in\Sigma_3,
    \end{aligned}
        \right.\\
  R_4(k)&=\left\{
        \begin{aligned}
    &\bar{\rho}(\omega^2k)(T_{13})_-(k), \quad k\in\omega\mathbb{R},\\
    &0, \quad k\in\Sigma_4,
    \end{aligned}
        \right.\\
    R_5(k)&=\left\{
        \begin{aligned}
    &\rho(\omega k)(T_{32})_+(k), \quad k\in\omega^2\mathbb{R},\\
    &0, \quad k\in\Sigma_5,
    \end{aligned}
        \right.\\
  R_6(k)&=\left\{
        \begin{aligned}
    &\bar{\rho}(\omega k)(T_{23})_-(k), \quad k\in\omega^2\mathbb{R},\\
    &0, \quad k\in\Sigma_6.
    \end{aligned}
        \right.
  \end{align}
Then we have
\begin{align}\label{dbarRj2}
	& |\bar{\partial}R_j(k)\left|\lesssim|\rho'\left({\rm sign}({\rm Re}k)|k|\right)\right|+|k|^{-\frac{1}{2}},  \quad k\in
         \Omega_j,\ j=1,\dots,6,\\
   &  \bar{\partial}R_j(k)=0, \quad \text{elsewhere}.
  \end{align}
\end{lemma}
Using $R_j(k)$ defined above, we can construct  the new matrix functions $\mathcal{R}^{(2)}(k)$ as
\begin{align}
\mathcal{R}^{(2)}(k)=\left\{
        \begin{aligned}
    &\begin{pmatrix} 1 & R_1(k)e^{\mathrm{i}t\theta_{12}} & 0 \\  0 & 1 & 0 \\ 0 & 0 & 1 \end{pmatrix}^{-1}, \quad k\in\Omega_1,\\
    &\begin{pmatrix} 1 &  0 & 0 \\ R_2(k)e^{-\mathrm{i}t\theta_{12}} & 1 & 0 \\ 0 & 0 & 1 \end{pmatrix}, \quad k\in\Omega_2,\\
    &\begin{pmatrix} 1 & 0 & 0 \\ 0 & 1 & 0 \\ R_{3}(k)e^{-\mathrm{i}t\theta_{13}} & 0 & 1 \end{pmatrix}^{-1}, \quad k\in\Omega_3,\\
    &\begin{pmatrix} 1 & R_4(k)e^{\mathrm{i}t\theta_{13}} & 0 \\ 0 & 1 & 0 \\ 0 & 0 & 1\end{pmatrix}, \quad k\in\Omega_4,\\
    &\begin{pmatrix} 1 & 0 & 0 \\ 0 & 1 & R_5(k)e^{\mathrm{i}t\theta_{23}} \\ 0 & 0 & 1 \end{pmatrix}^{-1}, \quad k\in\Omega_5,\\
    &\begin{pmatrix} 1 & 0 & 0 \\ 0 & 1 & 0 \\ 0 & R_6(k)e^{-\mathrm{i}t\theta_{23}} & 1\end{pmatrix}, \quad k\in\Omega_6,\\
    &I, \quad elsewhere.
    \end{aligned}
        \right.
\label{RR22}
\end{align}

\subsection{opening $\bar{\partial}$-lenses in region $\xi\in(-\frac{3}{8},0)\cup\left[0,3\right)$}

From \eqref{nop}, there are  $24$  and $12$ phase points in regions $\xi\in(-\frac{3}{8},0)$ and $\left[0,3\right)$ respectively. We open the contours  $\mathbb{R}\cup\omega\mathbb{R}\cup\omega^2\mathbb{R}$ at phase points $k_i,\ i=1,\dots,p(\hat{\xi})$. Define the contours $\Sigma_j$ and domains $\Omega_j$ in Figure \ref{Fig5} and let
\begin{equation}\label{defsigma}
\Sigma^{(ju)}= \mathop{\cup}\limits_{j=1}^6\Sigma_j, \quad \Omega=\mathop{\cup}\limits_{j=1}^6\Omega_j.\\
\end{equation}
\begin{figure}
\begin{center}
	\subfigure[  The opened contour $\Sigma^{(ju)}  $ for $\hat{\xi}\in\left[0,3\right)$. ]{
\begin{tikzpicture}[scale=0.8]
\draw [thick]  (-4,-0.6) to [out=0,in=180] (-2,0.6)
to [out=0,in=180] (0,-0.6) to [out=0,in=180] (2,0.6)  to  [out=0,in=180] (4,-0.6);
\draw [thick,rotate=240] (-4,-0.6) to [out=0,in=180] (-2,0.6)
to [out=0,in=180] (0,-0.6) to [out=0,in=180] (2,0.6)  to  [out=0,in=180] (4,-0.6);
\draw [thick,rotate=120] (-4,-0.6) to [out=0,in=180] (-2,0.6)
to [out=0,in=180] (0,-0.6) to [out=0,in=180] (2,0.6)  to  [out=0,in=180] (4,-0.6);

\draw [thick] (-4,0.6) to [out=0,in=180] (-2,-0.6)
to [out=0,in=180] (0,0.6) to [out=0,in=180] (2,-0.6)  to [out=0,in=180] (4,0.6);
\draw [thick,rotate=240](-4,0.6) to [out=0,in=180] (-2,-0.6)
to [out=0,in=180] (0,0.6) to [out=0,in=180] (2,-0.6)  to [out=0,in=180] (4,0.6);
\draw [thick,rotate=120](-4,0.6) to [out=0,in=180] (-2,-0.6)
to [out=0,in=180] (0,0.6) to [out=0,in=180] (2,-0.6)  to [out=0,in=180] (4,0.6);

\filldraw [white] (0,0) circle [radius=1];
\draw[dotted](-5,0)--(5,0);
\draw[dotted,rotate=240](-5,0)--(5,0);
\draw[dotted,rotate=120](-5,0)--(5,0);
\draw [thick] (0,0) circle [radius=1];

\draw[-latex](0.88,0.48)--(0.9,0.44);
\draw[-latex](0.88,-0.48)--(0.9,-0.44);
\draw[-latex](-0.8,-0.6)--(-0.79,-0.613);
\draw[-latex](-0.8,0.6)--(-0.79,0.613);
\draw[-latex](0.1,0.995)--(-0.1,0.995);
\draw[-latex](0.1,-0.995)--(-0.1,-0.995);

\draw[-latex](-3.2,0.21)--(-3.15,0.16);
\draw[-latex](-3.2,-0.21)--(-3.15,-0.16);
\draw[-latex](-2.8,-0.21)--(-2.85,-0.15);
\draw[-latex](-2.8,0.21)--(-2.85,0.15);
\draw[-latex](-1.2,0.21)--(-1.3,0.3);
\draw[-latex](-1.2,-0.21)--(-1.3,-0.3);

\draw[-latex](1.2,-0.21)--(1.15,-0.15);
\draw[-latex](1.2,0.21)--(1.15,0.15);
\draw[-latex](2.8,0.21)--(2.7,0.3);
\draw[-latex](2.8,-0.21)--(2.7,-0.3);
\draw[-latex](3.2,-0.21)--(3.3,-0.3);
\draw[-latex](3.2,0.21)--(3.3,0.3);

\draw[-latex,rotate=240](-3.2,0.21)--(-3.15,0.16);
\draw[-latex,rotate=240](-3.2,-0.21)--(-3.15,-0.16);
\draw[-latex,rotate=240](-2.8,-0.21)--(-2.85,-0.15);
\draw[-latex,rotate=240](-2.8,0.21)--(-2.85,0.15);
\draw[-latex,rotate=240](-1.2,0.21)--(-1.3,0.3);
\draw[-latex,rotate=240](-1.2,-0.21)--(-1.3,-0.3);

\draw[-latex,rotate=120](-3.2,0.21)--(-3.15,0.16);
\draw[-latex,rotate=120](-3.2,-0.21)--(-3.15,-0.16);
\draw[-latex,rotate=120](-2.8,-0.21)--(-2.85,-0.15);
\draw[-latex,rotate=120](-2.8,0.21)--(-2.85,0.15);
\draw[-latex,rotate=120](-1.2,0.21)--(-1.3,0.3);
\draw[-latex,rotate=120](-1.2,-0.21)--(-1.3,-0.3);

\draw[-latex,rotate=240](1.2,-0.21)--(1.15,-0.15);
\draw[-latex,rotate=240](1.2,0.21)--(1.15,0.15);

\draw[-latex,rotate=120](1.2,-0.21)--(1.15,-0.15);
\draw[-latex,rotate=120](1.2,0.21)--(1.15,0.15);

\draw[-latex,rotate=240](2.8,0.21)--(2.7,0.3);
\draw[-latex,rotate=240](2.8,-0.21)--(2.7,-0.3);
\draw[-latex,rotate=240](3.2,-0.21)--(3.3,-0.3);
\draw[-latex,rotate=240](3.2,0.21)--(3.3,0.3);

\draw[-latex,rotate=120](2.8,0.21)--(2.7,0.3);
\draw[-latex,rotate=120](2.8,-0.21)--(2.7,-0.3);
\draw[-latex,rotate=120](3.2,-0.21)--(3.3,-0.3);
\draw[-latex,rotate=120](3.2,0.21)--(3.3,0.3);
\node  at (4.3,0.6) {\scriptsize$\Sigma_1$};
\node  at (4.3,-0.6) {\scriptsize$\Sigma_2$};

\coordinate (A) at (1.6,0);
\fill[red] (A) circle [radius=0.03] node[below] {\tiny$1$};
\coordinate (B) at (-1.6,0);
\fill[red] (B) circle [radius=0.03] node[below] {\tiny$-1$};

\draw [dotted] (0,0) circle [radius=1];
\draw [dotted] (0,0) circle [radius=3];

\filldraw  (1,0) circle [radius=0.05];
\node [below] at  (1.1,-0.1) {\tiny$k_2$};
\filldraw  (3,0) circle [radius=0.05];
\node [below] at  (3,-0.1) {\tiny$k_1$};
\filldraw  (-1,0) circle [radius=0.05];
\node [below] at  (-1.1,-0.1) {\tiny$k_3$};
\filldraw  (-3,0) circle [radius=0.05];
\node [below] at  (-3,-0.1) {\tiny$k_4$};

\coordinate (I) at (0,0);
\fill[red] (I) circle [radius=0.03] node[below] {\tiny$O$};

\node  at (-2.5,3.45) {\scriptsize$\Sigma_3$};
\node  at (-1.5,4.05) {\scriptsize$\Sigma_4$};
\filldraw  (-0.5,0.87) circle [radius=0.05];
\node [above] at   (-0.85,0.9) {\tiny$\omega k_2$};
\filldraw  (-1.5,2.6) circle [radius=0.05];
\node [above] at  (-1.75,2.7) {\tiny$\omega k_1$};
\filldraw  (0.5,-0.87) circle [radius=0.05];
\node [below] at   (0.9,-0.9) {\tiny$\omega k_3$};
\filldraw  (1.5,-2.6) circle [radius=0.05];
\node [below] at  (1.85,-2.7) {\tiny$\omega k_4$};

\node  at (-2.5,-3.45) {\scriptsize$\Sigma_5$};
\node  at (-1.5,-4.05) {\scriptsize$\Sigma_6$};
\filldraw  (-0.5,-0.87) circle [radius=0.05];
\node [below] at   (-0.85,-0.9) {\tiny$\omega^2 k_2$};
\filldraw  (-1.5,-2.6) circle [radius=0.05];
\node [below] at  (-1.75,-2.7) {\tiny$\omega^2 k_1$};
\filldraw  (0.5,0.87) circle [radius=0.05];
\node [above] at   (0.9,0.9) {\tiny$\omega^2 k_3$};
\filldraw  (1.5,2.6) circle [radius=0.05];
\node [above] at  (1.85,2.7) {\tiny$\omega^2 k_4$};
		
\end{tikzpicture}
		\label{case1}}
	\subfigure[ The opened contour $\Sigma^{(ju)}  $ for $\hat{\xi}\in(-\frac{3}{8},0)$]{
\begin{tikzpicture}[scale=0.8]
\draw[dotted](-6.5,0)--(6.8,0);
\draw [thick] (-6,-0.6)to [out=0,in=180](-4.5,0.6)to [out=0,in=180](-3,-0.6) to [out=0,in=180] (-1.5,0.6)
to [out=0,in=180] (0,-0.6) to [out=0,in=180] (1.5,0.6)  to  [out=0,in=180] (3,-0.6) to [out=0,in=180] (4.5,0.6) to
[out=0,in=180] (6,-0.6);
\draw [thick](-6,0.6)to [out=0,in=180](-4.5,-0.6)to [out=0,in=180](-3,0.6) to [out=0,in=180] (-1.5,-0.6)
to [out=0,in=180] (0,0.6) to [out=0,in=180] (1.5,-0.6)  to [out=0,in=180] (3,0.6) to [out=0,in=180] (4.5,-0.6) to  [out=0,in=180] (6,0.6);

\draw[dotted,rotate=240](-6.5,0)--(6.8,0);
\draw [thick,rotate=240] (-6,-0.6)to [out=0,in=180](-4.5,0.6)to [out=0,in=180](-3,-0.6) to [out=0,in=180] (-1.5,0.6)
to [out=0,in=180] (0,-0.6) to [out=0,in=180] (1.5,0.6)  to  [out=0,in=180] (3,-0.6) to [out=0,in=180] (4.5,0.6) to
[out=0,in=180] (6,-0.6);
\draw [thick,rotate=240](-6,0.6)to [out=0,in=180](-4.5,-0.6)to [out=0,in=180](-3,0.6) to [out=0,in=180] (-1.5,-0.6)
to [out=0,in=180] (0,0.6) to [out=0,in=180] (1.5,-0.6)  to [out=0,in=180] (3,0.6) to [out=0,in=180] (4.5,-0.6) to  [out=0,in=180] (6,0.6);

\draw[dotted,rotate=120](-6.5,0)--(6.8,0);
\draw [thick,rotate=120] (-6,-0.6)to [out=0,in=180](-4.5,0.6)to [out=0,in=180](-3,-0.6) to [out=0,in=180] (-1.5,0.6)
to [out=0,in=180] (0,-0.6) to [out=0,in=180] (1.5,0.6)  to  [out=0,in=180] (3,-0.6) to [out=0,in=180] (4.5,0.6) to
[out=0,in=180] (6,-0.6);
\draw [thick,rotate=120](-6,0.6)to [out=0,in=180](-4.5,-0.6)to [out=0,in=180](-3,0.6) to [out=0,in=180] (-1.5,-0.6)
to [out=0,in=180] (0,0.6) to [out=0,in=180] (1.5,-0.6)  to [out=0,in=180] (3,0.6) to [out=0,in=180] (4.5,-0.6) to  [out=0,in=180] (6,0.6);

\filldraw [white] (0,0) circle [radius=0.75];
\draw[dotted](-6.5,0)--(6.8,0);
\draw[dotted,rotate=240](-6.5,0)--(6.8,0);
\draw[dotted,rotate=120](-6.5,0)--(6.8,0);
\draw [thick] (0,0) circle [radius=0.75];

\draw [dotted] (0,0) circle [radius=0.75];
\draw [dotted] (0,0) circle [radius=2.25];
\draw [dotted] (0,0) circle [radius=3.75];
\draw [dotted] (0,0) circle [radius=5.25];

\draw[-latex](0.69,0.29)--(0.7,0.27);
\draw[-latex](0.69,-0.29)--(0.7,-0.27);
\draw[-latex](-0.63,0.4)--(-0.62,0.42);
\draw[-latex](-0.65,-0.37)--(-0.64,-0.39);
\draw[-latex](0.1,0.74)--(-0.1,0.74);
\draw[-latex](0.1,-0.74)--(-0.1,-0.74);

\draw[-latex](-5.5,0.35)--(-5.4,0.23);
\draw[-latex](-5,0.36)--(-5.1,0.23);
\draw[-latex](-3.9,0.23)--(-4,0.35);
\draw[-latex](-3.6,0.23)--(-3.5,0.36);
\draw[-latex](-2.5,0.35)--(-2.4,0.23);
\draw[-latex](-2,0.36)--(-2.1,0.23);
\draw[-latex](-0.9,0.23)--(-1,0.35);

\draw[-latex](5.4,0.23)--(5.5,0.35);
\draw[-latex](5.1,0.23)--(5,0.36);
\draw[-latex](4,0.35)--(3.9,0.23);
\draw[-latex](3.5,0.36)--(3.6,0.23);
\draw[-latex](2.4,0.23)--(2.5,0.35);
\draw[-latex](2.1,0.23)--(2,0.36);
\draw[-latex](1,0.35)--(0.9,0.23);

\draw[-latex](-5.5,-0.35)--(-5.4,-0.23);
\draw[-latex](-5,-0.36)--(-5.1,-0.23);
\draw[-latex](-3.9,-0.23)--(-4,-0.35);
\draw[-latex](-3.6,-0.23)--(-3.5,-0.36);
\draw[-latex](-2.5,-0.35)--(-2.4,-0.23);
\draw[-latex](-2,-0.36)--(-2.1,-0.23);
\draw[-latex](-0.9,-0.23)--(-1,-0.35);

\draw[-latex](5.4,-0.23)--(5.5,-0.35);
\draw[-latex](5.1,-0.23)--(5,-0.36);
\draw[-latex](4,-0.35)--(3.9,-0.23);
\draw[-latex](3.5,-0.36)--(3.6,-0.23);
\draw[-latex](2.4,-0.23)--(2.5,-0.35);
\draw[-latex](2.1,-0.23)--(2,-0.36);
\draw[-latex](1,-0.35)--(0.9,-0.23);

\draw[-latex,rotate=120](-5.5,0.35)--(-5.4,0.23);
\draw[-latex,rotate=120](-5,0.36)--(-5.1,0.23);
\draw[-latex,rotate=120](-3.9,0.23)--(-4,0.35);
\draw[-latex,rotate=120](-3.6,0.23)--(-3.5,0.36);
\draw[-latex,rotate=120](-2.5,0.35)--(-2.4,0.23);
\draw[-latex,rotate=120](-2,0.36)--(-2.1,0.23);
\draw[-latex,rotate=120](-0.9,0.23)--(-1,0.35);

\draw[-latex,rotate=120](5.4,0.23)--(5.5,0.35);
\draw[-latex,rotate=120](5.1,0.23)--(5,0.36);
\draw[-latex,rotate=120](4,0.35)--(3.9,0.23);
\draw[-latex,rotate=120](3.5,0.36)--(3.6,0.23);
\draw[-latex,rotate=120](2.4,0.23)--(2.5,0.35);
\draw[-latex,rotate=120](2.1,0.23)--(2,0.36);
\draw[-latex,rotate=120](1,0.35)--(0.9,0.23);

\draw[-latex,rotate=120](-5.5,-0.35)--(-5.4,-0.23);
\draw[-latex,rotate=120](-5,-0.36)--(-5.1,-0.23);
\draw[-latex,rotate=120](-3.9,-0.23)--(-4,-0.35);
\draw[-latex,rotate=120](-3.6,-0.23)--(-3.5,-0.36);
\draw[-latex,rotate=120](-2.5,-0.35)--(-2.4,-0.23);
\draw[-latex,rotate=120](-2,-0.36)--(-2.1,-0.23);
\draw[-latex,rotate=120](-0.9,-0.23)--(-1,-0.35);

\draw[-latex,rotate=120](5.4,-0.23)--(5.5,-0.35);
\draw[-latex,rotate=120](5.1,-0.23)--(5,-0.36);
\draw[-latex,rotate=120](4,-0.35)--(3.9,-0.23);
\draw[-latex,rotate=120](3.5,-0.36)--(3.6,-0.23);
\draw[-latex,rotate=120](2.4,-0.23)--(2.5,-0.35);
\draw[-latex,rotate=120](2.1,-0.23)--(2,-0.36);
\draw[-latex,rotate=120](1,-0.35)--(0.9,-0.23);

\draw[-latex,rotate=240](-5.5,0.35)--(-5.4,0.23);
\draw[-latex,rotate=240](-5,0.36)--(-5.1,0.23);
\draw[-latex,rotate=240](-3.9,0.23)--(-4,0.35);
\draw[-latex,rotate=240](-3.6,0.23)--(-3.5,0.36);
\draw[-latex,rotate=240](-2.5,0.35)--(-2.4,0.23);
\draw[-latex,rotate=240](-2,0.36)--(-2.1,0.23);
\draw[-latex,rotate=240](-0.9,0.23)--(-1,0.35);

\draw[-latex,rotate=240](5.4,0.23)--(5.5,0.35);
\draw[-latex,rotate=240](5.1,0.23)--(5,0.36);
\draw[-latex,rotate=240](4,0.35)--(3.9,0.23);
\draw[-latex,rotate=240](3.5,0.36)--(3.6,0.23);
\draw[-latex,rotate=240](2.4,0.23)--(2.5,0.35);
\draw[-latex,rotate=240](2.1,0.23)--(2,0.36);
\draw[-latex,rotate=240](1,0.35)--(0.9,0.23);

\draw[-latex,rotate=240](-5.5,-0.35)--(-5.4,-0.23);
\draw[-latex,rotate=240](-5,-0.36)--(-5.1,-0.23);
\draw[-latex,rotate=240](-3.9,-0.23)--(-4,-0.35);
\draw[-latex,rotate=240](-3.6,-0.23)--(-3.5,-0.36);
\draw[-latex,rotate=240](-2.5,-0.35)--(-2.4,-0.23);
\draw[-latex,rotate=240](-2,-0.36)--(-2.1,-0.23);
\draw[-latex,rotate=240](-0.9,-0.23)--(-1,-0.35);

\draw[-latex,rotate=240](5.4,-0.23)--(5.5,-0.35);
\draw[-latex,rotate=240](5.1,-0.23)--(5,-0.36);
\draw[-latex,rotate=240](4,-0.35)--(3.9,-0.23);
\draw[-latex,rotate=240](3.5,-0.36)--(3.6,-0.23);
\draw[-latex,rotate=240](2.4,-0.23)--(2.5,-0.35);
\draw[-latex,rotate=240](2.1,-0.23)--(2,-0.36);
\draw[-latex,rotate=240](1,-0.35)--(0.9,-0.23);

\node  at (6.3,0.6) {\tiny$\Sigma_1$};
\node  at (6.3,-0.6) {\tiny$\Sigma_2$};
\node  at (-3.55,5.1) {\tiny$\Sigma_3$};
\node  at (-2.45,5.75) {\tiny$\Sigma_4$};
\node  at (-3.55,-5.1) {\tiny$\Sigma_5$};
\node  at (-2.45,-5.75) {\tiny$\Sigma_6$};

\coordinate (I) at (0,0);
		\fill[red] (I) circle (1pt) node[below] {\tiny$O$};
\filldraw [red] (2.7,0) circle [radius=0.03];
\node [below,red] at  (2.8,0) {\tiny$ 1$};
\filldraw [red] (-2.7,0) circle [radius=0.03];
\node [below,red] at  (-2.8,0) {\tiny$ -1$};

\filldraw  (3.75,0) circle [radius=0.05];
\node [below] at  (3.75,-0.1) {\tiny$ k_2$};
\filldraw  (5.25,0) circle [radius=0.05];
\node [below] at   (5.27,-0.1) {\tiny$ k_1$};
\filldraw  (2.25,0) circle [radius=0.05];
\node [below] at   (2.28,-0.1) {\tiny$ k_3$};
\filldraw  (0.75,0) circle [radius=0.05];
\node [below] at   (0.83,-0.15) {\tiny$ k_4$};
\filldraw  (-3.75,0) circle [radius=0.05];
\node [below] at   (-3.75,-0.1) {\tiny$ k_7$};
\filldraw  (-5.25,0) circle [radius=0.05];
\node [below] at   (-5.25,-0.1) {\tiny$ k_8$};
\filldraw  (-2.25,0) circle [radius=0.05];
\node [below] at   (-2.25,-0.1) {\tiny$ k_6$};
\filldraw  (-0.75,0) circle [radius=0.05];
\node [below] at   (-0.75,-0.12) {\tiny$k_5$};

\filldraw  (-1.88,3.25) circle [radius=0.05];
\node [above] at    (-2.05,3.25) {\tiny$\omega k_2$};
\filldraw  (-2.63,4.55) circle [radius=0.05];
\node [above] at   (-2.8,4.55) {\tiny$\omega k_1$};
\filldraw  (-1.13,1.95) circle [radius=0.05];
\node [above] at    (-1.35,1.95) {\tiny$\omega k_3$};
\filldraw  (-0.38,0.65) circle [radius=0.05];
\node [above] at  (-0.6,0.65) {\tiny$\omega k_4$};
\filldraw  (1.88,-3.25) circle [radius=0.05];
\node [below] at    (2.15,-3.25) {\tiny$\omega k_7$};
\filldraw  (2.63,-4.55) circle [radius=0.05];
\node [below] at   (2.8,-4.55) {\tiny$\omega k_8$};
\filldraw  (1.13,-1.95) circle [radius=0.05];
\node [below] at    (1.35,-1.95) {\tiny$\omega k_6$};
\filldraw  (0.38,-0.65) circle [radius=0.05];
\node [below] at  (0.6,-0.65) {\tiny$\omega k_5$};

\filldraw  (-1.88,-3.25) circle [radius=0.05];
\node [below] at    (-2.1,-3.25) {\tiny$\omega^2 k_2$};
\filldraw  (-2.63,-4.55) circle [radius=0.05];
\node [below] at   (-2.8,-4.55) {\tiny$\omega^2 k_1$};
\filldraw  (-1.13,-1.95) circle [radius=0.05];
\node [below] at    (-1.35,-1.95) {\tiny$\omega^2 k_3$};
\filldraw  (-0.38,-0.65) circle [radius=0.05];
\node [below] at  (-0.6,-0.65) {\tiny$\omega^2 k_4$};
\filldraw  (1.88,3.25) circle [radius=0.05];
\node [above] at    (2.2,3.25) {\tiny$\omega^2 k_7$};
\filldraw  (2.63,4.55) circle [radius=0.05];
\node [above] at   (2.9,4.55) {\tiny$\omega^2 k_8$};
\filldraw  (1.13,1.95) circle [radius=0.05];
\node [above] at    (1.45,1.95) {\tiny$\omega^2 k_6$};
\filldraw  (0.38,0.65) circle [radius=0.05];
\node [above] at  (0.7,0.65) {\tiny$\omega^2 k_5$};
		\end{tikzpicture}
		\label{case2}}
	\caption{\footnotesize
Opening the original jump contours $\mathbb{R}\cup\omega\mathbb{R}\cup\omega^2\mathbb{R}$ at phase points  $k_i,\ \omega k_i,\ \omega^2 k_i, \ i=1,\cdots, p(\xi)$. }
	\label{Fig5}
\end{center}
\end{figure}
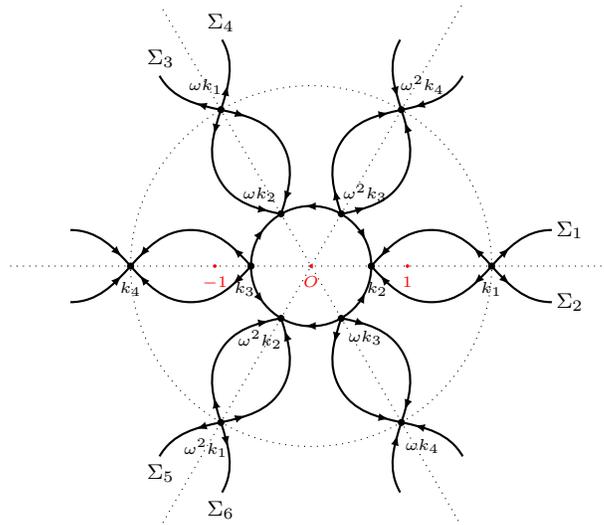
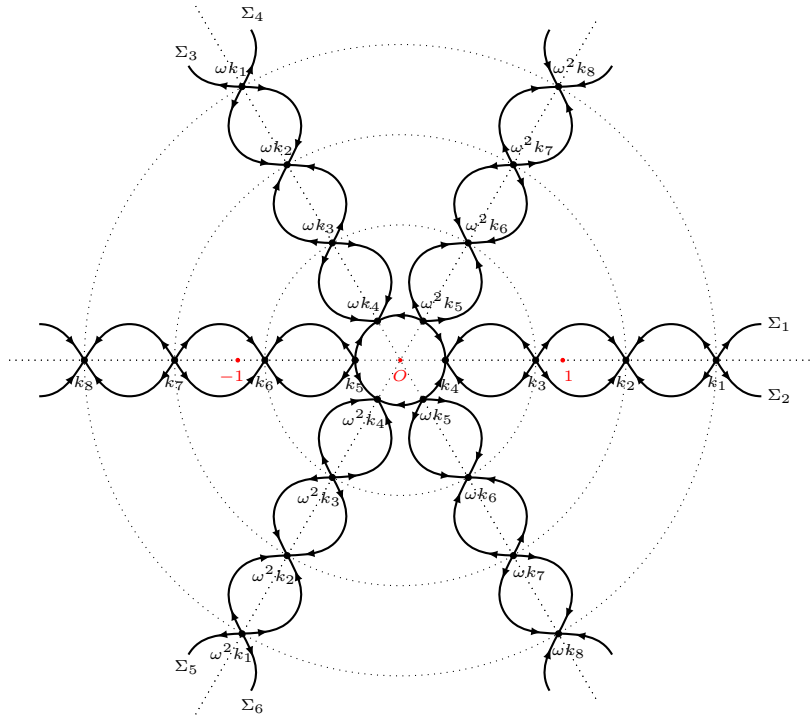

Now we define extension functions as follows.
\begin{lemma}\label{estRij}
Define functions $R_{ij}(k): \ \bar{\Omega}_j\rightarrow\mathbb{C}$ continuous on $\Omega_j$, with continuous first partial derivative
on $\bar\Omega_j, \ j=1,\dots,6$ and boundary values
  \begin{align}
  &R_1(k)=\Bigg\{\begin{array}{ll}
    \rho(k) (T_{12})_+ (k), \ k\in \mathbb{R},\\
    \rho(k_i)\big(T^{(i)}_{12}\big)_+(k_i)\big(\eta(k-k_i)\big)^{2\mathrm{i}\eta\nu(k_i)}\big(1-\mathcal{X}_{\mathcal{K}}(k)\big), \ k\in\Sigma_1,
    \end{array}\label{R21}\\
  &R_2(k)=\Bigg\{\begin{array}{ll}
    \bar{\rho}(k)(T_{12})^{-1}_-(k), \ k\in \mathbb{R},\\
    \bar{\rho}(k_i)\big(T^{(i)}_{12}\big)^{-1}_-(k_i)\big(\eta(k-k_i)\big)^{-2\mathrm{i}\eta\nu(k_i)}\big(1-\mathcal{X}_{\mathcal{K}}(k)\big), \ k\in\Sigma_2,
  \end{array}\\
  &R_3(k)=\Bigg\{\begin{array}{ll}
    \rho(\omega^2 k) (T_{13})^{-1}_+ (k), \ k\in \omega\mathbb{R},\\
    \rho(\omega^2 k_i)\big(T^{(i)}_{13}\big)^{-1}_+(\omega k_i)\big(\eta(k-\omega k_i)\big)^{2\mathrm{i}\eta\nu(\omega k_i)}\big(1-\mathcal{X}_{\mathcal{K}}(k)\big), \ k\in\Sigma_3,
   \end{array}\\
  &R_4(k)=\Bigg\{\begin{array}{ll}
    \bar{\rho}(\omega^2 k)(T_{13})_-(k), \ k\in\omega\mathbb{R},\\
    \bar{\rho}(\omega^2 k_i)\big(T^{(i)}_{13}\big)_-(\omega k_i)\big(\eta(k-\omega k_i)\big)^{-2\mathrm{i}\eta\nu(\omega k_i)}\big(1-\mathcal{X}_{\mathcal{K}}(k)\big), \ k\in\Sigma_4,
   \end{array}\\
   &R_5(k)=\Bigg\{\begin{array}{ll}
    \rho(\omega k) (T_{23})^{-1}_+ (k), \ k\in \omega^2\mathbb{R},\\
    \rho(\omega k_i)\big(T^{(i)}_{23}\big)^{-1}_+(\omega^2 k_i)\big(\eta(k-\omega^2 k_i)\big)^{2\mathrm{i}\eta\nu(\omega^2 k_i)}\big(1-\mathcal{X}_{\mathcal{K}}(k)\big), \ k\in\Sigma_5,
   \end{array}\\
   &R_6(k)=\Bigg\{\begin{array}{ll}
    \bar{\rho}(\omega k)(T_{23})_-(k), \ k\in\omega^2\mathbb{R},\\
    \bar{\rho}(\omega k_i)\big(T^{(i)}_{23}\big)_-(\omega^2 k_i)\big(\eta(k-\omega^2 k_i)\big)^{-2\mathrm{i}\eta\nu(\omega^2 k_i)}\big(1-\mathcal{X}_{\mathcal{K}}(k)\big), \ k\in\Sigma_6,
   \end{array}
  \end{align}
where $\mathcal{X}_{\mathcal{K}}(k)\in C^\infty_0(\mathbb{C},[0,1])$ is supported near the discrete spectrum such that
\begin{align}
	\mathcal{X}_{\mathcal{K}}(k)=\left\{\begin{array}{llll}
		1, & {\rm}dist(k,\mathcal{K})<\varrho/3,\\[4pt]
		1,  & {\rm}dist(k,\mathcal{K})>2\varrho/3,
	\end{array}\right.
\end{align}
$\eta:=\eta(\hat{\xi},i)$ is a constant depend on $\hat{\xi}$ and $i$
\begin{align}
\eta(\hat{\xi},i)=\left\{ \begin{array}{ll}
		(-1)^i, \quad \hat{\xi}\in(-\frac{3}{8},0),\\
	(-1)^{i+1}, \quad \hat{\xi}\in\left[0,3\right),
	\end{array}\right.\label{eta}
\end{align}
and $\nu(k)=-\frac{1}{2\pi}\log{1-\lvert r(k)\rvert^2}$.
Then for $k\in \Omega_j$, we have
\begin{align}
&|R_j(k)|\lesssim \sin^2\big(k_0\arg(k-k_i)\big)+\langle {\rm Re}(k) \rangle^{-\frac{1}{2}},\label{efR}
\\
&|\bar{\partial}R_j(k)|\lesssim|\rho'( {\rm Re}(k))|+|\mathcal{X}_{\mathcal{K}}^{'}(k)|+|k-k_i|^{-\frac{1}{2}}.\label{br1}
\end{align}
\end{lemma}
By using $R_{j}(k)$ obtained by Lemma \ref{estRij},  we define
\begin{align}
\mathcal{R}^{(2)}(k)=\left\{
\begin{aligned}
    &\begin{pmatrix} 1 & 0 & 0 \\  R_1(k)e^{-\mathrm{i}t\theta_{12}(k)} & 1 & 0 \\ 0 & 0 & 1 \end{pmatrix}^{-1}, \quad k\in\Omega_1, \\
    &\begin{pmatrix} 1 &  R_2(k)e^{\mathrm{i}t\theta_{12}(k)} & 0 \\ 0 & 1 & 0 \\ 0 & 0 & 1 \end{pmatrix}, \quad k\in\Omega_2,\\
    &\begin{pmatrix} 1 & 0 & R_3(\omega k)e^{\mathrm{i}t\theta_{13}(k)} \\ 0 & 1 & 0 \\ 0 & 0 & 1 \end{pmatrix}^{-1}, \ k\in\Omega_3,\\
    &\begin{pmatrix} 1 & 0 & 0 \\ 0 & 1 & 0 \\ R_4(\omega k)e^{-\mathrm{i}t\theta_{13}(k)} & 0 & 1\end{pmatrix}, \ k\in\Omega_4,\\
    &\begin{pmatrix} 1 & 0 & 0 \\ 0 & 1 & 0 \\ 0 & R_5(\omega^2 k)e^{-\mathrm{i}t\theta_{23}(k)} & 1 \end{pmatrix}^{-1}, \ k\in\Omega_5,\\
    &\begin{pmatrix} 1 & 0 & 0 \\ 0 & 1 & R_6(\omega^2 k)e^{\mathrm{i}t\theta_{23}(k)} \\ 0 & 0 & 1\end{pmatrix}, \quad k\in\Omega_6,\\
    &I, \quad elsewhere.
\end{aligned}
\right.
\end{align}

\subsection{A hybrid $\bar{\partial}$-RH problem and its decomposition}

Define the new jump contours as
\begin{align}
\Sigma^{(2)}=\left\{
\begin{aligned}
    &\Sigma^{(ju)}, \quad\hat{\xi}\in(-\infty,-\frac{3}{8})\cup(3,+\infty),\\
    &\Sigma^{(ju)}\cup\Sigma^{(ci)}, \quad \hat{\xi}\in(-\frac{3}{8},3).
\end{aligned}
\right.
\end{align}
With the aid of $\mathcal{R}^{(2)}(k)$, we make the transformation
\begin{equation}
	m^{(2)}(k) =m^{(1)}(k)\mathcal{R}^{(2)}(k)\label{transm22},
\end{equation}
which satisfies the following hybrid $\bar{\partial}$-RH problem.
\begin{RHP}\label{rhpM22}
Find a row vector-valued function  $m^{(2)}(k):= m^{(2)}(k;y,t)$ such that
\begin{itemize}
\item $m^{(2)}(k)$  has  sectionally continuous first partial derivatives in
$\mathbb{C}\setminus \Big( \Sigma^{(2)}\cup \left\lbrace\zeta_n \right\rbrace_{n-k_0N\in\Lambda} \Big) $,
 and is meromorphic out $\bar{\Omega}$.
\item $m^{(2)}(k)$ satisfies the jump relation
\begin{equation}
	m^{(2)}_+(k)=m^{(2)}_-(k)V^{(2)}(k),\quad k \in \Sigma^{(2)},
\end{equation}
where for $\hat{\xi}\in(-\infty,-\frac{3}{8})\cup(3,+\infty)$,
\begin{equation}\label{V21}
	V^{(2)}(k)=\left\{  \begin{array}{ll}
		T^{-1}(k)G(k)T(k), \quad k\in\partial\mathbb{D}_n\cap \left(\mathop{\cup}\limits_{\nu=1}^3D_{2\nu-1}\right),\\[12pt]
		T^{-1}(k)G^{-1}(k)T(k), \quad k\in\partial\mathbb{D}_n\cap \left(\mathop{\cup}\limits_{\nu=1}^3D_{2\nu}\right).
 \end{array}
 \right.
\end{equation}	
and for $\hat{\xi}\in(-\frac{3}{8},3)$,
\begin{align}\label{V222}
V^{(2)}(k)=\left\{
\begin{aligned}
    &\underset{k'\in\Omega\to k\in\Sigma^{(ju)}}{\lim}\mathcal{R}^{(2)}(k')^{-1}, \quad k\in\Sigma^{(ju)}\cap D_{2\nu-1},\\
    &\underset{k'\in\Omega\to k\in\Sigma^{(ju)}}{\lim}\mathcal{R}^{(2)}(k'), \quad k\in\Sigma^{(ju)}\cap D_{2\nu},\\
    &T^{-1}(k)G(k)T(k), \quad k\in\partial\mathbb{D}_n\cap D_{2\nu-1},\\
    &T^{-1}(k)G^{-1}(k)T(k), \quad k\in\partial\mathbb{D}_n\cap D_{2\nu}.
\end{aligned}
\right.
\end{align}
\item  $m^{(2)}(k) = (1, \ 1, \ 1)+\mathcal{O}(k^{-1}),\quad k \rightarrow\infty$.
\item  For $k\in\mathbb{C}$, we have
\begin{align}
	\bar{\partial}m^{(2)}(k)=m^{(2)}(k)\bar{\partial}\mathcal{R}^{(2)}(k).
\end{align}
\item $m^{(2)}(k)$ has simple poles at each point $\zeta_n, \ n-k_0N\in\Lambda$ with
\begin{equation}\label{resM22}
\underset{k=\zeta_n}{\rm Res}m^{(2)}(k)=\lim_{k\to \zeta_n}m^{(2)}(k)\Big(T^{-1}(k)B_nT(k)\Big).
\end{equation}
\end{itemize}	
\end{RHP}

To solve the RH problem \ref{rhpM22}, we decompose $ m^{(2)}(k)$ into a pure RH problem as $M^{R}(k)$ under the condition $\bar{\partial}\mathcal{R}^{(2)}(k)\equiv 0$, as well as a pure
$\bar{\partial}$-problem $m^{(3)}(k)$ with $\bar{\partial}\mathcal{R}^{(2)}(k)\neq 0$
\begin{equation}\label{m3m2}
  m^{(2)}(k)= m^{(3)}(k)M^{R}(k),
\end{equation}
where $M^{R}(k)$ and $m^{(3)}(k)$ will be solved in RH problem \ref{Mrhp} and $\bar{\partial}$-problem \ref{dbarproblem}, respectively.
\begin{RHP}\label{Mrhp}
Find a matrix-valued function $M^{R}(k):=M^{R}(k;y,t)$ such that
\begin{itemize}
\item $M^{R}(k)$ is analytic in
$\mathbb{C}\setminus \Big( \Sigma^{(2)}\cup \left\lbrace\zeta_n \right\rbrace_{n-k_0n\in\Lambda} \Big) $.
\item $M^{R}(k)$ has the same jump relation as $m^{(2)}(k)$.
\item  $M^{R}(k) =I+\mathcal{O}(k^{-1}),\hspace{0.5cm}k \rightarrow \infty.$
\item $M^{{R}}(k)$ has simple poles at $\zeta_n, \ n-K_0N\in\Lambda$ with residue condition
\begin{equation}
\underset{k=\zeta_n}{\rm Res}M^{R}(k)=\lim_{k\to \zeta_n}M^{R}(k)\Big( T^{-1}(k)B_nT(k)\Big).
\end{equation}
\end{itemize}	
\end{RHP}
The solvability of the RH problem \ref{Mrhp} can be proved in the following lemma.
\begin{lemma}\label{Mout}
The pure RH problem \ref{Mrhp} admits an unique solution given by
\begin{equation}
 M^{R}(k)=M^{sol}(k|\tilde{\mathcal{D}}),
\end{equation}
where $M^{sol}$ is the solution of RH problem \ref{rhpM} corresponding to the reflectionless scattering data $\mathcal{D}= \{\zeta_n,C_n\}_{n=1}^{N}$. And $\tilde{\mathcal{D}}= \{\zeta_n,\tilde{C}_n\}_{n=1}^{N}$ is the modified scattering data, where
\begin{equation}\label{tC}
\tilde{C}_n(x,t)=C_n(x,t)\delta_{\zeta_{n}}(x,t)
\end{equation}
with
\begin{align}\label{delzetan}
\delta_{\zeta_{n}}=\left\{
\begin{aligned}
    &\frac{\delta_1(\omega^2\zeta_{n})\delta_1(\omega\zeta_{n})}{\delta_1^{2}(\zeta_{n})}, \ n\in\widetilde{\mathcal{N}},\\
    &\frac{\delta_1(\omega^2\zeta_{n})\delta_1(\zeta_{n})}{\delta_1^{2}(\omega\zeta_{n})}, \ n\in\widetilde{\mathcal{N}}^A.
\end{aligned}
\right.
\end{align}
Furthermore, we have
\begin{equation}\label{ressol}
\begin{aligned}
&u^{sol,N}(y,t)=\frac{\partial}{\partial t}\log{\frac{m_{2}^{R}(e^{\frac{\pi}{6}\mathrm{i}};\hat{\xi},t)}{m^{R}_1(e^{\frac{\pi}{6}\mathrm{i}};\hat{\xi},t)}},\\
&x^{sol,N}(y,t)=y+\log{\frac{m^{R}_{2}(e^{\frac{\pi}{6}\mathrm{i}};\hat{\xi},t)}{m^{R}_1(e^{\frac{\pi}{6}\mathrm{i}};\hat{\xi},t)}}.
\end{aligned}
\end{equation}
\end{lemma}
\begin{proof}
Define
\begin{equation}
    W(k)=\mathrm{diag}\left(\frac{\Pi(k)\Pi^A(k)}{\Pi(\omega^2k)\Pi^A(\omega^2k)},\frac{\Pi(\omega k)\Pi^A(\omega k)}{\Pi(k)\Pi^A(k)}, \frac{\Pi(\omega^2k)\Pi^A(\omega^2k)}{\Pi(\omega k)\Pi^A(\omega k)} \right)
\end{equation}
and make the transformation
\begin{align}\label{tM}
\tilde{M}(k)=M^{R}(k)T^{-1}(k)G^{-1}(k)T(k)W(k).
\end{align}

Evidently, the transformation to $\tilde{M}(k)$ preserves the normalization conditions at the origin and infinity. Comparing \eqref{tM} to \eqref{V21}, it is clear that the new unknown $\tilde{M}(k)$ has no jumps. Utilizing \eqref{T}, RH problem \ref{Mrhp} and \eqref{tM}, one can deduce that $\tilde{M}(k)$ has simple poles at each of the points in $\mathcal{K}$, the discrete spectrum originating from the primary RH problem \ref{rhpM}. A direct calculation shows that the residues of $\tilde{M}(k)$ comply with \eqref{resM11} and \eqref{res22}, but with $C_n$ being substituted by \eqref{tC}. Thus, $\tilde{M}(k)$ is precisely the solution of RH problem \ref{rhpM} with scattering data $\left\lbrace  0, \ \{\zeta_n,\tilde{C}_n\}_{n=1}^{6N}\right\rbrace$, whose existence and uniqueness can be derived through procedures akin to those delineated in Appendix A \cite{CJ}.
\end{proof}
Define the union set of neighborhood around phase points
\begin{align}
&U: =U(\hat{\xi})=\mathop{\cup}\limits_{i=i}^{p(\hat{\xi})}\left( U_{k_i}\cup \omega U_{k_i}\cup \omega^2U_{k_i}\right),\\
&\omega^{\ell}U_{k_i}= \left\lbrace k: |k-\omega^{\ell}k_i|\leq \varrho^{0} \right\rbrace, \ \ell=0,1,2
\end{align}
with
\begin{equation*}
\varrho^{0}=\frac{1}{8}\underset{i,j=1,\cdots,p(\xi)}\min\left\{\varrho,\ \underset{i\neq j}\min|k_i-k_j|,\right\}.
\end{equation*}

From \eqref{V222}, we derive the subsequent estimate
\begin{align}\label{estV2O}
\Vert V^{(2)}(k)-I\Vert_{L^q(\Sigma^{(ju)}\setminus U)}= \mathcal{O}( e^{-K_qt}), \ t\to\infty,
\end{align}
for $1\leq q\leq+\infty$ and a positive constant $K_q$. This suggests that the jump matrix $V^{(2)}(k)$ converges uniformly to the identity matrix $I$ on $\Sigma^{(ju)}\setminus U$.
As a result, outside of $U$, there exists only an exponentially small error (in $t$) by completely
ignoring the jump condition in RH problem \ref{Mrhp}, prompting us to construct  $M^{R}(k)$ as follows
\begin{align}\label{desM2RHP}
M^{R}(k)=\left\{
\begin{aligned}
    &E(k)M^{r}(k), \quad k\notin U,\\
    &E(k)M^{r}(k)M^{lo}(k), \quad k\in U,
\end{aligned}
\right.
\end{align}
where $M^{r}(k)$ satisfying the following RH problem.
\begin{RHP}\label{Mr}
Find a  matrix valued function  $M^{(r)}(k):= M^{(r)}(k;y,t)$ such that
\begin{itemize}
\item $M^{r}(k)$ is analytic in $\mathbb{C} \setminus \{ \zeta_n\}_{n-k_0N\in\Lambda}$.

\item $M^{r}(k) = I+\mathcal{O}(k^{-1}),\quad k \rightarrow \infty.$
\item $M^{r}(k)$ satisfies the jump relation
\begin{equation}
	M^{r}_+(k)=M^{r}_-(k)V^{r}(k),\quad k \in \Sigma^{(ci)},
\end{equation}
and
\begin{align}\label{V22out}
V^{r}(k)=\left\{
\begin{aligned}
    &T^{-1}(k)G(k)T(k), \quad k\in\partial\mathbb{D}_n\cap D_{2\nu-1},\\
    &T^{-1}(k)G^{-1}(k)T(k), \quad k\in\partial\mathbb{D}_n\cap D_{2\nu}.
\end{aligned}
\right.
\end{align}
\item $M^{r}(k)$ has simple poles at $\zeta_n, \ n-k_0N\in\Lambda$  satisfying the residue relations as \eqref{resM22} with  $M^{r}(k)$ replacing  $m^{(2)}(k)$.
\end{itemize}	
\end{RHP}
\begin{remark}
For $\hat{\xi}\in\big(-\infty,-\frac{3}{8}\big)\cup(3,+\infty)$,  the phase point is absent. Consequently, the set $U=\emptyset$ in these instances, signifying that
$M^{R}(k)=M^{r}(k)$ holds. This results in a simplification of the decomposition \eqref{desM2RHP}.
\end{remark}

\section{Contribution from discrete spectrum}\label{5}

In this section, we formulate two RH problems concerning $M^{r}(k)$ and $M^{\Lambda}(k)$ in the discrete spectrum
context of reflection, and subsequently demonstrate that $M^{r}(k)$ can be closely modeled by $M^{\Lambda}(k)$.

\subsection{$M^{r}(k)$ and $M^{\Lambda}(k)$-solitons}
For $1< q<+\infty$, the jump matrix $V^{(2)}(k)$ satisfies
\begin{align}\label{estV21}
\Vert V^{(2)}(k)-I\Vert_{L^q(\Sigma^{(ci)})}=\mathcal{\mathcal{O}}(e^{- \min\{\rho_0,\delta_0\}t} ),
\end{align}
indicating that the jump matrices on $\Sigma^{(ci)}$ do not contribute to the asymptotic behavior of the solution. Rather, the principal impact on $M^{R}(k)$ stems from the discrete spectrum set $\mathcal{K}$. Let $V^{(2)}(k)\equiv 0$, RH problem \ref{Mr} reduces to the following RH problem.

\begin{RHP}\label{rhpMout1}
Find a $3\times 3$ matrix-valued function  $M^{\Lambda}(k):= M^{\Lambda}(k;y,t)$ such that
\begin{itemize}
\item $ M^{\Lambda}(k)$ is analytic in $\mathbb{C} \setminus \{ \zeta_n\}_{n-k_0N\in\Lambda}$.
\item $ M^{\Lambda}(k) = I+\mathcal{O}(k^{-1}),\hspace{0.5cm}k \rightarrow \infty.$
\item $ M^{\Lambda}(k)$ has the same form  residue condition as  $M^{R}(k)$.
\end{itemize}	
\end{RHP}
The solvability of this RH problem is given in the following lemma.
\begin{lemma}\label{uxlam}
The RH problem \ref{rhpMout1} admits an unique solution. Furthermore, we can express it as follows
\begin{equation}\label{reulam}
\begin{aligned}
&u^{\Lambda}(y,t)=\frac{\partial}{\partial t}\log{\frac{m^{\Lambda}_{2}(e^{\frac{\pi}{6}\mathrm{i}};\hat{\xi},t)}{m^{\Lambda}_1(e^{\frac{\pi}{6}\mathrm{i}};\hat{\xi},t)}},\\
&x^{\Lambda}(y,t)=y+\log{\frac{m^{\Lambda}_{2}(e^{\frac{\pi}{6}\mathrm{i}};\hat{\xi},t)}{m^{\Lambda}_1(e^{\frac{\pi}{6}\mathrm{i}};\hat{\xi},t)}},
\end{aligned}
\end{equation}
where
\begin{equation*}
m_i^{\Lambda}(e^{\frac{\pi}{6}\mathrm{i}};\hat{\xi},t)=\mathop{\sum}_{j=1}^3 M_{ij}^{\Lambda}(e^{\frac{\pi}{6}\mathrm{i}};\hat{\xi},t), \quad i=1,2.
\end{equation*}
For $\Lambda\neq\emptyset$, i.e. $\Lambda={j_0}$, we define
\begin{equation}\label{1U}
u^{\Lambda}(x,t)=\mathcal{U}^{sol}(\zeta_{j_0},x,t).
\end{equation}
\end{lemma}
\begin{proof}
The uniqueness of $M^{\Lambda}(k)$ can be guaranteed by Liouville’s theorem. As for the expression to $M^{\Lambda}(k)$, in the case where $\Lambda=\emptyset$, all the $\zeta_n$ are away from the critical line, resulting in $M^{\Lambda}(k)=I$.

However, in the case where $\Lambda\neq\emptyset$, i.e. $\Lambda={j_0}$, we can rewrite the residue condition as
\begin{equation}\label{resfor}
\res\limits_{k=\zeta_{j_0}} M(k)=\lim_{k\rightarrow \zeta_n}M(k)\big(T^{-1}(k)B_nT(k)\big):=\begin{pmatrix} 0 & \alpha & 0 \\ 0 & \beta & 0 \\ 0 & \gamma & 0 \end{pmatrix}.
\end{equation}
By Plemelj formula, it follows that
\begin{align}\label{mlamex1}
M^{\Lambda}(k)&=I+\frac{1}{k-\zeta_{j_0}}\begin{pmatrix} 0 & \alpha & 0 \\ 0 & \beta & 0 \\ 0 &  \gamma & 0 \end{pmatrix}+\frac{1}{k-\omega\zeta_{j_0}}\begin{pmatrix}   \omega\beta & 0  & 0\\ \omega\gamma & 0 & 0 \\ \omega\alpha & 0 &  0 \end{pmatrix}+\frac{1}{k-\omega\bar{\zeta}_{j_0}}\begin{pmatrix}  0 & 0 & \omega\bar\alpha  \\   0 &  0 & \omega\bar\gamma \\ 0  & 0\ & \omega\bar\beta  \end{pmatrix}\nonumber\\
&+\frac{1}{k-\omega^2\bar\zeta_{j_0}}\begin{pmatrix}  0 & \omega^2\bar\gamma & 0 \\  0 &  \omega^2\bar\beta & 0 \\ 0 & \omega^2\bar\alpha &  0 \end{pmatrix}+\frac{1}{k-\omega^2\zeta_{j_0}}\begin{pmatrix}  0  & 0 & \omega^2\gamma \\  0  & 0 &  \omega^2\alpha\\ 0  &  0 & \omega^2\beta\end{pmatrix}+\frac{1}{k-\bar{\zeta}_{j_0}}\begin{pmatrix}   \bar\beta & 0  & 0\\ \bar\alpha  & 0 & 0 \\ \bar\gamma & 0 &  0 \end{pmatrix}.
\end{align}
Then the symmetry \eqref{S1} inherited by $M^{\Lambda}(k)$ indicates that
\begin{align}\label{sys1}
\left\{
\begin{aligned}
    &-\frac{1}{2{\rm Im} \ \zeta_{j_0}}(\beta-\bar\beta)=\frac{\omega}{(1-\omega)\zeta_{j_0}}\beta+\frac{\omega^2}{(1-\omega^2)\bar\zeta_{j_0}}\bar\beta,\\
    &-\frac{1}{2{\rm Im} \ \zeta_{j_0}}(\gamma-\bar\gamma)=\frac{\omega}{(1-\omega)\zeta_{j_0}}\alpha+\frac{\omega^2}{(1-\omega^2)\bar\zeta_{j_0}}\bar\alpha,\\
    &-\frac{1}{2{\rm Im} \ \zeta_{j_0}}(\alpha-\bar\alpha)=\frac{\omega}{(1-\omega)\zeta_{j_0}}\gamma+\frac{\omega^2}{(1-\omega^2)\bar\zeta_{j_0}}\bar\gamma.
\end{aligned}
\right.
\end{align}
Substituting \eqref{mlamex1} into the residue condition \eqref{resfor}, we obtain
\begin{align}\label{rela11}
\begin{pmatrix} 0 & \alpha & 0 \\ 0 & \beta & 0 \\ 0 &  \gamma & 0 \end{pmatrix}=\begin{pmatrix}    0  &  \Upsilon_0 & 0\\ 0 & 0 &  0 \\   0 & 0 & 0  \end{pmatrix}+\frac{1}{(1-\omega)\zeta_{j_0}}\begin{pmatrix}    0  & \omega\beta\Upsilon_0 & 0\\ 0 & \omega\gamma\Upsilon_0 &  0 \\   0 & \omega\alpha\Upsilon_0 & 0  \end{pmatrix}+\frac{1}{2{\rm Im}\zeta_{j_0}}\begin{pmatrix}    0  &\bar\beta\Upsilon_0 & 0\\ 0 &  \alpha\Upsilon_0 &  0 \\   0 &  \gamma\Upsilon_0 & 0  \end{pmatrix}
\end{align}
with $\Upsilon_0=T_{21}(\zeta_{j_0})C_{j_0}$. From \eqref{rela11}, we derive
\begin{align}\label{sys2}
\left\{
\begin{aligned}
    &\alpha=\Upsilon_0+\frac{\omega \Upsilon_0}{(1-\omega)\zeta_{j_0}}\beta+\frac{\Upsilon_0}{2{\rm Im}\zeta_{j_0}}\bar\beta,\\
    &\beta=\frac{\omega \Upsilon_0}{(1-\omega)\zeta_{j_0}}\gamma+\frac{\Upsilon_0}{2{\rm Im}\zeta_{j_0}}\bar\alpha,\\
    &\gamma=\frac{\omega \Upsilon_0}{(1-\omega)\zeta_{j_0}}\alpha+\frac{\Upsilon_0}{2{\rm Im}\zeta_{j_0}}\bar\gamma.
\end{aligned}
\right.
\end{align}
Upon solving the linear system \eqref{sys1} and \eqref{sys2},  one can derive the important parameters $\alpha,\ \beta, \ \gamma, \ \bar\alpha,\ \bar\beta, \ \bar\gamma $. With this accomplishment, we successfully establish the existence and uniqueness of the RH problem \ref{rhpMout1}.
\end{proof}

\subsection{Residual error between $M^{r}(k)$ and $M^{\Lambda}(k)$-solitons}\label{merr}

Now we show that $M^{\Lambda}(k)$ gives the leading order behavior to $M^{r}(k)$ for $t\gg 1$. Naturally, the error between $M^{r}(k)$ and $M^{\Lambda}(k)$ is given by
\begin{equation}
M^{err}(k)=M^{r}(k)M^{\Lambda}(k)^{-1},
\end{equation}
which satisfies the following RH problem.
\begin{RHP}\label{rhpmerr}
Find a $3\times 3$ matrix-valued function  $M^{err}(k):= M^{err}(k;y,t)$ such that
\begin{itemize}
\item $ M^{err}(k)$ is analytic in $\mathbb{C} \setminus \Sigma^{(ci)}$.
\item $ M^{err}(k) = I+\mathcal{O}(k^{-1}), \quad k \rightarrow \infty.$
\item $ M^{err}(k)$ satisfies the following jump condition
\begin{equation}
	M^{err}_+(k)=M^{err}_-(k)V^{err}(k), \quad k \in \Sigma^{(ci)},
\end{equation}
where
\begin{equation}\label{Verr}
	V^{err}(k)=M^{\Lambda}(k)V^{(2)}(k)M^{\Lambda}(k)^{-1}.
\end{equation}
\end{itemize}	
\end{RHP}

A simple calculation gives the following estimations
\begin{equation}\label{estV22}
\Vert V^{err}(k)-I\Vert_{L^q\left(\Sigma^{(ci)}\right)}=\mathcal{\mathcal{O}}\left(e^{- \min\{\rho_0,\delta_0\}t} \right).
\end{equation}

Therefore, the existence and uniqueness of the $M^{err}(k)$ are valid by using a small norm RH
problem \cite{RN10}. Moreover, according to Beals-Coifman theory \cite{BC1984}, the solution of $ M^{err}(k)$ can
be given in terms of the following integral
\begin{equation}
M^{err}(k)=I+ \frac{1}{2\pi \mathrm{i}}\int_{\Sigma^{(ci)}}\dfrac{ \eta(\varsigma)  (V^{err}(\varsigma)-I)}{\varsigma-k}d\varsigma,\label{tEz}
\end{equation}
where  $\eta(k)-I\in L^2(\partial\mathbb{D})$ is an unique solution of the Fredholm-type equation
\begin{equation}
(I-\mathcal{C}_{err})(\eta(k)-I)=\mathcal{C}_{err}I.
\end{equation}
And $\mathcal{C}_{err}$: $L^2 \to L^2 $ is an integral operator defined by $\mathcal{C}_{err}(\eta) =\mathcal{C}^-\left( \eta (V^{err}-I)\right)$ with the Cauchy projection operator
$\mathcal{C}^-$ on $\Sigma^{(ci)}$. Then it follows that
\begin{equation}
\Vert \mathcal{C}_{err}\Vert_{L^2 } \leq \Vert \mathcal{C}^- \Vert_{L^2 } \Vert V^{err}(k)-I\Vert_{L^\infty} \lesssim e^{- \min\{\rho_0,\delta_0\}t} ,
\end{equation}
which means $\Vert \mathcal{C}_{err}\Vert_{L^2 }<1$ for sufficiently large $t$, so   $\eta(k)$ uniquely  exists and
\begin{equation}
\Vert \eta(k)-I\Vert_{L^2 } \lesssim e^{- \min\{\rho_0,\delta_0\}t}.\label{normeta}
\end{equation}
In order to reconstruct the solution $u(x,t)$ of \eqref{DP}--\eqref{intva}, we need to evaluate the asymptotic behavior of $M^{err}(k)$ as $k\rightarrow e^{\frac{\pi}{6}\mathrm{i}}$. As $k\rightarrow e^{\frac{\pi}{6}\mathrm{i}}$, it is deduced that
\begin{align}
M^{err}(k)=M^{err}(e^{\frac{\pi}{6}\mathrm{i}})+M^{err}_1(k-e^{\frac{\pi}{6}\mathrm{i}})+\mathcal{O}\Big( (k-e^{\frac{\pi}{6}\mathrm{i}})^2\Big),\label{texpE}
\end{align}
where
\begin{align}
&M^{err}(e^{\frac{\pi}{6}\mathrm{i}})=I+\frac{1}{2\pi\mathrm{i}}\int_{\Sigma^{(ci)}}\dfrac{ \eta(\varsigma)  \Big(V^{err}(\varsigma)-I\Big)}{\varsigma-e^{\frac{\pi}{6}\mathrm{i}}}d\varsigma,\label{tEi}\\	
&M^{err}_1=\frac{1}{2\pi\mathrm{i}}\int_{\Sigma^{(ci)}}\frac{ \eta(\varsigma)  \Big(V^{err}(\varsigma)-I\Big)}{\Big(\varsigma-e^{\frac{\pi}{6}\mathrm{i}}\Big)^2}d\varsigma.
\end{align}
\begin{lemma}\label{tasyE}
As $t \to \infty$, we can derive the following estimation
\begin{align}
&\Big|M^{err}(k)-I\Big|\lesssim e^{- \min\{\rho_0,\delta_0\}t}, \\
&\Big|M^{err}(e^{\frac{\pi}{6}\mathrm{i}})-I\Big|\lesssim e^{- \min\{\rho_0,\delta_0\}t}, \quad \Big\lvert M^{err}_1\Big\rvert\lesssim e^{- \min\{\rho_0,\delta_0\}t}.\label{tE1t}
\end{align}
\end{lemma}
\begin{proof}
From \eqref{estV22}--\eqref{normeta}, it appears that
\begin{align}\label{merr-I}
\lvert M^{err}(k)-I\rvert&= \frac{1}{2\pi \mathrm{i}}\int_{\Sigma^{(ci)}}\dfrac{\lvert V^{err}(\varsigma)-I\rvert}{\lvert\varsigma-k\rvert}d\varsigma+\frac{1}{2\pi \mathrm{i}}\int_{\Sigma^{(ci)}}\dfrac{ \lvert\eta(\varsigma)-I\rvert\lvert V^{err}(\varsigma)-I\rvert}{\rvert\varsigma-k\lvert}d\varsigma,\nonumber\\
&\leq\Vert V^{err}-I\Vert_{L^2}\left\lVert \frac{1}{\varsigma-k}\right\rVert_{L^2}+\Vert V^{err}-I\Vert_{L^\infty}\Vert \eta-I\Vert_{L^2}\left\lVert \frac{1}{\varsigma-k}\right\rVert_{L^2}\nonumber\\
&\lesssim e^{- \min\{\rho_0,\delta_0\}t}.
\end{align}
Let $k=e^{\frac{\pi}{6}\mathrm{i}}$, we obtain the first estimate in \eqref{tE1t}. Noting that  $|s-e^{\frac{\pi}{6}\mathrm{i}}|^{-2}$ is bounded on $\Sigma^{(ci)}$, we obtain the second estimate in \eqref{tE1t}.
\end{proof}
\begin{proposition} \label{mout}
As $t \to \infty$, the relation can be expressed as
\begin{align}
&M^{r}(k)=M^{\Lambda}(k)\left[I+\mathcal{O}(e^{- \min\{\rho_0,\delta_0\}t})\right], \label{asymsol}\\
&u^{sol,N}(y,t)=u^{\Lambda}(y,t)\left[I+\mathcal{O}(e^{- \min\{\rho_0,\delta_0\}t})\right],\label{uuLam}
\end{align}
where $u^{sol,N}(y,t)$ and $u^{\Lambda}(y,t)$ are defined in \eqref{ressol} and \eqref{reulam}, respectively.
\end{proposition}

\section{Contribution from jump contours}\label{6}
\subsection{Local model near phase points}

We solve the local model near phase points. Denote some new contours
\begin{align}
\Sigma^{lo}=\mathop{\cup}\limits_{\ell=0,1,2}\Sigma^{\ell}, \quad
\Sigma^{\ell}=\left(\mathop{\cup}\limits_{j=1}^6\omega^{\ell}\Sigma_j\right)\cap U,
\end{align}
see Figure \ref{Sigmalo}. We consider the following RH problem:
\begin{RHP}\label{rhpMlo}
Find a matrix-valued function $M^{lo}(k):=M^{lo}(k;y,t)$ such that
\begin{itemize}
\item $M^{lo}(k)$ is analytic in $U\backslash\Sigma^{lo}$.
\item $M^{lo}(k)$ satisfies the jump relation
\begin{equation}
M^{lo}_+(k)=M^{lo}_-(k)V^{lo}(k),\quad k \in \Sigma^{lo},\label{jumlo}
\end{equation}
where $ V^{lo}(k)= V^{(2)}(k)\big|_{\Sigma^{lo}}$.
\item Asymptotic behaviors: $M^{lo}(k)M^{pc}(\zeta(k))^{-1}\rightarrow I, \quad k\in\partial U$.
\end{itemize}
\end{RHP}
    The RH problem \ref{rhpMlo} has jump relations but no poles, with   the Beals-Coifman theory, we can obtain  the   $M^{lo}(k)$ by the sum of  all  local model   $M^{lo}_{i,\ell}(k)$,  where $M^{lo}_{i,\ell}(k)$ is the local RH problem at phase point $\omega^{\ell} k_i$ with jump $V^{lo}_{i,\ell}(k)$, and its solution can be constructed with  parabolic cylinder equation. $V^{lo}_{i,\ell}(k)$ admits a factorization
\begin{equation}
   V^{lo}_{i,\ell}(k)=\left(I-w_{i-}^{\ell}\right)^{-1}\left(I+w_{i+}^{\ell}\right),
\end{equation}
\begin{equation}\label{wjk}
    w_{i-}^{\ell}=V^{lo}_{i\ell}(k)-I, \quad w_{i+}^{\ell}=0,
\end{equation}
and the superscript $\pm$ indicate the analyticity in the positive/negative neighborhood of the contour.
\begin{figure}
\begin{center}
	\subfigure[]{
\begin{tikzpicture}[scale=0.8]
\draw [thick]  (-4,-0.6) to [out=0,in=180] (-2,0.6)
to [out=0,in=180] (0,-0.6) to [out=0,in=180] (2,0.6)  to  [out=0,in=180] (4,-0.6);
\draw [thick,rotate=240] (-4,-0.6) to [out=0,in=180] (-2,0.6)
to [out=0,in=180] (0,-0.6) to [out=0,in=180] (2,0.6)  to  [out=0,in=180] (4,-0.6);
\draw [thick,rotate=120] (-4,-0.6) to [out=0,in=180] (-2,0.6)
to [out=0,in=180] (0,-0.6) to [out=0,in=180] (2,0.6)  to  [out=0,in=180] (4,-0.6);

\draw [thick] (-4,0.6) to [out=0,in=180] (-2,-0.6)
to [out=0,in=180] (0,0.6) to [out=0,in=180] (2,-0.6)  to [out=0,in=180] (4,0.6);
\draw [thick,rotate=240](-4,0.6) to [out=0,in=180] (-2,-0.6)
to [out=0,in=180] (0,0.6) to [out=0,in=180] (2,-0.6)  to [out=0,in=180] (4,0.6);
\draw [thick,rotate=120](-4,0.6) to [out=0,in=180] (-2,-0.6)
to [out=0,in=180] (0,0.6) to [out=0,in=180] (2,-0.6)  to [out=0,in=180] (4,0.6);

\coordinate (A) at (1.6,0);
\fill[red] (A) circle [radius=0.03] node[below] {\tiny$1$};
\coordinate (B) at (-1.6,0);
\fill[red] (B) circle [radius=0.03] node[below] {\tiny$-1$};

\draw [dotted] (0,0) circle [radius=1];
\draw [dotted] (0,0) circle [radius=3];

\filldraw [white] (0,0) circle [radius=0.8];
\filldraw[white,even odd rule] (0,0) circle (2.8)
                       (0,0) circle (1.3);
\filldraw[white,even odd rule] (0,0) circle (4.5)
                       (0,0) circle (3.2);

\draw[dotted](-5,0)--(5,0);
\draw[dotted,rotate=240](-5,0)--(5,0);
\draw[dotted,rotate=120](-5,0)--(5,0);

\filldraw  (1,0) circle [radius=0.05];
\node [below] at  (1,-0.1) {\tiny$k_2$};
\filldraw  (3,0) circle [radius=0.05];
\node [below] at  (3,-0.1) {\tiny$k_1$};
\filldraw  (-1,0) circle [radius=0.05];
\node [below] at  (-1,-0.1) {\tiny$k_3$};
\filldraw  (-3,0) circle [radius=0.05];
\node [below] at  (-3,-0.1) {\tiny$k_4$};

\coordinate (I) at (0,0);
\fill[red] (I) circle [radius=0.03] node[below] {\tiny$O$};
\coordinate (A) at (1.6,0);
\fill[red] (A) circle [radius=0.03] node[below] {\tiny$1$};
\coordinate (B) at (-1.6,0);
\fill[red] (B) circle [radius=0.03] node[below] {\tiny$-1$};

\filldraw  (-0.5,0.87) circle [radius=0.05];
\node [above] at   (-0.85,0.9) {\tiny$\omega k_2$};
\filldraw  (-1.5,2.6) circle [radius=0.05];
\node [above] at  (-1.75,2.7) {\tiny$\omega k_1$};
\filldraw  (0.5,-0.87) circle [radius=0.05];
\node [below] at   (0.9,-0.9) {\tiny$\omega k_3$};
\filldraw  (1.5,-2.6) circle [radius=0.05];
\node [below] at  (1.85,-2.7) {\tiny$\omega k_4$};

\filldraw  (-0.5,-0.87) circle [radius=0.05];
\node [below] at   (-0.85,-0.9) {\tiny$\omega^2 k_2$};
\filldraw  (-1.5,-2.6) circle [radius=0.05];
\node [below] at  (-1.75,-2.7) {\tiny$\omega^2 k_1$};
\filldraw  (0.5,0.87) circle [radius=0.05];
\node [above] at   (0.9,0.9) {\tiny$\omega^2 k_3$};
\filldraw  (1.5,2.6) circle [radius=0.05];
\node [above] at  (1.85,2.7) {\tiny$\omega^2 k_4$};
		
\end{tikzpicture}
		\label{case1}}
	\subfigure[]{
\begin{tikzpicture}[scale=0.8]
\draw [thick] (-6,-0.6)to [out=0,in=180](-4.5,0.6)to [out=0,in=180](-3,-0.6) to [out=0,in=180] (-1.5,0.6)
to [out=0,in=180] (0,-0.6) to [out=0,in=180] (1.5,0.6)  to  [out=0,in=180] (3,-0.6) to [out=0,in=180] (4.5,0.6) to
[out=0,in=180] (6,-0.6);
\draw [thick](-6,0.6)to [out=0,in=180](-4.5,-0.6)to [out=0,in=180](-3,0.6) to [out=0,in=180] (-1.5,-0.6)
to [out=0,in=180] (0,0.6) to [out=0,in=180] (1.5,-0.6)  to [out=0,in=180] (3,0.6) to [out=0,in=180] (4.5,-0.6) to  [out=0,in=180] (6,0.6);

\draw [thick,rotate=240] (-6,-0.6)to [out=0,in=180](-4.5,0.6)to [out=0,in=180](-3,-0.6) to [out=0,in=180] (-1.5,0.6)
to [out=0,in=180] (0,-0.6) to [out=0,in=180] (1.5,0.6)  to  [out=0,in=180] (3,-0.6) to [out=0,in=180] (4.5,0.6) to
[out=0,in=180] (6,-0.6);
\draw [thick,rotate=240](-6,0.6)to [out=0,in=180](-4.5,-0.6)to [out=0,in=180](-3,0.6) to [out=0,in=180] (-1.5,-0.6)
to [out=0,in=180] (0,0.6) to [out=0,in=180] (1.5,-0.6)  to [out=0,in=180] (3,0.6) to [out=0,in=180] (4.5,-0.6) to  [out=0,in=180] (6,0.6);

\draw [thick,rotate=120] (-6,-0.6)to [out=0,in=180](-4.5,0.6)to [out=0,in=180](-3,-0.6) to [out=0,in=180] (-1.5,0.6)
to [out=0,in=180] (0,-0.6) to [out=0,in=180] (1.5,0.6)  to  [out=0,in=180] (3,-0.6) to [out=0,in=180] (4.5,0.6) to
[out=0,in=180] (6,-0.6);
\draw [thick,rotate=120](-6,0.6)to [out=0,in=180](-4.5,-0.6)to [out=0,in=180](-3,0.6) to [out=0,in=180] (-1.5,-0.6)
to [out=0,in=180] (0,0.6) to [out=0,in=180] (1.5,-0.6)  to [out=0,in=180] (3,0.6) to [out=0,in=180] (4.5,-0.6) to  [out=0,in=180] (6,0.6);

\filldraw [white] (0,0) circle [radius=0.64];
\filldraw[white,even odd rule] (0,0) circle (2.05)
                        (0,0) circle (0.95);
\filldraw[white,even odd rule] (0,0) circle (3.55)
                        (0,0) circle (2.45);
\filldraw[white,even odd rule] (0,0) circle (5.05)
                        (0,0) circle (3.95);
\filldraw[white,even odd rule] (0,0) circle (6.2)
                        (0,0) circle (5.45);

\draw[dotted](-6.5,0)--(6.8,0);
\draw[dotted,rotate=240](-6.5,0)--(6.8,0);
\draw[dotted,rotate=120](-6.5,0)--(6.8,0);

\draw [dotted] (0,0) circle [radius=0.75];
\draw [dotted] (0,0) circle [radius=2.25];
\draw [dotted] (0,0) circle [radius=3.75];
\draw [dotted] (0,0) circle [radius=5.25];

\coordinate (I) at (0,0);
		\fill[red] (I) circle (1pt) node[below] {\tiny$O$};
\filldraw [red] (2.7,0) circle [radius=0.03];
\node [below,red] at  (2.8,0) {\tiny$ 1$};
\filldraw [red] (-2.7,0) circle [radius=0.03];
\node [below,red] at  (-2.8,0) {\tiny$ -1$};

\filldraw  (3.75,0) circle [radius=0.05];
\node [below] at  (3.75,-0.1) {\tiny$ k_2$};
\filldraw  (5.25,0) circle [radius=0.05];
\node [below] at   (5.25,-0.1) {\tiny$ k_1$};
\filldraw  (2.25,0) circle [radius=0.05];
\node [below] at   (2.25,-0.1) {\tiny$ k_3$};
\filldraw  (0.75,0) circle [radius=0.05];
\node [below] at   (0.75,-0.1) {\tiny$ k_4$};
\filldraw  (-3.75,0) circle [radius=0.05];
\node [below] at   (-3.75,-0.1) {\tiny$ k_7$};
\filldraw  (-5.25,0) circle [radius=0.05];
\node [below] at   (-5.25,-0.1) {\tiny$ k_8$};
\filldraw  (-2.25,0) circle [radius=0.05];
\node [below] at   (-2.25,-0.1) {\tiny$ k_6$};
\filldraw  (-0.75,0) circle [radius=0.05];
\node [below] at   (-0.75,-0.1) {\tiny$k_5$};

\filldraw  (-1.88,3.25) circle [radius=0.05];
\node [above] at    (-2.05,3.25) {\tiny$\omega k_2$};
\filldraw  (-2.63,4.55) circle [radius=0.05];
\node [above] at   (-2.8,4.55) {\tiny$\omega k_1$};
\filldraw  (-1.13,1.95) circle [radius=0.05];
\node [above] at    (-1.35,1.95) {\tiny$\omega k_3$};
\filldraw  (-0.38,0.65) circle [radius=0.05];
\node [above] at  (-0.6,0.65) {\tiny$\omega k_4$};
\filldraw  (1.88,-3.25) circle [radius=0.05];
\node [below] at    (2.15,-3.25) {\tiny$\omega k_7$};
\filldraw  (2.63,-4.55) circle [radius=0.05];
\node [below] at   (2.8,-4.55) {\tiny$\omega k_8$};
\filldraw  (1.13,-1.95) circle [radius=0.05];
\node [below] at    (1.35,-1.95) {\tiny$\omega k_6$};
\filldraw  (0.38,-0.65) circle [radius=0.05];
\node [below] at  (0.6,-0.65) {\tiny$\omega k_5$};

\filldraw  (-1.88,-3.25) circle [radius=0.05];
\node [below] at    (-2.1,-3.25) {\tiny$\omega^2 k_2$};
\filldraw  (-2.63,-4.55) circle [radius=0.05];
\node [below] at   (-2.8,-4.55) {\tiny$\omega^2 k_1$};
\filldraw  (-1.13,-1.95) circle [radius=0.05];
\node [below] at    (-1.35,-1.95) {\tiny$\omega^2 k_3$};
\filldraw  (-0.38,-0.65) circle [radius=0.05];
\node [below] at  (-0.6,-0.65) {\tiny$\omega^2 k_4$};
\filldraw  (1.88,3.25) circle [radius=0.05];
\node [above] at    (2.2,3.25) {\tiny$\omega^2 k_7$};
\filldraw  (2.63,4.55) circle [radius=0.05];
\node [above] at   (2.9,4.55) {\tiny$\omega^2 k_8$};
\filldraw  (1.13,1.95) circle [radius=0.05];
\node [above] at    (1.45,1.95) {\tiny$\omega^2 k_6$};
\filldraw  (0.38,0.65) circle [radius=0.05];
\node [above] at  (0.7,0.65) {\tiny$\omega^2 k_5$};
		\end{tikzpicture}
		\label{case2}}
\caption{\footnotesize Figure (a) is the local jump contour $\Sigma^{lo}$ consisting of  12  crosses  for the case
   $\hat{\xi}\in\left[0,3\right)$; Figure (b) is the jump contour $\Sigma^{lo}$ consisting of 24 crosses for the case $\hat{\xi}\in(-\frac{3}{8},0)$.}
	\label{Sigmalo}
\end{center}
\end{figure}
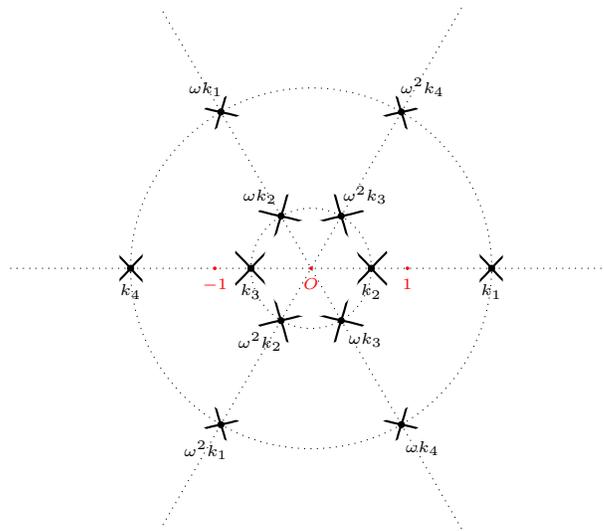
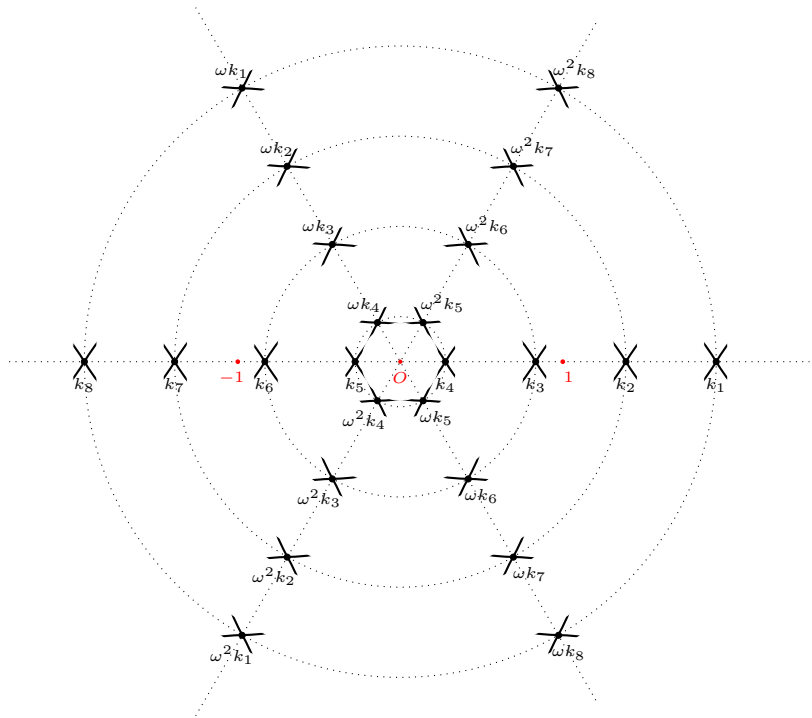
Recall the Cauchy projection operator, we can define the following Beals-Coifman operator
\begin{equation}
    \mathcal{C}_{w_i^{\ell}}(f):=\mathcal{C}_{+}(fw_{i-}^{\ell})+\mathcal{C}_{-}(fw_{i+}^{\ell}).
\end{equation}
Let
\begin{equation}
    w=\mathop{\sum}_{\ell=0,1,2}\Big(\mathop{\sum}\limits_{j=1}^{p(\hat{\xi})}w_i^{\ell}\Big), \quad  \mathcal{C}_w=\mathop{\sum}\limits_{\ell=0,1,2}\Big(\mathop{\sum}\limits_{j=1}^{p(\hat{\xi})}\mathcal{C}_{w_i^{\ell}}\Big).
\end{equation}
A simple calculation gives the following lemma.
\begin{lemma}\label{wjkest}
The matrix functions $w $ and $w_i^{\ell}$ defined above admit
    \begin{equation}
    \Vert w \Vert_{L^{2}(\Sigma^{lo})}=\mathcal{O}(t^{-\frac{1}{2}}), \quad \Vert w_i^{\ell} \Vert =\mathcal{O}(t^{-\frac{1}{2}}).
    \end{equation}
\end{lemma}
This lemma suggests the existence of $(1- \mathcal{C}_{w})^{-1}$ and  $(1- \mathcal{C}_{w_i^{\ell}})^{-1}$, thus indicating that the RH problem \ref{rhpMlo} exists an unique solution
\begin{equation}
   M^{lo}(k)=I+\frac{1}{2\pi \mathrm{i}}\int_{\Sigma^{lo}}\frac{(1- \mathcal{C}_w)^{-1}Iw}{s-k}ds.
\end{equation}
Following the established methodology outlined by Deift and Zhou \cite{RN6}, we proceed to derive the subsequent lemma.
\begin{lemma}\label{cwjcwk}
As $t\rightarrow +\infty$, for $i\neq j$, $\ell\neq m$, it holds that
\begin{equation}
\Vert \mathcal{C}_{w_i^{\ell}}\mathcal{C}_{w_j^{\ell}} \Vert_{L^{2}(\Sigma^{l})}=\mathcal{O}(t^{-1}), \quad \Vert \mathcal{C}_{w_i^{\ell}}\mathcal{C}_{w_i^m} \Vert_{ L^{2}(\Sigma^{l})}=\mathcal{O}(t^{-1}).
\end{equation}
\begin{align}
	\int_{\Sigma^{lo} }\frac{(I-\mathcal{C}_w)^{-1}I\ w}{\varsigma-k}d\varsigma
=\sum_{\ell=0,1,2}\left(\mathop{\sum}\limits_{j=1}^{p(\hat{\xi})}\int\frac{(I-\mathcal{C}_{w_i^{\ell}})^{-1}I\ w_i^{\ell}}{\varsigma-k}d\varsigma\right)+\mathcal{O}(t^{-\frac{3}{2}}).
	\end{align}
\end{lemma}

This Lemma infers that $M^{lo}(k)$ can be expressed as the summation of the separate contributions from $M^{lo}_{i,\ell}(k)$. As an illustrative example, we shall focus on the local model corresponding to phase point $k_1$.
\begin{RHP}\label{rhpMlo1}
Find a matrix-valued function $M^{lo}_{1,0}(k):=M^{lo}_{1,0}(k;y,t)$ such that
\begin{itemize}
\item $M^{lo}_{1,0}(k)$ is analytic in $U_{k_1}\backslash\big((\Sigma_1\cup\Sigma_2)\cap U_{k_1}\big)$.
\item $M^{lo}_{1,0}(k)$ satisfies the jump relation
\begin{equation}
\big(M^{lo}_{1,0}\big)_+(k)=\big(M^{lo}_{1,0}\big)_-(k)V^{lo}_{1,0}(k),\quad k \in (\Sigma_1\cup\Sigma_2)\cap U_{k_1},\label{jumlo1}
\end{equation}
where $ V^{lo}_{1,0}(k)= V^{(2)}(k)\big|_{(\Sigma_1\cup\Sigma_2)\cap U_{k_1}}$.
\item Asymptotic behaviors: $M^{lo}_{1,0}(k)M^{pc}_{1,0}(\zeta)^{-1}\rightarrow I, \quad k\in\partial U_{k_1}$.
\end{itemize}
\end{RHP}
In the above RH problem \ref{rhpMlo1}, $\zeta:=\zeta(k)$  denote the rescaled local variable
\begin{align}\label{resva}
\zeta(k)=t^{\frac{1}{2}}\sqrt{-4\eta(\hat{\xi},1)\theta''_{12}(k_1)}(k-k_1).
\end{align}
This change of variable maps $U_{k_1}$ to an neighborhood of $\zeta=0$. Additionally, let
\begin{align}
r_{k_1}=r(k_1)T_{12}^{(1)}(k_1)e^{-2\mathrm{i}t\theta(k_1)}\zeta^{-2\mathrm{i}\eta\nu(k_1)}\exp\left\lbrace -\mathrm{i}\eta\nu(k_1)\log \left( 4t\theta''_{12}(k_1)\tilde{\eta}(k_1)\right) \right\rbrace ,
\end{align}
with $\tilde{\eta}(\hat{\xi}, 1)=-1$ for $\hat{\xi}\in(-\frac{3}{8},0)$ and $\tilde{\eta}(\hat{\xi}, 1)=1$ for $ \hat{\xi}\in\left[0,3\right)$.

By transformation \eqref{resva}, the jump matrix $V^{lo}_{1,0}(k)$ approximates to the jump matrix of a parabolic cylinder model problem as follows.
Moreover, $M^{pc}_{1,0}\big(\zeta(k)\big)$ satisfies the following RH problem.
\begin{RHP}\label{rhpMpc1}
Find a matrix-valued function $M^{pc}_{1,0}(\zeta):=M^{pc}_{1,0}(\zeta;\hat{\xi},t)$ such that
\begin{itemize}
\item $M^{pc}_{1,0}(\zeta)$ is analytic in $U\backslash\Sigma_{1,0}^{pc}$ with $\Sigma_{1,0}^{pc}=\mathop{\cup}\limits_{j=1}^4\left(e^{\frac{(2j-1)\pi\mathrm{i}}{4}}\mathbb{R}_+\right)$.
\item $M^{pc}_{1,0}(\zeta)$ satisfies the jump relation
\begin{equation}
\big(M^{pc}_{1,0}\big)_+(\zeta)=\big(M^{pc}_{1,0}\big)_-(\zeta)V^{pc}_{1,0}(\zeta),\quad k \in \Sigma_{1,0}^{pc},\label{jumpc1}
\end{equation}
for $\hat{\xi}\in(-\frac{3}{8},0)$,
\begin{align}
V^{pc}_{1,0}(\zeta)=\left\{
\begin{aligned}
    &\begin{pmatrix} 1 & 0 & 0 \\ \frac{r_{k_1}}{1-|r_{k_1}|^2}\zeta^{2\mathrm{i}\nu(k_1)}e^{-\frac{\mathrm{i}}{2}\zeta^2} & 1 & 0 \\ 0 & 0 & 1 \end{pmatrix}, \quad \zeta\in e^{\frac{\pi\mathrm{i}}{4}}\mathbb{R}_+,\\
   &\begin{pmatrix} 1 & \bar{r}_{k_1}\zeta^{-2\mathrm{i}\nu(k_1)}e^{\frac{\mathrm{i}}{2}\zeta^2} & 0 \\ 0 & 1 & 0 \\ 0 & 0 & 1 \end{pmatrix}, \quad \zeta\in e^{\frac{3\pi\mathrm{i}}{4}}\mathbb{R}_+,\\
   &\begin{pmatrix} 1 & 0 & 0 \\ r_{k_1}\zeta^{2\mathrm{i}\nu(k_1)}e^{-\frac{\mathrm{i}}{2}\zeta^2} & 1 & 0 \\ 0 & 0 & 1 \end{pmatrix}, \quad \zeta\in e^{\frac{5\pi\mathrm{i}}{4}}\mathbb{R}_+,\\
   &\begin{pmatrix} 1 & \frac{\bar{r}_{k_1}}{1-|r_{k_1}|^2}\zeta^{-2\mathrm{i}\nu(k_1)}e^{\frac{\mathrm{i}}{2}\zeta^2} & 0 \\ 0 & 1 & 0 \\ 0 & 0 & 1 \end{pmatrix}, \quad \zeta\in e^{\frac{7\pi\mathrm{i}}{4}}\mathbb{R}_+,
\end{aligned}
\right.
\end{align}
and for $\hat{\xi}\in\left[0,3\right)$,
\begin{align}
V^{pc}_{1,0}(\zeta)=\left\{
\begin{aligned}
   &\begin{pmatrix} 1 & \bar{r}_{k_1}\zeta^{-2\mathrm{i}\nu(k_1)}e^{\frac{\mathrm{i}}{2}\zeta^2} & 0 \\ 0 & 1 & 0 \\ 0 & 0 & 1 \end{pmatrix}, \quad \zeta\in e^{\frac{\pi\mathrm{i}}{4}}\mathbb{R}_+,\\
   &\begin{pmatrix} 1 & 0 & 0 \\ \frac{r_{k_1}}{1-|r_{k_1}|^2}\zeta^{2\mathrm{i}\nu(k_1)}e^{-\frac{\mathrm{i}}{2}\zeta^2} & 1 & 0 \\ 0 & 0 & 1 \end{pmatrix}, \quad \zeta\in e^{\frac{3\pi\mathrm{i}}{4}}\mathbb{R}_+,\\
   &\begin{pmatrix} 1 & \frac{\bar{r}_{k_1}}{1-|r_{k_1}|^2}\zeta^{-2\mathrm{i}\nu(k_1)}e^{\frac{\mathrm{i}}{2}\zeta^2} & 0 \\ 0 & 1 & 0 \\ 0 & 0 & 1 \end{pmatrix}, \quad \zeta\in e^{\frac{5\pi\mathrm{i}}{4}}\mathbb{R}_+,\\
   &\begin{pmatrix} 1 & 0 & 0 \\ r_{k_1}\zeta^{2\mathrm{i}\nu(k_1)}e^{-\frac{\mathrm{i}}{2}\zeta^2} & 1 & 0 \\ 0 & 0 & 1 \end{pmatrix}, \quad \zeta\in e^{\frac{7\pi\mathrm{i}}{4}}\mathbb{R}_+,
\end{aligned}
\right.
\end{align}
\item Asymptotic behaviors: $M^{pc}_{1,0}(\zeta)=I+(M^{pc}_{1,0})^{(1)}\zeta^{-1}+\mathcal{O}(\zeta^{-2}), \quad \zeta\rightarrow\infty$.
\end{itemize}
\end{RHP}

The RH problem \ref{rhpMpc1} has an explicit solution $M^{pc}_{1,0}(\zeta)$, which can be expressed in terms of solutions
of the parabolic cylinder equation \cite{YF2}. Further, through the application of the variable transformation \eqref{resva}, we are able to derive the subsequent outcomes.
\begin{lemma}
As $t\rightarrow \infty$, we have the following relation
\begin{equation}
M^{lo}_{1,0}(k)=M^{pc,}_{1,0}(\zeta)+\mathcal{O}(t^{-1}).
\end{equation}
\end{lemma}
\begin{proposition}\label{Mlo0i}
As $t\rightarrow +\infty$, $M^{lo}_{1,0}(k)$ can be expressed as
\begin{equation}
M^{l}_{1,0}(k)=I+\frac{t^{-\frac{1}{2}}}{2(k-k_1)\sqrt{\eta(\hat{\xi},1)\theta''_{12}(k_1)}}\begin{pmatrix} 0 & \tilde{\beta}_{12}^1 & 0 \\ \tilde{\beta}_{21}^1 & 1 & 0 \\ 0 & 0 & 1 \end{pmatrix}+\mathcal{O}(t^{-1}).
\end{equation}	
Similarly, for local models at other phase points $k_i, \ i=1,\dots,p(\hat{\xi}),$
\begin{equation}
M^{lo}_{i,0}(k)=I+\frac{t^{-\frac{1}{2}}}{2(k-k_i)\sqrt{\eta(\hat{\xi},i)\theta''_{12}(k_i)}}\begin{pmatrix} 0 & \tilde{\beta}_{12}^i & 0 \\ \tilde{\beta}_{21}^i & 1 & 0 \\ 0 & 0 & 1 \end{pmatrix}+\mathcal{O}(t^{-1}),
\end{equation}
where
\begin{align}
\tilde{\beta}_{12}^i &=\left\{
\begin{aligned}
		&\frac{\sqrt{2\pi}}{\bar{r}_{k_i}\Gamma(\mathrm{i}\nu(k_i))}e^{\frac{\pi\nu(k_i)}{2}}e^{-\frac{\pi}{4}\mathrm{i}}, \quad
       \hat{\xi}\in(-\frac{3}{8},0),\\
		&-\frac{\sqrt{2\pi}}{\bar{r}_{k_i}\Gamma(-\mathrm{i}\nu(k_i))}e^{\frac{5\pi\nu(k_i)}{2}}e^{-\frac{7\pi}{4}\mathrm{i}},
       \quad  \hat{\xi}\in\left[0,3\right),\label{beta1}
\end{aligned}
\right.\\
\lvert\tilde{\beta}_{21}^i\rvert &=\left\{
\begin{aligned}
       &-\frac{\nu(k_i)}{1-|r_{k_i}|^2}, \quad \hat{\xi}\in(-\frac{3}{8},0),\\
		&\frac{\nu(k_i)}{(1-|r_{k_i}|^2)^3}, \quad \hat{\xi}\in\left[0,3\right),\label{beta2}\\
\end{aligned}
\right.\\
\arg{\tilde{\beta}_{21}^i}&=\left\{
\begin{aligned}
   &\frac{\pi}{2}\nu(k_i)-\frac{\pi}{4}\mathrm{i}-\arg(-\bar{r}_{k_i})-\arg{\Gamma(\mathrm{i})\nu(k_i)}, \quad \hat{\xi}\in(-\frac{3}{8},0),\\
   &\frac{5\pi}{2}\nu(k_i)-\frac{7\pi}{4}\mathrm{i}-\arg(-\bar{r}_{k_i})-\arg{\Gamma(-\mathrm{i})\nu(k_i)}, \quad \hat{\xi}\in\left[0,3\right).\label{beta3}
\end{aligned}
\right.
\end{align}

\end{proposition}\label{asyMlo1}
Combining Lemma \ref{cwjcwk} and Proposition \ref{Mlo0i}, along with the symmetry \eqref{S1}, gives rise to the subsequent proposition.
\begin{proposition}\label{asyMlo}
As $t\to+\infty$,
\begin{align}
M^{lo}(k)=&I+\frac{1}{2}t^{-\frac{1}{2}}\sum_{ i=1 }^{p(\xi)}\mathcal{F}_i(k) +\mathcal{O}(t^{-1}),
\end{align}
where
\begin{align}\label{Fi}
\mathcal{F}_i(k)=\frac{\mathcal{A}_i(\hat{\xi})}{\sqrt{|\theta_{12}''(k_i)|}(k-k_i)}+\frac{\omega\Gamma_3\overline{\mathcal{A}_i(\hat{\xi})}\Gamma_3}
{\sqrt{|\theta_{12}''(\omega k_i)|}(k-\omega k_i)}
+\frac{\omega^2\Gamma_2\overline{\mathcal{A}_i(\hat{\xi})}\Gamma_2}{\sqrt{|\theta_{12}''(\omega^2k_i)|}(k-\omega^2k_i)},
\end{align}
and
\begin{align}
\mathcal{A}_i(\hat{\xi})=\left(\begin{array}{ccc}
0 & \tilde{\beta}^i_{12} &0\\
\tilde{\beta}^i_{21} & 0 &0\\
0&0&0
\end{array}\right),
\end{align}
$\tilde{\beta}^i_{12}$, $\tilde{\beta}^i_{21}$ are difined in \eqref{beta1}--\eqref{beta3}.
\end{proposition}

\subsection{The small-norm RH problem for $E(k)$}

Consider the small-norm RH problem for $E(k)$, which is analytic in $\mathbb{C}\setminus \Sigma^{E}$, where
\begin{align}
\Sigma^{E}=\Big(\Sigma^{(2)}\setminus U\Big)\cup\partial U,
\end{align}
as illustrated in Figure\ref{FIGE3}. And, $E(k)$ satisfies the following RH problem.
\begin{RHP}\label{E}
Find a matrix-valued function $E(k):=E(k;y,t)$ such that
\begin{itemize}
\item $E(k)=I+\mathcal{O}(k^{-1}), \quad k\rightarrow \infty $.
\item Jump relations: $E_+(k)=E_-(k)V^{E}(k), \ k\in\Sigma^E$ with
\begin{align}\label{VE}
V^{E}(k)=\left\{
\begin{aligned}
   &M^{r}(k)V^{(2)}(k)M^{r}(k)^{-1}, \quad k\in\Sigma^{(2)}\setminus U,\\
   &M^{r}(k)M^{lo}(k)M^{r}(k)^{-1}, \quad k\in\partial U.
\end{aligned}
\right.
\end{align}
\end{itemize}
\end{RHP}

\begin{figure}
\begin{center}
	\subfigure[]{
\begin{tikzpicture}[scale=0.8]
\draw [thick]  (-4,-0.6) to [out=0,in=180] (-2,0.6)
to [out=0,in=180] (0,-0.6) to [out=0,in=180] (2,0.6)  to  [out=0,in=180] (4,-0.6);
\draw [thick,rotate=240] (-4,-0.6) to [out=0,in=180] (-2,0.6)
to [out=0,in=180] (0,-0.6) to [out=0,in=180] (2,0.6)  to  [out=0,in=180] (4,-0.6);
\draw [thick,rotate=120] (-4,-0.6) to [out=0,in=180] (-2,0.6)
to [out=0,in=180] (0,-0.6) to [out=0,in=180] (2,0.6)  to  [out=0,in=180] (4,-0.6);

\draw [thick] (-4,0.6) to [out=0,in=180] (-2,-0.6)
to [out=0,in=180] (0,0.6) to [out=0,in=180] (2,-0.6)  to [out=0,in=180] (4,0.6);
\draw [thick,rotate=240](-4,0.6) to [out=0,in=180] (-2,-0.6)
to [out=0,in=180] (0,0.6) to [out=0,in=180] (2,-0.6)  to [out=0,in=180] (4,0.6);
\draw [thick,rotate=120](-4,0.6) to [out=0,in=180] (-2,-0.6)
to [out=0,in=180] (0,0.6) to [out=0,in=180] (2,-0.6)  to [out=0,in=180] (4,0.6);

\filldraw [white] (0,0) circle [radius=1];
\draw[dotted](-5,0)--(5,0);
\draw[dotted,rotate=240](-5,0)--(5,0);
\draw[dotted,rotate=120](-5,0)--(5,0);

\node  at (4.3,0.6) {\scriptsize$\Sigma_1$};
\node  at (4.3,-0.6) {\scriptsize$\Sigma_2$};

\draw [thick] (0,0) circle [radius=1];

\coordinate (A) at (1.6,0);
\fill[red] (A) circle [radius=0.03] node[below] {\tiny$1$};
\coordinate (B) at (-1.6,0);
\fill[red] (B) circle [radius=0.03] node[below] {\tiny$-1$};

\draw [dotted] (0,0) circle [radius=1];
\draw [dotted] (0,0) circle [radius=3];

\filldraw [white] (1,0) circle [radius=0.2];
\draw  (1,0) circle [radius=0.2];
\filldraw  (1,0) circle [radius=0.05];
\node [below] at  (1.2,-0.1) {\tiny$k_2$};
\filldraw [white] (3,0) circle [radius=0.2];
\draw  (3,0) circle [radius=0.2];
\filldraw  (3,0) circle [radius=0.05];
\node [below] at  (3,-0.1) {\tiny$k_1$};
\filldraw [white] (-1,0) circle [radius=0.2];
\draw  (-1,0) circle [radius=0.2];
\filldraw  (-1,0) circle [radius=0.05];
\node [below] at  (-1.15,-0.1) {\tiny$k_3$};
\filldraw [white] (-3,0) circle [radius=0.2];
\draw  (-3,0) circle [radius=0.2];
\filldraw  (-3,0) circle [radius=0.05];
\node [below] at  (-3,-0.1) {\tiny$k_4$};

\coordinate (I) at (0,0);
\fill[red] (I) circle [radius=0.03] node[below] {\tiny$O$};

\filldraw [white] (-0.5,0.87) circle [radius=0.2];
\draw  (-0.5,0.87) circle [radius=0.2];
\filldraw  (-0.5,0.87) circle [radius=0.05];
\node [above] at   (-0.85,0.9) {\tiny$\omega k_2$};
\filldraw [white] (-1.5,2.6) circle [radius=0.2];
\draw (-1.5,2.6) circle [radius=0.2];
\filldraw  (-1.5,2.6) circle [radius=0.05];
\node [above] at  (-1.75,2.7) {\tiny$\omega k_1$};
\filldraw [white] (0.5,-0.87) circle [radius=0.2];
\draw (0.5,-0.87) circle [radius=0.2];
\filldraw  (0.5,-0.87) circle [radius=0.05];
\node [below] at   (0.9,-0.9) {\tiny$\omega k_3$};
\filldraw [white] (1.5,-2.6) circle [radius=0.2];
\draw (1.5,-2.6) circle [radius=0.2];
\filldraw  (1.5,-2.6) circle [radius=0.05];
\node [below] at  (1.85,-2.7) {\tiny$\omega k_4$};

\filldraw [white] (-0.5,-0.87) circle [radius=0.2];
\draw  (-0.5,-0.87) circle [radius=0.2];
\filldraw  (-0.5,-0.87) circle [radius=0.05];
\node [below] at   (-0.85,-0.9) {\tiny$\omega^2 k_2$};
\filldraw [white] (-1.5,-2.6) circle [radius=0.2];
\draw (-1.5,-2.6) circle [radius=0.2];
\filldraw  (-1.5,-2.6) circle [radius=0.05];
\node [below] at  (-1.75,-2.7) {\tiny$\omega^2 k_1$};
\filldraw [white] (0.5,0.87) circle [radius=0.2];
\draw (0.5,0.87) circle [radius=0.2];
\filldraw  (0.5,0.87) circle [radius=0.05];
\node [above] at   (0.9,0.9) {\tiny$\omega^2 k_3$};
\filldraw  (1.5,2.6) circle [radius=0.05];
\filldraw  (1.5,2.6) circle [radius=0.05];
\filldraw  (1.5,2.6) circle [radius=0.05];
\node [above] at  (1.85,2.7) {\tiny$\omega^2 k_4$};
		
\end{tikzpicture}
		\label{case1}}
	\subfigure[]{
\begin{tikzpicture}[scale=0.8]
\draw [thick] (-6,-0.6)to [out=0,in=180](-4.5,0.6)to [out=0,in=180](-3,-0.6) to [out=0,in=180] (-1.5,0.6)
to [out=0,in=180] (0,-0.6) to [out=0,in=180] (1.5,0.6)  to  [out=0,in=180] (3,-0.6) to [out=0,in=180] (4.5,0.6) to
[out=0,in=180] (6,-0.6);
\draw [thick](-6,0.6)to [out=0,in=180](-4.5,-0.6)to [out=0,in=180](-3,0.6) to [out=0,in=180] (-1.5,-0.6)
to [out=0,in=180] (0,0.6) to [out=0,in=180] (1.5,-0.6)  to [out=0,in=180] (3,0.6) to [out=0,in=180] (4.5,-0.6) to  [out=0,in=180] (6,0.6);

\draw [thick,rotate=240] (-6,-0.6)to [out=0,in=180](-4.5,0.6)to [out=0,in=180](-3,-0.6) to [out=0,in=180] (-1.5,0.6)
to [out=0,in=180] (0,-0.6) to [out=0,in=180] (1.5,0.6)  to  [out=0,in=180] (3,-0.6) to [out=0,in=180] (4.5,0.6) to
[out=0,in=180] (6,-0.6);
\draw [thick,rotate=240](-6,0.6)to [out=0,in=180](-4.5,-0.6)to [out=0,in=180](-3,0.6) to [out=0,in=180] (-1.5,-0.6)
to [out=0,in=180] (0,0.6) to [out=0,in=180] (1.5,-0.6)  to [out=0,in=180] (3,0.6) to [out=0,in=180] (4.5,-0.6) to  [out=0,in=180] (6,0.6);

\draw [thick,rotate=120] (-6,-0.6)to [out=0,in=180](-4.5,0.6)to [out=0,in=180](-3,-0.6) to [out=0,in=180] (-1.5,0.6)
to [out=0,in=180] (0,-0.6) to [out=0,in=180] (1.5,0.6)  to  [out=0,in=180] (3,-0.6) to [out=0,in=180] (4.5,0.6) to
[out=0,in=180] (6,-0.6);
\draw [thick,rotate=120](-6,0.6)to [out=0,in=180](-4.5,-0.6)to [out=0,in=180](-3,0.6) to [out=0,in=180] (-1.5,-0.6)
to [out=0,in=180] (0,0.6) to [out=0,in=180] (1.5,-0.6)  to [out=0,in=180] (3,0.6) to [out=0,in=180] (4.5,-0.6) to  [out=0,in=180] (6,0.6);

\draw [dotted] (0,0) circle [radius=0.75];
\draw [dotted] (0,0) circle [radius=2.25];
\draw [dotted] (0,0) circle [radius=3.75];
\draw [dotted] (0,0) circle [radius=5.25];

\filldraw [white] (0,0) circle [radius=0.75];
\draw[dotted](-6.5,0)--(6.8,0);
\draw[dotted,rotate=240](-6.5,0)--(6.8,0);
\draw[dotted,rotate=120](-6.5,0)--(6.8,0);

\draw [thick] (0,0) circle [radius=0.75];

\node  at (6.3,0.6) {\tiny$\Sigma_1$};
\node  at (6.3,-0.6) {\tiny$\Sigma_2$};
\node  at (-3.55,5.1) {\tiny$\Sigma_3$};
\node  at (-2.45,5.75) {\tiny$\Sigma_4$};
\node  at (-3.55,-5.1) {\tiny$\Sigma_5$};
\node  at (-2.45,-5.75) {\tiny$\Sigma_6$};

\coordinate (I) at (0,0);
		\fill[red] (I) circle (1pt) node[below] {\tiny$O$};
\filldraw [red] (2.7,0) circle [radius=0.03];
\node [below,red] at  (2.8,0) {\tiny$ 1$};
\filldraw [red] (-2.7,0) circle [radius=0.03];
\node [below,red] at  (-2.8,0) {\tiny$ -1$};

\filldraw [white] (3.75,0) circle [radius=0.2];
\draw  (3.75,0) circle [radius=0.2];
\filldraw  (3.75,0) circle [radius=0.05];
\node [below] at  (3.75,-0.1) {\tiny$ k_2$};
\filldraw [white] (5.25,0) circle [radius=0.2];
\draw  (5.25,0) circle [radius=0.2];
\filldraw  (5.25,0) circle [radius=0.05];
\node [below] at   (5.25,-0.1) {\tiny$ k_1$};
\filldraw [white] (2.25,0) circle [radius=0.2];
\draw  (2.25,0) circle [radius=0.2];
\filldraw  (2.25,0) circle [radius=0.05];
\node [below] at   (2.25,-0.1) {\tiny$ k_3$};
\filldraw [white] (0.75,0) circle [radius=0.2];
\draw  (0.75,0) circle [radius=0.2];
\filldraw  (0.75,0) circle [radius=0.05];
\node [below] at   (0.8,-0.1) {\tiny$ k_4$};
\filldraw [white] (-3.75,0) circle [radius=0.2];
\draw  (-3.75,0) circle [radius=0.2];
\filldraw  (-3.75,0) circle [radius=0.05];
\node [below] at   (-3.75,-0.1) {\tiny$ k_7$};
\filldraw [white] (-5.25,0) circle [radius=0.2];
\draw  (-5.25,0) circle [radius=0.2];
\filldraw  (-5.25,0) circle [radius=0.05];
\node [below] at   (-5.25,-0.1) {\tiny$ k_8$};
\filldraw [white] (-2.25,0) circle [radius=0.2];
\draw  (-2.25,0) circle [radius=0.2];
\filldraw  (-2.25,0) circle [radius=0.05];
\node [below] at   (-2.25,-0.1) {\tiny$ k_6$};
\filldraw [white] (-0.75,0) circle [radius=0.2];
\draw  (-0.75,0) circle [radius=0.2];
\filldraw  (-0.75,0) circle [radius=0.05];
\node [below] at   (-0.75,-0.1) {\tiny$k_5$};

\filldraw [white] (-1.88,3.25) circle [radius=0.2];
\draw  (-1.88,3.25) circle [radius=0.2];
\filldraw  (-1.88,3.25) circle [radius=0.05];
\node [above] at    (-2.05,3.35) {\tiny$\omega k_2$};
\filldraw [white]  (-2.63,4.55) circle [radius=0.2];
\draw   (-2.63,4.55) circle [radius=0.2];
\filldraw  (-2.63,4.55) circle [radius=0.05];
\node [above] at   (-2.8,4.6) {\tiny$\omega k_1$};
\filldraw [white]  (-1.13,1.95) circle [radius=0.2];
\draw   (-1.13,1.95) circle [radius=0.2];
\filldraw  (-1.13,1.95) circle [radius=0.05];
\node [above] at    (-1.35,2) {\tiny$\omega k_3$};
\filldraw [white]  (-0.38,0.65) circle [radius=0.2];
\draw   (-0.38,0.65) circle [radius=0.2];
\filldraw  (-0.38,0.65) circle [radius=0.05];
\node [above] at  (-0.6,0.7) {\tiny$\omega k_4$};
\filldraw [white]  (1.88,-3.25) circle [radius=0.2];
\draw   (1.88,-3.25) circle [radius=0.2];
\filldraw  (1.88,-3.25) circle [radius=0.05];
\node [below] at    (2.15,-3.25) {\tiny$\omega k_7$};
\filldraw [white]  (2.63,-4.55) circle [radius=0.2];
\draw   (2.63,-4.55) circle [radius=0.2];
\filldraw  (2.63,-4.55) circle [radius=0.05];
\node [below] at   (2.8,-4.55) {\tiny$\omega k_8$};
\filldraw [white] (1.13,-1.95) circle [radius=0.2];
\draw   (1.13,-1.95) circle [radius=0.2];
\filldraw  (1.13,-1.95) circle [radius=0.05];
\node [below] at    (1.35,-1.95) {\tiny$\omega k_6$};
\filldraw [white] (0.38,-0.65) circle [radius=0.2];
\draw   (0.38,-0.65) circle [radius=0.2];
\filldraw  (0.38,-0.65) circle [radius=0.05];
\node [below] at  (0.6,-0.65) {\tiny$\omega k_5$};

\filldraw [white] (-1.88,-3.25) circle [radius=0.2];
\draw  (-1.88,-3.25) circle [radius=0.2];
\filldraw  (-1.88,-3.25) circle [radius=0.05];
\node [below] at    (-2.1,-3.25) {\tiny$\omega^2 k_2$};
\filldraw [white] (-2.63,-4.55) circle [radius=0.2];
\draw  (-2.63,-4.55) circle [radius=0.2];
\filldraw  (-2.63,-4.55) circle [radius=0.05];
\node [below] at   (-2.8,-4.55) {\tiny$\omega^2 k_1$};
\filldraw [white] (-1.13,-1.95) circle [radius=0.2];
\draw  (-1.13,-1.95) circle [radius=0.2];
\filldraw  (-1.13,-1.95) circle [radius=0.05];
\node [below] at    (-1.35,-1.95) {\tiny$\omega^2 k_3$};
\filldraw [white] (-0.38,-0.65)  circle [radius=0.2];
\draw  (-0.38,-0.65)  circle [radius=0.2];
\filldraw  (-0.38,-0.65) circle [radius=0.05];
\node [below] at  (-0.6,-0.65) {\tiny$\omega^2 k_4$};
\filldraw [white] (1.88,3.25)  circle [radius=0.2];
\draw  (1.88,3.25)  circle [radius=0.2];
\filldraw  (1.88,3.25) circle [radius=0.05];
\node [above] at    (2.2,3.25) {\tiny$\omega^2 k_7$};
\filldraw [white] (2.63,4.55)  circle [radius=0.2];
\draw (2.63,4.55)  circle [radius=0.2];
\filldraw  (2.63,4.55) circle [radius=0.05];
\node [above] at   (2.9,4.55) {\tiny$\omega^2 k_8$};
\filldraw [white] (1.13,1.95)  circle [radius=0.2];
\draw (1.13,1.95)  circle [radius=0.2];
\filldraw  (1.13,1.95) circle [radius=0.05];
\node [above] at    (1.45,1.95) {\tiny$\omega^2 k_6$};
\filldraw [white] (0.38,0.65)  circle [radius=0.2];
\draw (0.38,0.65)  circle [radius=0.2];
\filldraw  (0.38,0.65) circle [radius=0.05];
\node [above] at  (0.7,0.65) {\tiny$\omega^2 k_5$};
		\end{tikzpicture}
		\label{case2}}
\caption{\footnotesize Figures (a) and (b) depict the jump contour $\Sigma^{E}$ for the case
   $\hat{\xi}\in\left[0,3\right)$ and $\hat{\xi}\in(-\frac{3}{8},0)$, respectively. }
\label{FIGE3}
\end{center}
\end{figure}
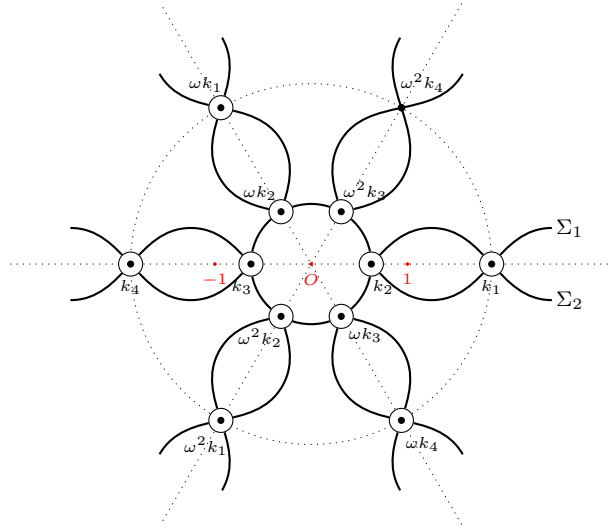
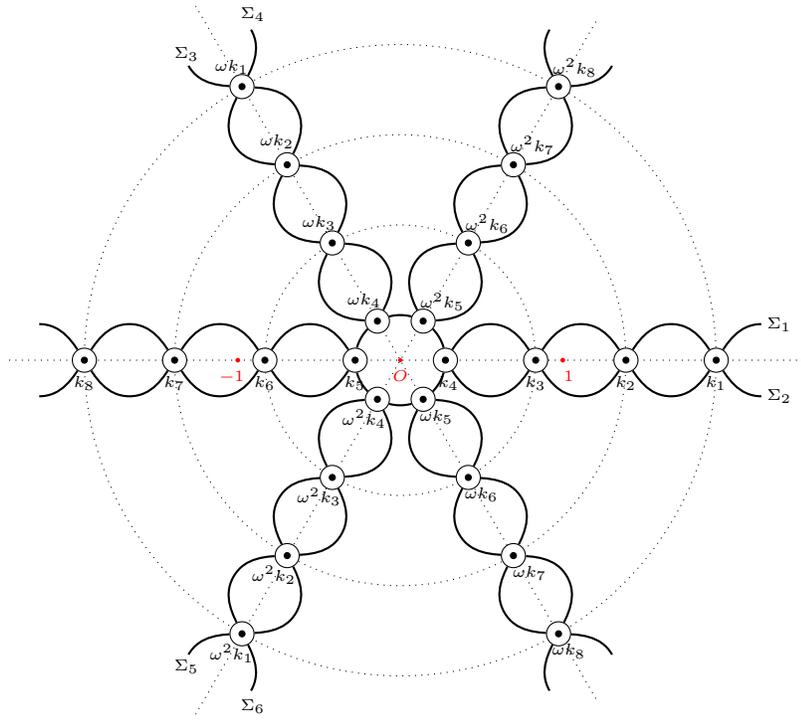

For $k\in\Sigma^{(ju)}\setminus U$,  referring to \eqref{estV2O} and observing the boundedness of $M^{r}$ is bounded, we find
\begin{equation}\label{estVE1}
\Vert V^{E}(k)-I\Vert_{L^q}= \mathcal{O}( e^{-K_qt}),  \quad t\to\infty;
\end{equation}
For $k\in \partial U$, as indicated by Proposition \ref{asyMlo}, we obtain
\begin{equation}
\lvert V^{E}(k)-I\rvert=\lvert M^{r}(k)(M^{lo}(k)-I)M^{r}(k)^{-1} \rvert = \mathcal{O}(t^{-\frac{1}{2}}), \quad t\to\infty.
\end{equation}
In a manner akin to the methodology employed in determining the value of $M^{err}(k)$ in subsection \ref{merr}, we derive the formulation for $E(k)$ and elucidate its associated asymptotic results.

\begin{proposition}
As $t\to\infty$, it follows that
\begin{equation}\label{asyE16pi}
E(e^{\frac{\pi}{6}\mathrm{i}})=I+t^{-\frac{1}{2}}H^{(0)}(e^{\frac{\pi}{6}\mathrm{i}})+\mathcal{O}(t^{-1+\rho}), \quad 0<\rho<1/4,
\end{equation}
where
\begin{equation}\label{H0}
H^{(0)}(e^{\frac{\pi}{6}\mathrm{i}})=-\frac{1}{2}\sum_{i=1}^{p(\hat{\xi})}M^{r}(k_i)\mathcal{F}_i(e^{\frac{\pi}{6}\mathrm{i}}) M^{r}(k_i)^{-1}.
\end{equation}
and $\mathcal{F}_i(k)$ is defined in \eqref{Fi}.
\end{proposition}

\section{Contribution from $\bar{\partial}$-components}\label{7}

By using $M^{R}(k)$ to reduce $m^{(2)}(k)$ to a pure $\bar\partial$-problem, which will be analyzed in this section. From \eqref{m3m2}, we have
\begin{equation}\label{defM3}
m^{(3)}(k)=m^{(2)}(k)M^{rhp}(k)^{-1},
\end{equation}
satisfying the following pure $\bar{\partial}$-problem.
\begin{Dbar}\label{dbarproblem}
 Find a row vector-valued function  $m^{(3)}(k):=m^{(3)}(k;y,t)$ such that
\begin{itemize}
\item $m^{(3)}(k)$ is continuous in $\mathbb{C}$.
\item Asymptotic behavior: $m^{(3)}(k)=(1, \ 1, \ 1)+\mathcal{O}(k^{-1}),\quad k \rightarrow \infty$.\label{asyM3}
\item $m^{(3)}(k)$ satisfies the $\bar\partial$-equation
\begin{equation}\label{Dbar1}
\bar{\partial}m^{(3)}(k)=m^{(3)}(k)W^{(3)}(k),\ \ k\in \mathbb{C}
\end{equation}
with
\begin{equation}\label{Dbar2}
W^{(3)}(k)=M^{R}(k)\bar{\partial}R^{(2)}(k)M^{R}(k)^{-1}.
\end{equation}
\end{itemize}
\end{Dbar}

The solution of $\bar{\partial}$-Problem \ref{dbarproblem} can be given by an integral equation
\begin{equation}\label{M3int}
m^{(3)}(k)=(1, \ 1, \ 1)+\frac{1}{\pi}\underset{\mathbb{C}}\iint\dfrac{m^{(3)}(\varsigma)W^{(3)} (\varsigma)}{\varsigma-k}dA(\varsigma),
\end{equation}
where $dA(\varsigma)$ is Lebesgue measure on the plane. Equation \eqref{M3int} can be equivalently reformulated as the operator equation
\begin{equation}
m^{(3)}(k)(I-\mathcal{C}_k)=(1, \ 1, \ 1),
\end{equation}
where $\mathcal{C}_k$ is left Cauchy-Green integral  operator,
\begin{equation}\label{m3os}
f\mathcal{C}_k(k)=\frac{1}{\pi}\underset{\mathbb{C}}\iint\dfrac{f(\varsigma)W^{(3)} (\varsigma)}{\varsigma-k}dA(\varsigma).
\end{equation}
Aiming at estimating $m^{(3)}(k)$, it’s suffcient to evaluate the norm of the integral operator $(I-\mathcal{C}_k)^{-1}$.

\subsection{In space-time regions $\xi\in(-\infty,-\frac{3}{8})\cup(3,+\infty)$}\label{7.1}

\begin{lemma}\label{estS}
The norm of the integral operator $\mathcal{C}_k$ admits  estimate
\begin{equation}
\Vert \mathcal{C}_k \Vert_{L^{\infty}\rightarrow L^{\infty}}=\mathcal{O}(t^{-\frac{1}{2}}), \quad t\rightarrow\infty.
\end{equation}
\end{lemma}
\begin{proof}
For any $f\in L^\infty$,
\begin{equation}\label{fck}
\Vert f\mathcal{C}_k \Vert_{\infty}\leq\Vert f \Vert_{\infty}\frac{1}{\pi}\underset{\mathbb{C}}\iint\dfrac{\lvert W^{(3)} (\varsigma)\rvert}{\lvert\varsigma-k\rvert}dA(\varsigma).
\end{equation}
Aiming at reaching our goal, it is suffcient to evaluate the integral in the R.H.S of \eqref{fck}. Noticing that $M^{R}(k)$, $M^{R^{-1}}(k)$ are bounded for $k\in\Omega$ as well as $W^{(3)}(k)\equiv0$ for $k\in\mathbb{C}\backslash\Omega$, we turn to consider the matrix-valued function $W^{(3)}(k)$ on $\Omega$. The details are given to control the
integral term of \eqref{fck} for $k\in\Omega_1$ in the case of $\hat{\xi}\in(3,+\infty)$.

Let $\varsigma=u+\mathrm{i}v$, \eqref{Dbar2} together with Lemma \ref{lR1} lead to
\begin{equation}
\begin{aligned}
\frac{1}{\pi}\underset{\mathbb{C}}\iint\dfrac{\lvert W^{(3)} (\varsigma)\rvert}{\lvert\varsigma-k\rvert}dA(\varsigma)&\lesssim\frac{1}{\pi}\underset{\Omega_1}\iint\frac{\left\lvert\bar\partial\mathcal{R}^{(2)}(\varsigma)\right\rvert  e^{-c(\hat{\xi})tv}}{\lvert\varsigma-k\rvert}dA(\varsigma)\\
& \lesssim I_1+I_2,
\end{aligned}
\end{equation}
where
\begin{equation}
 I_1=\int_{0}^{+\infty}\int_{v}^{\infty}\dfrac{\lvert\rho' (\varsigma)\rvert}{\lvert\varsigma-k\rvert}e^{-c(\hat{\xi})tv}dudv, \quad I_2=\int_{0}^{+\infty}\int_{v}^{\infty}\dfrac{\lvert\varsigma\rvert^{-\frac{1}{2}}}{\lvert\varsigma-k\rvert}e^{-c(\hat{\xi})tv}dudv,
\end{equation}
and $I_1$ satisfies the following estimate
\begin{equation}
\begin{aligned}
I_1&\leq\int_{0}^{+\infty}\left\Vert |\varsigma-k|^{-1}\right\Vert_{L^2(\mathbb{R}_+)}  \left\Vert \rho'\right\Vert_{L^2(\mathbb{R}_+)} e^{-c(\hat{\xi})tv}dv\\
&\lesssim \int_{0}^{+\infty}\lvert v-y\rvert^{-\frac{1}{2}} e^{-c(\hat{\xi})tv}dv\\
&\lesssim t^{-\frac{1}{2}}.
\end{aligned}
\end{equation}
For $I_2$, the H\"{o}lder inequality for $p>2$ and $\frac{1}{p}+\frac{1}{q}=1$ gives that
\begin{equation}
\begin{aligned}
I_2&\leq\int_{0}^{+\infty}\left\Vert\lvert\varsigma-k\rvert^{-1}\right\Vert_{L^q(\mathbb{R}^+)} \left\Vert\lvert k\rvert^{-\frac{1}{2}}\right\Vert_{L^p(\mathbb{R}^+)}e^{-c(\hat{\xi})tv}dv\\
&\lesssim\int_{0}^{+\infty}\lvert v-y\rvert^{\frac{1}{q}-1}v^{-\frac{1}{2}+\frac{1}{p}}e^{-c(\hat{\xi})tv}dv\\
&\lesssim t^{-\frac{1}{2}}.	
\end{aligned}
\end{equation}
Summarizing the estimations $I_1$ and $I_2$ we complete the proof.
\end{proof}
\begin{corollary}
As $t\rightarrow\infty$, $(I-\mathcal{C}_k)^{-1}$ exists, indicating the existence of a unique solution to the $\bar{\partial}$-Problem \ref{dbarproblem}.
\end{corollary}
To recover the solution of the Cauchy problem \eqref{DP}--\eqref{intva}, we need to evaluate $m^{(3)}(k)$ at $k=e^{\frac{\pi}{6}\mathrm{i}}$. Take $k=e^{\frac{\pi}{6}\mathrm{i}}$ in \eqref{M3int}, then
\begin{equation}\label{intm3}
m^{(3)}(e^{\frac{\pi}{6}\mathrm{i}})=(1, \ 1, \ 1)+\frac{1}{\pi}\underset{\mathbb{C}}\iint\dfrac{m^{(3)}(\varsigma)W^{(3)} (\varsigma)}{\varsigma-e^{\frac{\pi}{6}\mathrm{i}}}dA(\varsigma).
\end{equation}
The following proposition exhibits the asymptotics for $m^{(3)}(e^{\frac{\pi}{6}\mathrm{i}})$ as $t\rightarrow+\infty$.
\begin{proposition}\label{asyM3i}
There exists a constant $T_1>0$ such that for all $t>T_1$, the solution $m^{(3)}(k)$ of the pure $\bar{\partial}$-Problem \ref{dbarproblem} admits the following estimation
\begin{equation}\label{m3i}
 m^{(3)}(e^{\frac{\pi}{6}\mathrm{i}})=(1, \ 1, \ 1)+\mathcal{O}(t^{-1+\rho}),
\end{equation}
where $0<\rho<  {1}/{4} $.
\end{proposition}
\begin{proof}
The proof of \eqref{m3i} is exhibited for $k\in\Omega_1$ in the case of $\hat{\xi}\in(3,+\infty)$. Let $\varsigma=u+\mathrm{i}v$, \eqref{intm3} together with Lemma \ref{lR1} give that
\begin{equation*}
\frac{1}{\pi}\underset{\Omega_{01}}\iint\dfrac{\left\lvert W^{(3)} (\varsigma)\right\rvert}{\left\lvert \varsigma-e^{\frac{\pi}{6}\mathrm{i}}\right\rvert}dA(\varsigma)\lesssim \underset{\Omega_{01}}\iint\dfrac{|\bar{\partial}R_1(\varsigma) e^{\mathrm{i}t\theta_{12}}|}{|\varsigma-e^{\frac{\pi}{6}\mathrm{i}}|}dA(\varsigma)\lesssim I_3+I_4
\end{equation*}
with
\begin{equation}
I_3=\underset{\Omega_{01}}\iint\dfrac{ \lvert \rho' (\varsigma)\rvert^{-c(\hat{\xi})t}}{\lvert\varsigma-e^{\frac{\pi}{6}\mathrm{i}}\rvert}dA(\varsigma), \quad
I_4=\underset{\Omega_{01}}\iint\dfrac{\lvert \varsigma \rvert^{-\frac{1}{2}}e^{-c(\hat{\xi})vt}}{\lvert\varsigma-e^{\frac{\pi}{6}\mathrm{i}}\rvert}dA(\varsigma).	
\end{equation}	
Noticing that $|\varsigma-e^{\frac{\pi}{6}\mathrm{i}}|$ is bounded for $\varsigma\in\Omega_{1}$, we can deduce
\begin{align}
I_3\lesssim \int_{0}^{+\infty} \Vert \rho'\Vert_{L^1(\mathbb{R}^+)} e^{-c(\hat{\xi})tv} dv\nonumber\lesssim \int_{0}^{+\infty}  e^{-c(\hat{\xi})tv} dv\lesssim t^{-1}.
\end{align}
Since ${\rm Im} e^{\frac{\pi}{6}\mathrm{i}}=\frac{1}{2}$, we can divide the integral $I_4$ into two parts
\begin{align*}
I_4&\leq\int_{0}^{\frac{1}{2}}\int_{\frac{v}{\tan\varphi}}^{+\infty}\dfrac{\lvert\varsigma\rvert^{-\frac{1}{2}}e^{-c(\hat{\xi})tv}}{\left\lvert \varsigma-e^{\frac{i\pi}{6}}\right\rvert}dudv+
\int_{\frac{1}{2}}^{+\infty}\int_{\frac{v}{\tan\varphi}}^{+\infty}\dfrac{\lvert\varsigma\rvert^{-\frac{1}{2}}e^{-c(\hat{\xi})tv}}{\left\lvert\varsigma-e^{\frac{i\pi}{6}}\right\rvert}dudv\\
&= I_{41}+I_{42}.
\end{align*}
Observing that $|\varsigma|<|\varsigma-e^{\frac{i\pi}{6}}|$ for $0<v<\frac{1}{4}$ while $|\varsigma-e^{\frac{i\pi}{6}}|\lesssim|\varsigma|$ for $v>\frac{1}{4}$, then
\begin{align*}
I_{41}&=\int_{0}^{\frac{1}{4}}\int_{\frac{v}{\tan\varphi}}^{+\infty}\dfrac{|\varsigma|^{-\frac{1}{2}}e^{-c(\hat{\xi})tv}}{|\varsigma-e^{\frac{i\pi}{6}}|}dudv
+\int_{\frac{1}{4}}^{\frac{1}{2}}\int_{\frac{v}{\tan\varphi}}^{+\infty}\dfrac{|\varsigma|^{-\frac{1}{2}}e^{-c(\hat{\xi})tv}}{|\varsigma-e^{\frac{i\pi}{6}}|}dudv\\
&=I_{41}^{(1)}+I_{41}^{(2)},
\end{align*}
where
\begin{align*}
I_{41}^{(1)}\leq&\int_{0}^{\frac{1}{4}}\int_{\frac{v}{\tan\varphi}}^{+\infty} (u^2+v^2)^{-\frac{1}{4}-\frac{\rho}{2}}\left[( u- {\sqrt{3}}/{2})^2+(v-{1}/{2})^2\right]^{-\frac{1}{2}+\frac{\rho}{2}}due^{-c(\hat{\xi})tv}dv\\
\leq &\int_{0}^{\frac{1}{4}}\left[ \int_{v}^{+\infty}\left(1+\left( {u}/{v} \right)^2  \right) ^{-\frac{1}{4}-\frac{\rho}{2}}v^{-\rho} d {u}/{v}\right]  (v- {1}/{2})^{-1+\rho}e^{-c(\hat{\xi})tv}dv\\
\lesssim  &\int_{0}^{\frac{1}{4}}v^{-\rho}e^{-c(\hat{\xi})tv}dv\lesssim t^{-1+\rho},
\end{align*}
and
\begin{equation*}
I_{41}^{(2)}\leq\int_{\frac{1}{4}}^{\frac{1}{2}}\int_{\frac{v}{\tan\varphi}}^{+\infty}\left\Vert \left\lvert \varsigma-e^{\frac{\pi}{6}\mathrm{i}}\right\rvert^{-\frac{3}{2}}\right\Vert_{L_u^1(\frac{v}{\tan\varphi},+\infty)}e^{-c(\hat{\xi})tv}dv\lesssim t^{-1}.
\end{equation*}
In a analogous analysis of $I_{41}$, we are able to deduce that $I_{42}\lesssim t^{-1}$. Thus we arrive at the consequence.
\end{proof}

\subsection{In space-time regions $\xi\in(-\frac{3}{8},3)$}

The analysis of $\bar{\partial}$-problem $m^{(3)}(k)$ for this case is the analogue of Proposition \ref{asyM3i} Subsection \ref{7.1}, we directly present the following proposition without proof.
\begin{proposition}\label{asyM3i1}
There exists a constant $T_1>0$ such that for all $t>T_1$, the solution $m^{(3)}(k)$ of the pure $\bar{\partial}$-problem \ref{dbarproblem} exists unique and admits the following estimation
\begin{equation}\label{M3i1}
 m^{(3)}(e^{\frac{\pi}{6}\mathrm{i}})=(1, \ 1, \ 1)+\mathcal{O}(t^{-\frac{3}{4}}).
\end{equation}
\end{proposition}

\section{Proof of Theorem \ref{th1}}\label{8}
Now we begin to construct the long time asymptotic behavior for the solution of the Cauchy problem \eqref{DP}--\eqref{intva}.
\begin{itemize}
\item  {\rm\bf For the solitonic  region I:} $\xi\in(-\infty,-\frac{3}{8})\cup(3,+\infty)$. Recalling the series of transformations, we derive that
\begin{equation}\label{recall}
m(k)=m^{(3)}(k)M^{R}(k)\mathcal{R}^{(2)}(k)^{-1}T(k)^{-1}G(k)^{-1}.
\end{equation}
To reconstruct the solution $u(x,t)$ according to \eqref{rescon}, we take $k=e^{\frac{\pi}{6}\mathrm{i}}$. In this case, $\mathcal{R}^{(2)}(e^{\frac{\pi}{6}\mathrm{i}})=G(e^{\frac{\pi}{6}\mathrm{i}})=I$. By further employing Propositions \ref{asyM3i}, we arrive at
\begin{equation}\label{Mepi}
m(e^{\frac{\pi}{6}\mathrm{i}})=m^{R}(e^{\frac{\pi}{6}\mathrm{i}})T(e^{\frac{\pi}{6}\mathrm{i}})^{-1}+\mathcal{O}(t^{-1+\rho}), \quad t \to \infty.
\end{equation}
The substitution of \eqref{Mepi} into the reconstruction formula \eqref{rescon} yields
\begin{align}
u(y,t)&=\frac{\partial}{\partial t}\log\frac{m_2^{R}(e^{\frac{\pi}{6}\mathrm{i}};\hat{\xi},t)}{m_1^{R}(e^{\frac{\pi}{6}\mathrm{i}};\hat{\xi},t)}+\mathcal{O}(t^{-1+\rho}) \nonumber\\
&=u^{sol,N}(y,t)+\mathcal{O}(t^{-1+\rho}),\label{usol11}\\
x(y,t)&=y+\log\frac{m_2^{R}(e^{\frac{\pi}{6}\mathrm{i}};\hat{\xi},t)}{m_1^{R}(e^{\frac{\pi}{6}\mathrm{i}};\hat{\xi},t)}+\log\frac{T_1(e^{\frac{\pi}{6}\mathrm{i}})}{T_{2}(e^{\frac{\pi}{6}\mathrm{i}})}
+\mathcal{O}(t^{-1+\rho})\nonumber\\
&=x^{sol,N}(y,t)+\log{T_{12}(e^{\frac{\pi}{6}\mathrm{i}})}+\mathcal{O}(t^{-1+\rho}),\label{usol12}
\end{align}
where $u^{sol,N}(y,t)$ and $x^{sol,N}(y,t)$ are defined in \eqref{ressol}. It is worth noting that the above asymptotic formulas retain their form under differentiation with respect to time $t$, without perturbing the error term. For a detailed demonstration of this property, one may consult Theorem 5.1 in \cite{AJD}. Additionally, taking into account the boundedness of $\log{T_{12}(e^{\frac{\pi}{6}\mathrm{i}})}$ in \eqref{usol12}, it is thereby inferred that
\begin{equation}\label{x/ty/t}
\frac{x}{t}=\frac{y}{t}+\mathcal{O}(t^{-1}), \ \text{i.e.} \ \xi=\hat{\xi}+\mathcal{O}(t^{-1}).
\end{equation}
Noticing $m(e^{\frac{\pi}{6}\mathrm{i}}):=m(e^{\frac{\pi}{6}\mathrm{i}};\hat{\xi},t)$, $m^{R}(e^{\frac{\pi}{6}\mathrm{i}}):=m^{R}(e^{\frac{\pi}{6}\mathrm{i}};\hat{\xi},t)$, replacing $\hat{\xi}$ by $\xi$ in \eqref{Mepi} and using the reconstruction formula \eqref{rescon} leads to
\begin{equation}\label{uxt1}
 u(x,t)=u^{sol,N}(x,t)+\mathcal{O} (t^{-1+\rho}),
\end{equation}
thus establishing \eqref{uusol} as delineated in Theorem \ref{th1}. \eqref{uxt1} embodies the concept of soliton resolution in the following manner: the function $u^{sol,N}(x,t)$ is generically asymptotic to a superposition of one-soliton solutions\cite{Liu3,CL}. The discrete spectrum $\zeta_n, \ n=1,2,\dots,N$ can be arranged as
\begin{equation}\label{ordzeta}
{\rm Re} \ \zeta_1>{\rm Re} \ \zeta_2>\cdots>{\rm Re} \ \zeta_{N},
\end{equation}
repeatedly utilizing Lemma \ref{Mout} on the set $\mathcal{D}=\left\lbrace  r(k),\left\lbrace \zeta_n,C_n\right\rbrace_{n=1}^{6N}\right\rbrace$, each of which contains a single soliton. It is discerned that the solution of the Cauchy problem \eqref{DP}--\eqref{intva} fulfills
\begin{equation}
u(x,t)=\sum_{n=1}^{N} \mathcal{U}^{sol}(\zeta_n;x,t)+\mathcal{O}(t^{-1+\rho}),
\end{equation}
where $\mathcal{U}^{sol}(\zeta_n;x,t)$ is defined in \eqref{reulam}--\eqref{1U}.
\item  {\rm\bf   For Zakharov-Manakov    regions  II-1: $\xi\in(-\frac{3}{8},0)$ and II-2:} $\xi\in\left[0,3\right)$. Recalling the sequence of transformations outside $U(\hat{\xi})$, we derive
\begin{equation}\label{recallLO}
m(k)=m^{(3)}(k)E(k)M^{r}(k)\mathcal{R}^{(2)}(k)^{-1}T(k)^{-1}G(k)^{-1}.
\end{equation}
Consider \eqref{recallLO} at $k=e^{\frac{\pi}{6}\mathrm{i}}$ as $t\to \infty$, it becomes that
\begin{equation}
m(e^{\frac{\pi}{6}\mathrm{i}};\hat{\xi},t)=(1, \ 1, \ 1) \left(I+t^{-\frac{1}{2}}H^{(0)}(e^{\frac{\pi}{6}\mathrm{i}})\right)T(e^{\frac{\pi}{6}\mathrm{i}})^{-1}+\mathcal{O}(t^{-3/4}).
\end{equation}
Replacing $\hat{\xi}$ by $\xi$ in \eqref{recallLO} and employing the reconstruction formula \eqref{rescon}, we obtain
\begin{align}
u(x,t)=t^{-\frac{1}{2}}f_1(x,t,e^{\frac{\pi}{6}\mathrm{i}})+\mathcal{O}(t^{-3/4}),
\end{align}
where
\begin{align*}
f_1(x,t,e^{\frac{\pi}{6}\mathrm{i}})=\sum_{j=1}^{3}\frac{\partial}{\partial t}\left(\left(H^{(0)}(e^{\frac{\pi}{6}\mathrm{i}})\right)_{j2}-\left(H^{(0)}(e^{\frac{\pi}{6}\mathrm{i}})\right)_{j1}\right),
\end{align*}
$H^{(0)}(e^{\frac{\pi}{6}\mathrm{i}})$ is defined by \eqref{H0}, $\left(H^{(0)}(e^{\frac{\pi}{6}\mathrm{i}})\right)_{jk}$ represents the element in the $j$-th row and $k$-th column.
\end{itemize}
To sum up, Theorem \ref{th1} can be proved.
\begin{remark}\label{r8}
Due to the relationship between $\xi$ and $\hat{\xi}$ in \eqref{x/ty/t}, the regions about $\xi$ can be considered as approaching asymptotic equivalence with the regions about $\hat{\xi}$ in Figure \ref{figtheta}.
\end{remark}

    \noindent\textbf{Acknowledgements}

    This work is supported by  the National Natural Science
    Foundation of China (Grant No. 12271104, 51879045).\vspace{2mm}

    \noindent\textbf{Data Availability Statements}

    The data which supports the findings of this study is available within the article.\vspace{2mm}

    \noindent{\bf Conflict of Interest}

    The authors have no conflicts to disclose.


\begin{thebibliography}{99}


\bibitem{AM} A. Degasperis, M. Procesi,
\newblock{Asymptotic integrability},
\textit{Symmetry and Perturbation Theory,  Ed. A. Degasperis and G. Gaeta, World Scientific, Singapore},  1999, 23-37.


\bibitem{RS}  R. S. Johnson,
\newblock{Camassa-Holm, Korteweg-de Vries and related models for water waves},
\textit{ J. Fluid Mech.}, 455(2002), 63-82


\bibitem{RI1} R. I. Ivanov,
\newblock{Water waves and integrability},
\textit{ Phil. Trans. R. Soc. Lond. A}, 365(2007), 2267-2280.


\bibitem{AD} A. Constantin, D. Lannes,
\newblock{The hydrodynamical relevance of the Camassa-Holm and Degasperis-Procesi equations},
\textit{ Arch. Ration. Mech. Anal.}, 192(2009), 165-186.

\bibitem{BD} B. Alvarez-Samaniego, D. Lannes,
\newblock{Large time existence for 3D water-waves and asymptotics},
\textit{ Invent. Math.}, 171(2008), 485-541.


\bibitem{RI2} R. I. Ivanov,
\newblock{On the integrability of a class of nonlinear dispersive wave equations},
\textit{ J. Nonlinear Math. Phys.}, 12(2005), 462-468.


\bibitem{GM} G. M. Coclite, K. H. Karlsen, N. H. Risebro,
\newblock {Numerical schemes for computing discontinuous solutions of the Degasperis-Procesi equation},
\textit{IMA J. Numer. Anal.}, 28(2008), 80-105.


\bibitem{HH} H. H. Dai,
\newblock{Model equations for nonlinear dispersive waves in a compressible Mooney-Rivlin rod},
\textit{Acta Mech.}, 127(1998), 193-207.


\bibitem{H} H. Lundmark,
\newblock{Formation and dynamics of shock waves in the Degasperis-Procesi equation},
\textit{J. Nonl. Sci.}, 17 (3) (2007), 169-198.


\bibitem{RD} R. Camassa, D. Holm,
\newblock{An integrable shallow water equation with peaked solitons},
\textit{ Phys. Rev. Lett.}, 71(1993), 1661-1664.


\bibitem{ADD} A. Degasperis, D. D. Holm, A. N. W. Hone,
\newblock{A new integral equation with peakon solutions},
\textit{ Theor. Math. Phys.}, 133(2002), 1463-1474.


\bibitem{AB} A. Constantin, B. Kolev,
\newblock{Geodesic flow on the diffeomorphism group of the circle},
\textit{Comment. Math. Helv.}, 78(2003), 787-804.


\bibitem{G} G. Misiolek,
\newblock{A shallow water equation as a geodesic flow on the Bott-Virasoro group},
\textit{J. Geom. Phys.}, 24(1998), 203-208.







\bibitem{J1} J. Lenells,
\newblock{Traveling wave solutions of the Degasperis-Procesi equation},
\textit{J. Math. Anal. Appl.}, 306(2005), 72-82.


\bibitem{Y1}  Y. Matsuno,
\newblock{The $N$-soliton solution of the Degasperis-Procesi equation},
\textit{Inverse Problems}, 21(2005), 2085-2101.




\bibitem{YZ} Y. Liu, Z. Yin,
\newblock{Global existence and blow-up phenomena for the Degasperis-Procesi equation},
\textit{Commun. Math. Phys.}, 267(2006), 801-820.



\bibitem{ARIJ} A. Constantin, R. I. Ivanov, J. Lenells,
\newblock{Inverse scattering transform for the Degasperis-Procesi equation},
\textit{Nonlinearity}, 23(2010), 2559-2575.





\bibitem{J2} J. Lenells,
\newblock{The Degasperis-Procesi equation on the half-line},
\textit{Nonlinear Anal.}, 76 (2013), 122-139.



\bibitem{AJD} A. Boutet de Monvel, J. Lenells, D. Shepelsky,
\newblock{Long-time asymptotics for the Degasperis-Procesi equation on the half-line},
\textit{Ann. Inst. Fourier}, 69(2019), 171-230.

\bibitem{HZF}
  Y. Hou,  P. Zhao, E. G. Fan, Z. J. Qiao,
  \newblock{ Algebro-geometric solutions for Degasperis-Procesi
hierarchy,}
\textit{ SIAM J. Math. Anal.}, 45(2013), 1216-1266.

\bibitem{FGP}
  R. Feola, F. Giuliani, M. Procesi,
  \newblock{ Reducible KAM tori for the Degasperis-Procesi equation,}
\textit{ Commun. Math. Phys.},  377 (2020), 1681–1759.

\bibitem{AD2} A. Boutet de Monvel,  D.  Shepelsky,
\newblock{A Riemann-Hilbert approach for the Degasperis-Procesi equation},
\textit{Nonlinearity}, 26(2013), 2081-2107.

\bibitem{RN6}
P. Deift, X. Zhou,
\newblock{A steepest descent method for oscillatory Riemann-Hilbert problems},
\textit{Ann. Math.}, 137(1993),  295-368.

\bibitem{RN9}
P. Deift, X. Zhou,
\newblock{Long-time behavior of the non-focusing nonlinear Schr\"odinger equation-a case study},
\textit{Lectures in Mathematical Sciences}, Graduate School of Mathematical Sciences, University of Tokyo, 1994.





\bibitem{KTRPD1} K. T. R. McLaughlin, P. D. Miller,
\newblock{The $\bar{\partial}$-steepest descent method and the asymptotic behavior of polynomials orthogonal on the unit circle with fixed and exponentially varying non-analytic weights},
\textit{Int. Math. Res. Not.}, (2006), Art. ID 48673.


\bibitem{KTRPD2} K. T. R. McLaughlin, P. D. Miller,
\newblock{The $\bar{\partial}$-steepest descent method for orthogonal polynomials on the real line with varying weights},
\textit{Int. Math. Res. Not.}, (2008), Art. ID 075.


\bibitem{DM} M. Dieng, K. T. R. McLaughlin,
\newblock{Dispersive asymptotics for linear and integrable equations by the Dbar steepest descent method,
Nonlinear dispersive partial differential equations and inverse scattering},
\textit{Fields Inst. Commun.}, Springer, New York, 2019, 253-291.

\bibitem{CJ} S. Cuccagna, R. Jenkins,
\newblock{On asymptotic stability $N$-solitons of the defocusing nonlinear Schr\"odinger equation}
\textit{Commun. Math. Phys.}, 343(2016), 921-969.


\bibitem{BJ} M. Borghese, R. Jenkins, K. T. R. McLaughlin,
\newblock{Long-time asymptotic behavior of the focusing nonlinear Schr\"odinger equation},
\textit{Ann. Inst. H. Poincar\'{e} C Anal. Non Lin\'{e}aire}, 35(2018), 887-920.


\bibitem{Liu3}
R. Jenkins, J. Liu, P. Perry, C. Sulem,
\newblock{Soliton resolution for the derivative nonlinear Schr\"odinger equation},
\textit{Commun. Math. Phys.},  363(2018), 1003-1049.


\bibitem{CL}
G. Chen, J. Q. Liu,
\newblock{Soliton resolution for the focusing modified KdV equation},
\textit{Ann. Inst. H. Poincar\'{e} C Anal. Non Lin\'{e}aire}, 38(2021), 2005-2071.

\bibitem{YF} Y. L. Yang, E. G. Fan,
\newblock{On the long-time asymptotics of the modified Camassa-Holm equation in space-time solitonic regions},
\textit{Adv. Math.}, 402(2022), 108340.

\bibitem{YF2} Y. L. Yang, E. G. Fan,
\newblock{Soliton resolution and large time behavior of solutions to the Cauchy problem for the Novikov
equation with a nonzero background},
\textit{Adv. Math.}, 426(2023), 109088.


\bibitem{CF} Q. Y. Cheng, E. G. Fan,
\newblock{Long-time asymptotics for the focusing Fokas-Lenells equation in the solitonic region of space-time},
\textit{ J. Differential Equations}, 309(2022), 883-948.

\bibitem{ZF} X. Zhou, E. G. Fan,
\newblock{Long time asymptotics for the nonlocal mKdV equation with finite density initial data},
\textit{Phys. D}, 440 (2022), 133458.


\bibitem{BI} P. M.   Bleher,  A.  Its,  Double scaling limit in the random matrix model: the Riemann-Hilbert
approach. Commun. Pure Appl. Math. 56 (2003),  433-516.

\bibitem{WF}   Z. Y. Wang,  E. G. Fan,  The defocusing nonlinear Schr\"odinger equation with a nonzero background: Painlev\'e asymptotics in two transition regions,
 {Commun. Math. Phys.},  {  402} (2023), 2879-2930.





\bibitem{RN10}
P. Deift, X. Zhou,
\newblock{Long-time asymptotics for solutions of the NLS equation with initial data in a weighted Sobolev space},
\textit{Commun. Pure Appl. Math.}, 56(2002), 1029-1077.






\bibitem{BC1984}
R.  Beals, R. R.   Coifman,
\newblock {Scattering and inverse scattering for first-order systems}.
\textit{Commun. Pure Appl. Math.}, 37(1984), 39-90.









    \end{thebibliography}
\end{document}